\documentclass[11pt,a4paper]{article}
\usepackage{amssymb,amsmath, amsfonts}
\usepackage{graphicx,graphics}
\usepackage{mathtools}
\usepackage{bbold}
\usepackage{comment}
\usepackage[english]{babel}
\usepackage[utf8]{inputenc}
\usepackage{epsfig,url} 
\usepackage{bbm,theorem}  
\usepackage{a4wide}
\usepackage{color}
\usepackage{enumerate} 
\usepackage{calrsfs} 
\usepackage[bookmarks=true]{hyperref}
\usepackage{bookmark} 
\usepackage{amsmath}
\usepackage{amsfonts}
\usepackage{amssymb}
\usepackage{graphicx}

\providecommand{\U}[1]{\protect\rule{.1in}{.1in}}
\DeclareMathAlphabet{\pazocal}{OMS}{zplm}{m}{n}
\newtheorem{theorem}{Theorem}

\newtheorem{definition}[theorem]{Definition}

\newtheorem{lemma}[theorem]{Lemma}

\newtheorem{proposition}[theorem]{Proposition}
\newtheorem{remark}[theorem]{Remark}

\newenvironment{proof}[1][Proof]{\noindent\textbf{#1.} }{\ \rule{0.5em}{0.5em}}
\newenvironment{proofof}[1][Proof]{\noindent \textbf{#1.} } {\ \rule{0.5em}{0.5em}}
\numberwithin{equation}{section}
\numberwithin{theorem}{section}

\newcommand{\E}{{\mathbb E}}

\newcommand{\R}{{\mathbb R}}

\newcommand{\cf}{{\mathbbm 1}}

\newcommand{\ep}{\epsilon}

\newcommand{\ud}{\;\mathrm{d}}

\newcommand{\eps}{{\epsilon}}

\newcommand{\beq}{\begin{equation}}
\newcommand{\eeq}{\end{equation}}
\newcommand{\beqs}{\begin{eqnarray}}
\newcommand{\eeqs}{\end{eqnarray}}

\newcounter{jlisti}

\newcommand{\func}[1]{\operatorname{#1}}

\def\be{\begin{equation}}
\def\ee{\end{equation}}
\def\bea{\begin{eqnarray}}
\def\eea{\end{eqnarray}}

\begin{document}

\title{ Interacting particle systems
with long range interactions:  approximation by tagged particles in random fields}

\author{Alessia Nota \thanks{\emailalessia} , Juan J. L. Vel\'azquez \thanks{\emailjuan}, Raphael Winter \thanks{\emailraphael} \\[1em]
$\,^*$\UAaddress\\[0.5em] $\, ^\dag$\UBaddress \\[0.5em]
 $\,^\ddag$\ULaddress}

\date{\today}

\newcommand{\email}[1]{E-mail: \tt #1}
\newcommand{\emailalessia}{\email{alessia.nota@univaq.it} (corresponding author)}
\newcommand{\emailjuan}{\email{velazquez@iam.uni-bonn.de}}
\newcommand{\emailraphael}{\email{raphael.winter@ens-lyon.fr}}

\newcommand{\UAaddress}{\em Dipartimento di Ingegneria e Scienze dell'Informazione e Matematica,\\ \em Universit\`a degli studi dell'Aquila, L'Aquila, 67100 Italy}
\newcommand{\UBaddress}{\em University of Bonn, Institute for Applied Mathematics\\
\em Endenicher Allee 60, D-53115 Bonn, Germany}
\newcommand{\ULaddress}{\em Universit\'e de Lyon, \\ \em 43 Boulevard du 11 Novembre 1918, 69100 Villeurbanne, France}

\date{\today }
\maketitle

\begin{abstract}

In this paper we continue the study of the derivation of different types of kinetic equations which arise from
scaling limits of interacting particle systems. We began this study in \cite{NVW}.  
More precisely, we consider the derivation of the kinetic equations for systems with long range interaction. Particular emphasis is
put on the fact that all the kinetic regimes can be obtained approximating 
the dynamics of interacting particle systems, as well as the dynamics of Rayleigh Gases, by a stochastic Langevin-type dynamics for a single particle.  We will present this approximation in detail and we will obtain precise formulas for the diffusion and friction coefficients appearing in the limit Fokker-Planck equation for the probability density of the tagged particle $f\left( x,v,t\right)$, for three different classes of potentials. The case of interaction potentials behaving as Coulombian
potentials at large distances will be considered in detail. In particular, we will discuss the onset of the the so-called Coulombian
logarithm. 
\end{abstract}

\tableofcontents

\bigskip\ \noindent

\vspace{0.5cm}

\bigskip

\section{Introduction} 
The goal of this paper and its companion paper \cite{NVW} is the kinetic description of interacting particle systems in the mathematical framework of scaling limits. The main emphasis is on the physically important case of long-range interaction, including the Coulomb case. 
The currently available mathematical tools do not allow for a fully rigorous derivation of the corresponding kinetic equations from a scaling limit of particle systems.  Indeed, even in the case of short-range interaction such a result has not been achieved for kinetic equations different from the Boltzmann equation, although only for short times. Instead we present a formalism in which the equations of motion for an interacting particle system, i.e. 
\begin{equation}
\frac{dX_{j}}{d\tau}=V_{j}\ \ \ ,\ \ \frac{dV_{j}}{d\tau}=-\sum_{k}\nabla
\Phi_{\eps}\left( X_{j}-X_{k}\right) \ \ ,\ \ j\in S  \label{eq:IntNew},
\end{equation}
or for a Rayleigh-gas system (see for instance \cite{S1, S2} and also \cite{NWL19})
\begin{equation} \label{eq:RayNew}
\begin{aligned}
\frac{dX}{d\tau }& =V\ \ ,\ \ \frac{dV}{d\tau }=-\sum_{j\in S}\nabla \Phi
_{\eps }\left( X-Y_{j}\right),  \\
\frac{dY_{k}}{d\tau }& =W_{k}\ \ ,\ \ \frac{dW_{k}}{d\tau }=-\nabla \Phi
_{\eps }\left( Y_{k}-X\right) \ \ ,\ \ k\in S , 
\end{aligned}
\end{equation}
are approximated by a stochastic system of Langevin type which describes the evolution of the tagged particle $(X,V)$. We will assume that the initial positions and velocities $\left\{ \left(
X_{j},V_{j}\right) :j\in S\right\} $ in~\eqref{eq:IntNew} are chosen according to
some probability distribution which is spatially homogeneous and with a
distribution of velocities $g=g\left( v\right) .$ 
Then this system is of the form
\begin{equation}
\frac{dX}{d\tau}=V_0\ \ ,\ \ \frac{dV}{d\tau}=-\Lambda_{\eps}\left(
V;g\right) +F_{\eps}\left( X,\tau;\omega;g\right)  \label{eq:TagFric},
\end{equation} 
 where $F_{\eps}$ is a random force field defined for $\omega$ in a
suitable probability space $\Omega$ and $\Lambda_{\eps}\left(
V;g\right) $ is a function which can be thought of as a
friction term depending only on the particle velocity $V.$ In Section~\ref{KinRanForcFields}, we recall some of the notations and properties of random force fields. 

In this paper we will restrict to classes of long-range potentials for which the evolution is driven by weak deflections of particles. In Section \ref{sec:TimeScale} this is formulated precisely, in terms of the Landau- and Boltzmann-Grad timescales $T_{L}$ and $T_{BG}$, as  $T_{L}\ll T_{BG}$.

The main content of the paper is a rigorous derivation of the friction- and diffusion coefficients resulting from \eqref{eq:TagFric} on a suitable timescale. The results of this analysis will be presented in Section~\ref{KineLimit}, while the proofs are postponed to Section~\ref{Sec:Pfs}.

A formal validation of the approximation \eqref{eq:TagFric} is the content of  Section~\ref{ss:JustR}.

Due to the long-range nature of the interaction potentials considered here, collective effects and screening play a crucial role. In the case of Coulomb interaction, the contribution of interactions over distances between the interparticle distance and the Debye screening length yield the dominant contribution, by a logarithmic factor in the scaling parameter $\eps \rightarrow 0$. The onset of this so-called Coulomb-logarithm is one of the main difficulties which present in this paper. The results presented here clarify the connection between the mathematical kinetic theory in the framework of scaling limits and the physical theory of plasmas. In the latter case, the higher order corrections of size $|\log \eps|^{-1}$ might not be small enough to be neglected, while they disappear in the scaling limit $\eps \rightarrow 0$. In the companion paper \cite{NVW}, we give an overview of scaling limits and their resulting kinetic equations for a large class of long-range potentials.

\section{Mathematical framework} 
\label{KinRanForcFields}

\subsection{Generalities about random force fields}

As indicated in the Introduction one of our goals is to approximate the
dynamics of (\ref{eq:IntNew}), through suitable scaling limits, by means of the dynamics of
a tagged particle in some families of random force fields. In this Section,
we describe some general properties of random force fields as well as some
specific random force fields generated by sets of moving particles. Other
random force fields will not be directly related to fields generated by
point particles, but they will be useful in order to approximate the
dynamics (\ref{eq:IntNew}) by (\ref{eq:TagFric}). 

Since we want to consider random force fields having singularities at the
particle centers, we will introduce some notation to deal with this case. We
will denote as $\Lambda $ the space of particle configurations in $\mathbb{R}%
^{3}$. More precisely, the elements of $\Lambda $ are locally finite  
subsets of $\R^3,$ or more precisely, sequences with the form $%
\left\{ x_{k}\right\} _{k\in \mathbb{N}}$ such that $\#\left[ \left\{
x_{k}\right\} _{k\in \mathbb{N}}\cap B_{R}\left( 0\right) \right] <\infty $
for any $R<\infty .$ Notice that we do not need to assume that $x_{k}\neq
x_{j}$ for $k\neq j.$ In order to allow force fields which diverge at some
points we define $\mathbb{R}_{\ast }^{3}$ as the compactification of $%
\R^3$ using a single infinity point $\infty .$ For a function $%
F\in C\left( \R^3;\mathbb{R}_{\ast }^{3}\right) $ we write $%
F\left( x_{0}\right) =\infty $ for some $x_{0}\in \R^3$, if the
function satisfies $\lim_{x\rightarrow x_{0}}\left\vert F\left( x\right)
\right\vert =\infty .$

We then introduce the following notation:%
\begin{equation*}
C_{\ast }\left( \R^3\right) =C\left( \R^3;\mathbb{R}%
_{\ast }^{3}\right).
\end{equation*}

Since $\mathbb{R}_{\ast }^{3}$ is a metric space, we can endow $C_{\ast
}\left( \R^3\right) $ with a metric topology in the usual manner.
Most of the random force fields used in this paper will be more regular than
just continuous. We then define:%
\begin{equation*}
C_{\ast }^{k}\left( \R^3\right) =\left\{ F\in C_{\ast }\left( 
\R^3\right) :F\in C^{k}\left( \R^3\setminus F^{-1}\left(
\infty \right) \right) \right\} \ \text{for}\quad k=1,2,...
\end{equation*}

Therefore, the elements of $C_{\ast }^{k}\left( \R^3\right) $ are
just $C^{k}$ functions at the points where they are bounded. We will not
need to define any topology on the spaces $C_{\ast }^{k}\left( \mathbb{R}%
^{3}\right) .$

We are interested in time-dependent random force fields. We then define the
metric space $C\left( \mathbb{R}:C_{\ast }\left( \R^3\right)
\right) .$ We could define similarly $C\left( \left[ 0,T\right] :C_{\ast
}\left( \R^3\right) \right) ,$ but we will use in this paper only
time-dependent random force fields defined globally in time.

Therefore we define:%
\begin{equation*}
C^{k}\left( \mathbb{R}:C_{\ast }^{k}\left( \R^3\right) \right)
=\left\{ F\in C\left( \mathbb{R}:C_{\ast }\left( \R^3\right)
\right) :F\in C^{k}\left( \left( \mathbb{R}\times \R^3\right)
\diagdown F^{-1}\left( \infty \right) \right) \right\} \ \text{for}\quad k=1,2,\dots
\end{equation*}%
Notice that $C^{k}\left( \mathbb{R}:C_{\ast }^{k}\left( 
\R^3\right) \right) $ is the subset of the set of functions of $%
C\left( \mathbb{R}:C_{\ast }\left( \R^3\right) \right) $ which
have $k$ continuous derivatives at the points where $F$ is bounded. We will
use the shorthand notation $F\in C^{k}$ for $F\in C^{k}\left( \mathbb{R}%
:C_{\ast }^{k}\left( \R^3\right) \right) .$

We introduce a $\sigma -$algebra on the space $C\left( \mathbb{R}:C_{\ast
}\left( \R^3\right) \right) $ generated by the cylindrical sets,
i.e. the $\sigma -$algebra generated by the sets $\left\{ Y(t_0,x_0)\in B
\right\} $ where $B$ is a Borel set of $\R^3_{\ast }$ and $\left(
t_{0},x_{0}\right) \in \mathbb{R}\times \R^3$. We will denote this 
$\sigma -$algebra as $\mathcal{B}$.

All the random force fields in which we are interested in this paper are
contained in the following definition.

\begin{definition}
\label{RandForField}Let $\left( \Omega ,\mathcal{F},\mu \right) $ be a measure
space where $\mathcal{F}$ is a $\sigma -$algebra of subsets of $\Omega $ and 
$\mu $ is a probability measure. A random force field is a measurable
mapping $F$ from $\Omega $ to the set of functions $C\left( \mathbb{R}%
:C_{\ast }\left( \R^3\right) \right) $ with respect to the $\sigma
-$algebra $\mathcal{B}$.
\end{definition}

Notice that a random force field defines a probability $\mathbb{P}$ on the $%
\sigma -$algebra $\mathcal{B}$, which consists of subsets of $C\left( 
\mathbb{R}:C_{\ast }\left( \R^3\right) \right) $, by means of:%
\begin{equation*}
\mathbb{P}\left( A\right) =\mu \left( \left\{ \omega \in \Omega :F\left(
\omega \right) \in A\right\} \right) \ \ ,\ \ A\in \mathcal{B}.
\end{equation*}%
A random force field can be characterized by the family of random variables 
\begin{equation*}
\left\{\omega \rightarrow F\left( X,\tau ;\omega \right),\ X\in \mathbb{R}%
^{3},\ \tau \in \left[ 0,T\right] \right\}.
\end{equation*} 
We can define the action of the group of spatial translations on $%
C^{k}\left( \mathbb{R}:C_{\ast }^{k}\left( \R^3\right) \right) $
by means of:%
\begin{equation*}
T_{a}F\left( X,\tau \right) =F\left( X+a,\tau \right) \text{ for each }a\in 
\R^3
\end{equation*}%
and the group of time translations by means of:%
\begin{equation*}
U_{b}F\left( X,\tau \right) =F\left( X,\tau +b\right) \text{ for each }b\in 
\mathbb{R}.
\end{equation*}

We will say that a random force field $F$ is invariant under spatial
translations (or just invariant under translations) if we have:%
\begin{equation*}
\mu \left( \left\{ \omega \in \Omega :F\left( \cdot ;\omega \right) \in
A\right\} \right) =\mu \left( \left\{ \omega \in \Omega :T_{a}F\left( \cdot
;\omega \right) \in A\right\} \right)
\end{equation*}%
for each $A\in \mathcal{B} $ and any $a\in \R^3.$ We will say that
a random force field $F$ is invariant under time translations (or
stationary) if:%
\begin{equation*}
\mu \left( \left\{ \omega \in \Omega :F\left( \cdot ;\omega \right) \in
A\right\} \right) =\mu \left( \left\{ \omega \in \Omega :U_{b}F\left( \cdot
;\omega \right) \in A\right\} \right)
\end{equation*}%
for each $A\in \mathcal{B} $ and any $b\in \mathbb{R}$.

\bigskip

Actually, all the random force fields considered in this paper will be also
invariant under rotations. Given $M\in SO\left( 3\right) $ we define the
action of the group $SO\left( 3\right) $ on $C^{k}\left( \mathbb{R}:C_{\ast
}^{k}\left( \R^3\right) \right) $ by means of:%
\begin{equation*}
R_{M}F\left( X,\tau \right) =F\left( MX,\tau \right) \text{ for each }M\in
SO\left( 3\right).
\end{equation*}

Then, the random force field $F$ is invariant under the group $SO\left(
3\right) $ if: 
\begin{equation*}
\mu \left( \left\{ \omega \in \Omega :F\left( \cdot ;\omega \right) \in
A\right\} \right) =\mu \left( \left\{ \omega \in \Omega :R_{M}F\left( \cdot
;\omega \right) \in A\right\} \right)
\end{equation*}%
for each $A\in \mathcal{B}$ and any $M\in SO\left( 3\right) .$

\subsection{Time-dependent particle configurations and random force fields
\label{TimeDepConf}}

\subsubsection{Random particle configurations}

In this Section we describe a family of random force fields that are
generated by particles distributed randomly in the phase space according to
the Poisson distribution at time $t=0$ which move at constant velocity for
positive times. We will take as unit of length the typical distance between
particles $d,$ i.e. $d=1.$ The particle velocities are distributed according
to a finite nonnegative measure $g=g\left( dv\right) $, independently from
the particle positions. Note that $\int_{\mathbb{R}^{d}}g\left( dv\right) $
is the spatial particle density. We will choose the unit of time $\tau$ in
such a way that the average particle velocity is of order one. Each particle
is the center of a radial potential $\Phi=\Phi\left( \left\vert y\right\vert
\right) \in C^{2}\left( \R^3\diagdown\left\{ 0\right\} \right) .$
We will consider two different types of potentials. First we will consider
potentials of order one with short ranges, i.e. smaller range than the
particle distance $d.$ Secondly, we will consider weak potentials with
arbitrary range, but typically larger or equal than the particle distance $%
d. $

More precisely, we will denote as $\Lambda_{p}$ the space of locally finite
particle configurations in the phase space $\R^3\times\mathbb{R}%
^{3}.$ Each of these particle configurations can be represented by a
sequence $\left\{ \left( x_{k},v_{k}\right) \right\} _{k\in\mathbb{N}}$ with 
$x_{k}\in\R^3$ and $v_{k}\in\R^3$ where all the
sequences which can be obtained from another by means of a 
permutation of the particles are equivalent and represent the same particle
configuration. We consider the $\sigma-$algebra $\Sigma_{p}$ generated by
the sets:%
\begin{equation*}
U_{B,n}=\left\{ \left\{ \left( x_{k},v_{k}\right) \right\} _{k\in \mathbb{N}%
}\in\Lambda_{p}:\#\left[ \left\{ \left( x_{k},v_{k}\right) \right\} \cap B%
\right] =n\right\} 
\end{equation*}
for each $n\in\mathbb{N}_{\ast}$ any Borel set $B\subset\R^3\times%
\R^3.$ We define a measure in $\mathbb{R}^{d}\times \mathbb{R}^{d}$
by means of the product measure $dxg\left( dv\right) .$ We then define a
probability measure $\nu_{g}$ on $\Sigma_{p}$ by means of:%
\begin{equation}
\nu_{g}\left( \bigcap_{j=1}^{J}U_{B_{j},n_{j}}\right) =\prod_{j=1}^{J}\left[ 
\frac{\left\vert \int_{B_{j}}dxg\left( dv\right) \right\vert ^{n_{j}}}{%
\left( n_{j}\right) !}e^{-\int_{B_{j}}dxg\left( dv\right) }\right] 
\label{S4E2}
\end{equation}
where $B_{j}$ is a Borel set of $\R^3\times\R^3$ and $%
n_{j}\in\mathbb{N}_{\ast}$ for each $j\in\left\{ 1,2,...,J\right\} $ with $%
B_{j}\cap B_{k}=\varnothing$ if $j\neq k.$

We define the free flow evolution group $T\left( \tau\right) $, $\tau \in%
\mathbb{R}$, on the space of particle configurations $\Lambda_{p}$ as
follows. Suppose that we represent a particle configuration $\xi\in\Lambda
_{p}$ by the sequence $\left\{ \left( x_{k},v_{k}\right) \right\} _{k\in%
\mathbb{N}}.$ Then we define: 
\begin{equation}
T\left( \tau\right) \xi=\left\{ \left( x_{k}+v_{k}\tau,v_{k}\right) \right\}
_{k\in\mathbb{N}}\ \ ,\ \ \tau\in\mathbb{R} \ .  \label{S4E3}
\end{equation}

This definition yields a mapping $T\left( \tau\right) :\Lambda
_{p}\rightarrow\Lambda_{p}$ which is independent of the specific sequence
used to label $\xi\in\Lambda_{p}.$ It is not hard to exhibit examples of
particle configurations $\xi\in\Lambda_{p}$ for which the configuration
defined by means of (\ref{S4E3}) is not locally finite for $\tau\neq0$. This
can be achieved giving to some particles placed very far away from the
origin large velocities which transport infinitely many particles to a
bounded region for some times $t_{0}\in\mathbb{R}$. However, this does not
happen with probability one if the particles are chosen according to the
probability measure $\nu_{g}$ defined in (\ref{S4E2}). More precisely, we
have the following result:

\begin{proposition}
\label{PoissonEvolution}Let $\left( \Lambda_{p},\Sigma_{p},\nu_{g}\right) $ be
the measure space of particle configurations, with $\nu_{p}$ as in (\ref%
{S4E2}). Then, the evolution group $T\left( \tau\right) :\Lambda
_{p}\rightarrow\Lambda_{p}$ defined by means of (\ref{S4E3}) is well defined
(i.e. $T\left( \tau\right) \xi$ is locally finite) for any $\tau \in\mathbb{R%
}$, for $a.e.$ $\xi\in\Lambda_{p}.$

The pushforward measure $\nu_{g}\circ T\left( - \tau\right) $ satisfies $%
\nu_{g}\circ T\left( -\tau\right) =\nu_{g}$ for each $\tau\in\mathbb{R}$.
\end{proposition}

\begin{proof}
First we note that $T(\tau)\xi\in\Lambda_{p}$ holds $\nu_{g}$ a.e. for every 
$\tau\in{\mathbb{R}}$. We have 
\begin{align*}
\nu_{g}({\xi\in\Lambda_{p}: T(\tau)\xi\text{ not loc. finite}}) \leq\sum
_{n=1}^{\infty}\nu_{g}(\xi\in\Lambda_{p}: T(\tau)\xi\cap B_{n} \text{
infinite}),
\end{align*}
therefore it suffices to show $\nu_{g}(\xi\in\Lambda_{p}: T(\tau)\xi\cap
B_{R} \times{\mathbb{R}}^{3} \text{ infinite}) = 0$ for every $R>0$. To this
end we observe that 
\begin{align*}
\sum_{n=1}^{\infty}\nu_{g}(\xi\in\Lambda_{p}: \exists k\in S,
n\leq|v_{k}|\leq n+1, x_{k} \in B_{R} -\tau v_{k})< \infty.
\end{align*}
Hence by Borel-Cantelli we have $\nu_{g}(\xi\in\Lambda_{p}: \exists N>0\,
x_{k} +\tau v_{k} \in B_{R} \Rightarrow x_{k} \in B_{N})=1$, so $%
\nu_{g}(\xi\in\Lambda_{p}: T(\tau)\xi\cap B_{R} \times{\mathbb{R}}^{3} \text{
infinite}) = 0$ as claimed.

Furthermore, for any Borel sets $A,B\subset{\mathbb{R}}^{3}$ and $n\in{%
\mathbb{N}}$ we have 
\begin{align*}
\nu_{g}(\xi\in\Lambda_{p} : |\{T(\tau) \xi\cap A \times B\}| = n) & =
\nu_{g}(\xi\in\Lambda_{p} : |\{k \in S: v_{k} \in B, x_{k} \in A-\tau
v_{k}\}| = n) \\
& = \nu_{g}(\xi\in\Lambda_{p} : |\{k \in S: v_{k} \in B, x_{k} \in A\}| = n)
\\
& = \nu_{g}(\xi\in\Lambda_{p} : |\xi\cap A \times B\}| = n).
\end{align*}
Repeating the same computation for the cylinder sets shows $\nu_{g} =
\nu_{g} \circ T(-\tau)$.
\end{proof}

Actually several of the random force fields considered in this paper will
contain at least two different types of particles having different types of
charges. This is due to the fact that in order to define some of the long
range potentials, in particular those behaving for large values as
Coulombian potentials, an  electroneutrality condition is required in order to
be able to define spatially homogeneous random force fields (cf. \cite{NSV},
Theorem 2.13). On the other hand, there is no reason to assume that in
multicomponent systems all the particles have the same velocity
distribution. Suppose that we consider systems with $L$ different types of
particles having respectively the charges $\left\{ Q_{\ell}\right\}
_{\ell=1}^{L}$ and velocity distributions $\left\{ g_{\ell}\left( dv\right)
\right\} _{\ell=1}^{L}$ where $g_{\ell}$ are finite Radon measures in $%
\R^3.$ We can then generalize (\ref{S4E2}) as follows. We define a
set of configurations $\Lambda _{p}^{\left( L\right) }$ by means of: 
\begin{equation}
\Lambda_{p}^{\left( L\right) }=\left\{ \omega=\left\{ \left(
x_{k,1},v_{k,1};x_{k,2},v_{k,2};...;x_{k,L},v_{k,L}\right) \right\} _{k\in%
\mathbb{N}}:\left\{ \left( x_{k,\ell},v_{k,\ell}\right) \right\} _{k\in%
\mathbb{N}},\ \ell\in\left\{ 1,2,...,L\right\} \right\} / {\sim} 
\label{S5E5}
\end{equation}
where the equivalence relation $\sim$ identifies all the sequences which can
be obtained from another by means of a permutation of the particles within a
single species. Now, given Borel sets $B_{j}\subset\mathbb{R}^{d}\times%
\mathbb{R}^{d}$ and integers $n_{j,\ell}\in\mathbb{N}_{\ast}$ for $%
j\in\left\{ 1,2,...,J\right\} ,\ \ell\in\left\{ 1,2,...,L\right\} $ we
define sets:%
\begin{equation*}
U_{B_{j},n_{j,\ell}}^{\left( \ell\right) }=\left\{ \omega\in\Lambda
_{p}^{\left( L\right) }:\#\left[ \left\{ \left( x_{k,\ell},v_{k,\ell
}\right) \right\} \cap B_{j}\right] =n_{j,\ell}\right\} 
\end{equation*}
where we assume that $\omega$ is as in (\ref{S5E5}) and $\ell\in\left\{
1,2,...,L\right\} .$ We define the $\sigma-$algebra $\mathcal{F}_{L}$ of
subsets of $\Lambda_{p}^{\left( L\right) }$ as the smallest $\sigma-$algebra
containing all the sets $U_{B_{j},n_{j,k}}^{\left( k\right) }.$ We then
define a measure space $\left( \Lambda_{p}^{\left( L\right) },\mathcal{F}%
_{L},\nu_{\left\{ g_{\ell}\right\} _{\ell=1}^{L}}\right) $ by means of:%
\begin{equation}
\nu_{\left\{ g_{\ell}\right\} _{\ell=1}^{L}}\left(
\bigcap_{j=1}^{J}\bigcap_{\ell=1}^{L}U_{B_{j},n_{j,\ell}}\right)
=\prod_{j=1}^{J}\prod _{\ell=1}^{L}\left[ \frac{\left\vert
\int_{B_{j}}dxg_{\ell}\left( dv\right) \right\vert ^{n_{j,\ell}}}{\left(
n_{j,\ell}\right) !}e^{-\int_{B_{j}}dxg_{\ell}\left( dv\right) }\right] . 
\label{S5E6}
\end{equation}

We will say that the distribution of particles defined by means of the
probability measure $\nu_{\left\{ g_{\ell}\right\} _{\ell=1}^{L}}$ satisfies
the electroneutrality condition (with charges $\left\{ Q_{\ell}\right\}
_{\ell=1}^{L}$) if the following identity holds:%
\begin{equation}
\sum_{\ell=1}^{L}Q_{\ell}\int_{\mathbb{R}^{d}}g_{\ell}\left( dv\right) =0. 
\label{ElNeut}
\end{equation}

\subsubsection{Random force fields generated by freely moving random
particle distributions}

We can define now a family of random force fields taking as starting point
the random particle configurations defined in the previous section. Given a
family of radially symmetric interaction potentials $\phi=\phi\left(
\left\vert x\right\vert \right) $ such that $\nabla\phi\in
C_{\ast}^{2}\left( \R^3\right) ,$ we want to give a meaning to the
following expressions in order to define suitable random force fields. In
the case of particle configurations in $\Lambda_{p}:$%
\begin{equation}
F\left( x,\tau;\omega;g\right) =-\sum_{k\in\mathbb{N}}\nabla\phi\left(
x-x_{k}-v_{k}\tau\right) \ \ ,\ \ \ \omega=\left\{ \left( x_{k},v_{k}\right)
\right\} _{k\in\mathbb{N}}\in\Lambda_{p}   \label{S5E7}
\end{equation}
and in the case of particle configurations with different types of charges $%
\left\{ Q_{\ell}\right\} _{\ell=1}^{L}$ and velocity distributions $\left\{
g_{\ell}\left( dv\right) \right\} _{\ell=1}^{L}$ the goal is to give a
meaning to expressions like:%
\begin{equation}
F\left( x,\tau;\omega;\left\{ g_{\ell}\right\} _{\ell=1}^{L}\right)
=-\sum_{\ell=1}^{L}\sum_{k\in\mathbb{N}}Q_{\ell}\nabla\phi\left( x-x_{k,\ell
}-v_{k,\ell}\tau\right) \ \ ,\ \ \ \omega\in\Lambda_{p}^{\left( L\right) } 
\label{S5E8}
\end{equation}
with $\omega$ as in (\ref{S5E5}). Similarly, we define truncations of the
expressions above, defined by 
\begin{align}
F^{R}\left( x,\tau;\omega;g\right) & =-\sum_{k\in\mathbb{N}: |x_{k}|\leq
R}\nabla\phi\left( x-x_{k}-v_{k}\tau\right) \ \ ,\ \ \ \omega=\left\{ \left(
x_{k},v_{k}\right) \right\} _{k\in\mathbb{N}}\in\Lambda _{p}
\label{S5E7trunc} \\
F^{R}\left( x,\tau;\omega;\left\{ g_{\ell}\right\} _{\ell=1}^{L}\right) &
=-\sum_{\ell=1}^{L}\sum_{k\in\mathbb{N}: |x_{k,l}|\leq R}Q_{\ell}\nabla
\phi\left( x-x_{k,\ell}-v_{k,\ell}\tau\right) \ \ ,\ \ \ \omega\in
\Lambda_{p}^{\left( L\right) } .   \label{S5E8trunc}
\end{align}
\medskip

The convergence of the series on the right-hand side of (\ref{S5E7}), (\ref%
{S5E8}) for $\tau=0$ and a large class of interaction potentials $\phi$ has
been considered in \cite{NSV}. Given that the distribution of points $%
\left\{ x_{k}+v_{k}\tau\right\} _{k\in\mathbb{N}}$ is given by a Poisson
distribution (cf.~\ref{PoissonEvolution}) in $\R^3$ these results
hold for any $\tau\in\mathbb{R}$. The convergence of the series in (\ref%
{S5E7}), (\ref{S5E8}) is not immediate for potentials $\phi\left( \left\vert
x\right\vert \right) $ decreasing like nonintegrable power laws for large
values of $\left\vert x\right\vert .$ Using the methods in \cite{NSV} we
might then see that the right-hand side of (\ref{S5E7}) can be given a
meaning for any $\tau\in\mathbb{R}$ with probability one and $\phi\left(
\left\vert x\right\vert \right) \sim\frac{C}{\left\vert x\right\vert ^{s}},\
s>1$ defining the right-hand side of (\ref{S5E7}) as the limit as $%
R\rightarrow\infty$ of the sum over points contained in a sphere $%
B_{R}\left( 0\right) .$ In the case $s=1$ it has been proved in \cite{NSV}
that such limits exist and they define a random force field invariant
under translations if the electroneutrality condition (\ref{ElNeut}) holds.

The type of arguments used in \cite{NSV} can be adapted to prove that the
random force fields in (\ref{S5E7}), (\ref{S5E8}) are defined for all $\tau
\in\mathbb{R}$ with probability one. The following set of conditions for the
function $\phi$ will be used in the definition of the random force fields $F.
$

\begin{equation}
\left\vert \phi\left( x\right) \right\vert +\left\vert x\right\vert
\left\vert \nabla\phi\left( x\right) \right\vert \leq\frac{C}{\left\vert
x\right\vert ^{s}}\ \ \ \text{for }\left\vert x\right\vert \geq1\text{ with }%
s>2   \label{phiDec_slarge}
\end{equation}%
\begin{equation}
\left\vert \phi\left( x\right) -\frac{A}{\left\vert x\right\vert ^{s}}%
\right\vert +\left\vert x\right\vert \left\vert \nabla\phi\left( x\right) +%
\frac{Ax}{\left\vert x\right\vert ^{s+2}}\right\vert \leq\frac{C}{\left\vert
x\right\vert ^{s+1}}\ \ \ \text{for }\left\vert x\right\vert \geq1\text{
with }s\in(1,2) ,\ A\in\mathbb{R}   \label{phiDec_sInt}
\end{equation}%
\begin{equation}
\left\vert \phi\left( x\right) -\frac{A}{\left\vert x\right\vert }%
\right\vert +\left\vert x\right\vert \left\vert \nabla\phi\left( x\right) +%
\frac{Ax}{\left\vert x\right\vert ^{3}}\right\vert \leq\frac{C}{\left\vert
x\right\vert ^{2+\delta}}\ \ \ \text{for }\left\vert x\right\vert \geq1\text{
with } \ A\in\mathbb{R}\text{,\ }\delta>0.   \label{phiDec_sOne}
\end{equation}
In the three formulas (\ref{phiDec_slarge})-(\ref{phiDec_sOne}) we assume
that $C>0.$

We will further assume that the functions $g$, $g_{l}$ satisfy: 
\begin{align}
\int g(v)|v|^{6+\kappa }\;\mathrm{d}{v}& <\infty ,  \label{Momentg} \\
\int g_{l}(v)|v|^{6+\kappa }\;\mathrm{d}{v}& <\infty ,  \label{Momentgl}
\end{align}%
for some $\kappa >0$. We then have the following result:

\begin{proposition}
\label{RandForFieldMovPart} Let $\phi :\R^3\rightarrow \mathbb{R}%
_{\ast }$ be a radially symmetric interaction potential such that $\phi
,\nabla \phi ,\nabla ^{2}\phi \in C_{\ast }\left( \R^3\right) $. Let $B_{1}\left( 0\right) $ be the unit ball in $\R^3$. 
The following statements hold.

\begin{itemize}
\item[(i)] Suppose that $\phi $ satisfies (\ref{phiDec_slarge}) and that $\nabla
^{2}\phi $ is bounded in $L^{p}\left(
B_{1}\left( 0\right) \right) $ for some $p>1$. Let $\nu _{g}$ be the
probability measure in the measure space $\left( \Lambda _{p},\Sigma
_{p},\nu _{g}\right) $ defined by means of (\ref{S4E2}) where $g$ satisfies %
\eqref{Momentg}. Then for any $x\in \R^3$ the series in (\ref{S5E7}%
) converges absolutely for all $\tau \in \mathbb{R}$ for $\nu _{g}-$almost $%
\omega \in \Lambda _{p}.$ Moreover, the series (\ref{S5E7}) defines a random
force field in $C\left( \mathbb{R}:C_{\ast }\left( \R^3\right)
\right) .$

\item[(ii)] Suppose that $\phi $ satisfies (\ref{phiDec_sInt}) and that $%
\nabla ^{2}\phi $ is bounded in $L^{p}\left(
B_{1}\left( 0\right) \right) $ for some $p>1$. Let  $\nu _{g}$ be the
probability measure in the measure space $\left( \Lambda _{p},\Sigma
_{p},\nu _{g}\right) $ defined by means of (\ref{S4E2}) where $g$ satisfies %
\eqref{Momentg}. Then, for each $x\in \R^3$ the following limit
exists for all $\tau \in \mathbb{R}$ for $\nu _{g}-$almost $\omega \in
\Lambda _{p}:$ 
\begin{equation}
F\left( x,\tau ;\omega ;g\right) =-\lim_{R\rightarrow \infty }\sum_{\left\{
\left\vert x_{k}\right\vert \leq R\right\} }\nabla \phi \left(
x-x_{k}-v_{k}\tau \right) \ .  \label{S9E9}
\end{equation}
Moreover, the series in (\ref{S9E9}) defines a random force field in $%
C\left( \mathbb{R}:C_{\ast }\left( \R^3\right) \right) .$

\item[(iii)] Suppose that $\phi $ satisfies (\ref{phiDec_sOne}) and that $%
\nabla \phi $  is in the
Sobolev space $H^{s}\left( B_{1}\left( 0\right) \right) $ with $s>\frac{1}{2}
$. Let $\left( \Lambda _{p}^{\left( L\right) },\mathcal{F}_{L},\nu
_{\left\{ g_{\ell }\right\} _{\ell =1}^{L}}\right) $ be the measure space
defined by means of (\ref{S5E5}), (\ref{S5E6}) where the functions $\left\{
g_{\ell }\right\} _{\ell =1}^{L}$ satisfy \eqref{Momentgl} and the
electroneutrality condition (\ref{ElNeut}) holds. Then, for each $x\in 
\R^3$ the following limit exists for all $\tau \in \mathbb{R}$ for 
$\nu _{\left\{ g_{\ell }\right\} _{\ell =1}^{L}}-$almost $\omega \in \Lambda
_{p}^{\left( L\right) }:$%
\begin{equation}
F\left( x,\tau ;\omega ;\left\{ g_{\ell }\right\} _{\ell =1}^{L}\right)
=-\lim_{N\rightarrow \infty }\sum_{\ell =1}^{L}\sum_{\left\{ \left\vert
x_{k,\ell }\right\vert \leq 2^{N}\right\} }Q_{\ell }\nabla \phi \left(
x-x_{k,\ell }-v_{k,\ell }\tau \right) \ .  \label{T8E2}
\end{equation}
Moreover, the series in (\ref{T8E2}) defines a random force field in $%
C\left( \mathbb{R}:C_{\ast }\left( \R^3\right) \right) .$
\end{itemize}
\end{proposition}

\begin{proof}
We will only sketch the proof of Proposition \ref{RandForFieldMovPart} in
the case (iii) since it is the most involved. The generalization to (ii) and
(iii) can be made along similar lines using ideas analogous to the ones in the proof of
Theorem 2.6 in \cite{NSV}.

We first split $F$ into the contribution due to the close particles and the
long-range contribution.  To this end we split $\phi $ as $\phi =\phi
_{1}+\phi _{2}$ where $\phi _{1}$ is supported in the unit ball and it is
smooth away from the origin. The function $\phi _{2}$ is smooth in the whole
space $\R^3.$ More precisely we introduce a cutoff 
$\eta\in C^{\infty}\left(\R^3\right)$ such that $\eta\left(  x\right)=\eta\left(
\left\vert x\right\vert \right),$ $0\leq\eta\leq1,$ $\eta\left(  x\right)
=\frac 1 2$ if $\left\vert x\right\vert \leq1,$ $\eta\left(  x\right)  =0$ if
$\left\vert x\right\vert \geq 1$. We set %
\begin{equation*}
\phi_{1}\left(x\right)  :=\phi\left(  x \right)
\eta\left(  \left\vert x\right\vert \right)
\ \ ,\ \ \phi_{2}\left(  x \right)  :=\phi\left(  x 
\right)  \left[ 1-\eta\left(  \left\vert x\right\vert \right)  \right]\;.  \ \label{S4E5}%
\end{equation*}
We then define the random force fields $F_{1},\ F_{2}$, as in (%
\ref{T8E2}), using the potentials $\phi _{1},\ \phi _{2}$, so that $F=F_1+F_2$. We estimate first
the contribution $F_{2}.$ To this end we define the random variables%
\begin{eqnarray}
f^{(2)}_{0}\left( x,\tau ;\omega \right)  &=&-\sum_{\ell =1}^{L}\sum_{\left\{
\left\vert x_{k,\ell }\right\vert \leq 1\right\} }Q_{\ell }\nabla \phi
_{2}\left( x-x_{k,\ell }-v_{k,\ell }\tau \right)   \label{fPiecesDef} \\
f^{(2)}_{j}\left( x,\tau ;\omega \right)  &=&-\sum_{\ell =1}^{L}\sum_{\left\{
2^{j-1}<\left\vert x_{k,\ell }\right\vert \leq 2^{j}\right\} }Q_{\ell
}\nabla \phi _{2}\left( x-x_{k,\ell }-v_{k,\ell }\tau \right) \ \ ,\ \
j=1,2,3,...  \notag
\end{eqnarray}
Then, we have
\begin{equation}
-\sum_{\ell =1}^{L}\sum_{\left\{ \left\vert x_{k,\ell }\right\vert \leq
2^{N}\right\} }Q_{\ell }\nabla \phi _{2}\left( x-x_{k,\ell }-v_{k,\ell }\tau
\right) =\sum_{j=0}^{N}f_{j}\left( x,\tau ;\omega \right) \ .  \label{FfarFract}
\end{equation}

Notice that using the probability measure $\nu _{\left\{ g_{\ell }\right\}
_{\ell =1}^{L}}$ defined by means of (\ref{S5E6}) as well as (\ref{ElNeut})
we obtain
\begin{equation}
\mathbb{E}\left[ f^{(2)}_{j}\left( x,\tau ;\omega \right) \right] =0.
\end{equation}
We now observe that 
\begin{equation}\label{eq:corf2jf2m}
\mathbb{E}\left[ \int_{0}^{T}f^{(2)}_{j}\left( x,\tau ;\omega \right) \otimes
f^{(2)}_{m}\left( x,\tau ;\omega \right) d\tau \right]  
=\delta _{j,m}\int_{0}^{T}d\tau \mathbb{E}\left[ f_{j}\left( x,\tau
;\omega \right) \otimes f_{m}\left( x,\tau ;\omega \right) \right] .
\end{equation}
Using that the distributions of particles in the sets $\left\{ 2^{j-1}<\left\vert x_{k,\ell
}\right\vert \leq 2^{j}\right\} $ (and $\left\{ \left\vert x_{k,\ell
}\right\vert \leq 2\right\} $) are mutually independent,  which implies that the random
variables $f^{(2)}_{j}\left( x,\tau ;\omega \right) $ are mutually independent as well, we obtain 
\begin{eqnarray*}
&&\mathbb{E}\left[ f^{(2)}_{j}\left( x,\tau ;\omega \right) \otimes f^{(2)}_{m}\left(
x,\tau ;\omega \right) \right]  \\
&=&\sum_{\ell =1}^{L}\mathbb{E} \left[ \sum_{\left\{ 2^{j-1}<|x_{k,\ell }|
\leq 2^{j}\right\} }\left( Q_{\ell }\right) ^{2}\left[ \nabla \phi
_{2}\left( x-x_{k,\ell }-v_{k,\ell }\tau \right) \otimes \nabla \phi
_{2}\left( x-x_{k,\ell }-v_{k,\ell }\tau \right) \right]\right] .
\end{eqnarray*}
We now use the fact that the Poisson distributions for the different charges have rate $\int_{\mathbb{R}%
^{3}}g_{\ell }\left( dv\right) $ and that they are invariant under translations. Then:
\begin{eqnarray*}
&&\mathbb{E}\left[\sum_{\left\{ 2^{j-1}<|x_{k,\ell }|
	\leq 2^{j}\right\} } \nabla \phi _{2}\left( x-x_{k,\ell }-v_{k,\ell }\tau
\right) \otimes \nabla \phi _{2}\left( x-x_{k,\ell }-v_{k,\ell }\tau \right) %
\right]  \\
&=&\int_{\R^3}dy\int_{\left\{ 2^{j-1}<\left\vert y\right\vert \leq
2^{j}\right\} }\left[ \nabla \phi _{2}\left( x-y-v\tau \right) \otimes
\nabla \phi _{2}\left( x-y-v\tau \right) \right] g_{\ell }\left( dv\right) .
\end{eqnarray*}
Since $\phi _{2}$ is smooth near the origin and $\nabla\phi_2(\vert y\vert )\sim \frac{1}{\left\vert y\right\vert ^{2}}$ for
large distances, we obtain
\begin{eqnarray*}
&&\left\vert \mathbb{E}\left[ \sum_{\left\{ 2^{j-1}<|x_{k,\ell }|
	\leq 2^{j}\right\} }\nabla \phi_{2}\left( x-x_{k,\ell }-v_{k,\ell
}\tau _{1}\right) \otimes \nabla \phi_{2} \left( x-x_{k,\ell }-v_{k,\ell }\tau
_{2}\right) \right] \right\vert  \\
&\leq &\int_{\R^3}dy\int_{\left\{ 2^{j-1}<\left\vert y\right\vert
\leq 2^{j}\right\} }\frac{g_{\ell }\left( dv\right) }{\left( 1+\left\vert
x-y-v\tau _{1}\right\vert ^{2}\right) \left( 1+\left\vert x-y-v\tau
_{2}\right\vert ^{2}\right) }.
\end{eqnarray*}
We assume that $\left\vert x\right\vert $ is bounded, and that $%
\tau _{1},\ \tau _{2}$ are bounded. We estimate separately the contributions due to the region where $\left\vert v\right\vert \leq
2^{\frac{j}{2}}$ and the region where $\left\vert v\right\vert >2^{\frac{j}{2}}$ to obtain
\begin{equation*}
\int_{\R^3}dy\int_{\left\{ 2^{j-1}<\left\vert y\right\vert
	\leq 2^{j}\right\} }\frac{g_{\ell }\left( dv\right) }{\left( 1+\left\vert
	x-y-v\tau _{1}\right\vert ^{2}\right) \left( 1+\left\vert x-y-v\tau
	_{2}\right\vert ^{2}\right) }\leq C\left(\frac{1}{2^{j}}+2^{-\frac{\kappa j}{2}}\right).
\end{equation*}
Hence, from \eqref{eq:corf2jf2m} using the estimates above, we otain
\begin{equation*}
\mathbb{E}\left[ \int_{0}^{T}f^{(2)}_{j}\left( x,\tau ;\omega \right) \otimes
f^{(2)}_{\ell }\left( x,\tau ;\omega \right) d\tau \right]  =\left[
\int_{0}^{T}K_{j}\left( \tau \right) d\tau \right] \delta _{j,\ell }
\end{equation*}
for $j=0,1,2,...,$ where:%
\begin{equation*}
\left\vert K_{j}\left( \tau \right) \right\vert \leq C2^{- \frac{\kappa }{2} j}\ \ \text{for }\tau \in \left[ 0,T\right] 
\end{equation*}
for some $C>0,$ which depends only on $T.$

Arguing similarly, we can obtain estimates for the derivatives of the
functions $f^{(2)}_{j}.$ More precisely, we have the estimate
\begin{equation*}
\mathbb{E}\left[ \int_{0}^{T}\left\vert \partial _{\tau }f^{(2)}_{j}\left( x,\tau
;\omega \right) \right\vert ^{2}d\tau \right] \leq \tilde{C}_{T}2^{-\frac{\kappa }{2} j}\ \ ,\ \ j=0,1,2,...
\end{equation*}%
whence, using Morrey's Theorem, we obtain that%
\begin{equation*}
\mathbb{E}\left[ \sup_{0\leq t\leq T}\left\vert f^{(2)}_{j}\left( x,\tau ;\omega
\right) \right\vert \right] \leq C_{T}2^{-\frac{\kappa }{4}j}\ \ ,\ \ j=0,1,2,...
\end{equation*}

We can then prove, using Borel-Cantelli Theorem as in \cite{NSV}, that the
limit as $N\rightarrow \infty $ of the right-hand side of (\ref{FfarFract})
exists for all $\tau \in \left[ 0,T\right] .$

We now prove the existence of the random force field $F_1$ associated to the
localized part of the potential $\phi _{1}.$ To this end we define random
functions $f^{(1)}_{0},\ f^{(1)}_{j}$ as in (\ref{fPiecesDef}) replacing $\phi _{2}$ by $%
\phi _{1}.$ Arguing as before and using the fact that $\phi _{2}\in H^s(\R^3)$ we obtain the estimate%
\begin{equation*}
\mathbb{E}\left[ \int_{0}^{T}\left\vert \partial _{\tau }^{s}f^{(1)}_{j}\left(
x,\tau ;\omega \right) \right\vert ^{2}d\tau \right] \leq C_{T}2^{- 
\frac{\kappa }{2} j}\ \ ,\ \ j=0,1,2,...
\end{equation*}%
with $p>1,$ where $\partial _{\tau }^{s}$ is the fractional derivative of
order $s$ defined by means of incremental quotiens (see for instance \cite{Ad, Br}). Then, using that $s>\frac{1}{2}$ we obtain%
\begin{equation*}
\mathbb{E}\left[ \sup_{0\leq \tau \leq T}\left\vert f^{(1)}_{j}\left( x,\tau
;\omega \right) \right\vert \right] \leq C_{T}2^{- \frac{\kappa }{4} j}\ \ ,\ \ j=0,1,2,...
\end{equation*}

We can then argue as in the estimate of the random force field $F_2$ 
due to $\phi _{2}$ to prove the convergence of the corresponding series
uniformly for all $\tau \in \left[ 0,T\right] .$ Then the convergence of the
right hand side of (\ref{T8E2}) follows.  
\end{proof}

\begin{remark}
We could also define random force fields as in Proposition \ref%
{RandForFieldMovPart} for interaction potentials $\phi$ satisfying $%
\phi\left( x\right) \sim\frac{A}{\left\vert x\right\vert ^{s}}$ as $%
\left\vert x\right\vert \rightarrow\infty$ for $\frac{1}{2}<s<1,$ as it was
made in \cite{NSV} for stationary particle distributions. The kinetic
equations that can arise for this class of potentials for interacting
particle systems has been discussed in \cite{NVW}.
\end{remark}

\begin{remark}
We can define also singular potentials $\phi$ taking only the values $0$ and 
$\infty$ in the case of hard-sphere interactions with simple modifications
of the previous arguments.
\end{remark}

\begin{remark}
Notice that if we assume additional regularity for $\phi,$ say $\phi\in C^{k}
$ we obtain a similar regularity for the random force field $F.$
\end{remark}

\bigskip 

\subsubsection{Different classes of random force fields}

We now introduce some classes of random force fields which can be obtained
as in Proposition~\ref{RandForFieldMovPart} by means of suitable choices of
potentials $\phi =\Phi _{\eps }$ with $\eps \rightarrow 0.$
The resulting random force fields will be of two different types which will
be denoted as Boltzmann and Landau random force fields respectively. 
We will introduce both of them below for the sake of completeness but in this paper we will just consider Landau random force fields. 

\paragraph{Boltzmann random force fields.}

We will say that a family of random force fields defined by means of (\ref%
{S5E7}) or (\ref{S5E8}) in the sense of Proposition \ref{RandForFieldMovPart}
is a family of Boltzmann random force fields if the functions $\phi
_{\eps }$ have the form:%
\begin{equation}
\Phi _{\eps }\left( x\right) =\Phi \left( \frac{\left\vert
x\right\vert }{\eps }\right) \ \ ,\ \ x\in \R^3\diagdown
\left\{ 0\right\} ,\ \eps >0  \label{S4E4}
\end{equation}%
where $\Phi \left( s\right) $ is compactly supported.

The force fields with the form (\ref{S4E4}) will yield a kinetic limit as $%
\eps\rightarrow0.$ Their key feature is that they yield particle
deflections of order one if the interacting particles approach at distances
of order $\eps.$

\begin{remark}
It would be possible to make less restrictive assumptions on $\Phi ,$ for
instance that $\Phi $ decreasing exponentially or $\Phi \left( s\right) \sim 
\frac{C}{s^{a}}$ as $s\rightarrow \infty $ with $a>1$ (cf. \cite{NSV}). (The
critical exponents depend on the space dimension).
\end{remark}

\paragraph{Landau random force fields.}

The families of Landau random force fields are characterized by the fact
that the interactions generated by each individual particle are small and
tend to zero as $\eps \rightarrow 0.$ We will say that a family of
random force fields with anyone of the forms (\ref{S5E7}) or (\ref{S5E8}) is
a family of Landau random force fields if the function $\Phi _{\eps }$
has the form:%
\begin{equation}
\Phi _{\eps }\left( x\right) =\eps \Phi \left( \frac{%
\left\vert x\right\vert }{L_{\eps }}\right) \ \ ,\ \ x\in \mathbb{R}%
^{3},\ \eps >0  \label{S4E6}
\end{equation}%
where $\Phi \in C^{2}\left( \R^3\right) .$ We can choose the
characteristic length $L_{\eps }$ in one of the two possible ways.
Either:%
\begin{equation}
L_{\eps }\gg 1  \label{S4E7}
\end{equation}%
or%
\begin{equation}
L_{\eps }\sim 1\ \ \text{or \ }L_{\eps }\ll 1 \ . \label{S4E7b}
\end{equation}

In order to be able to define the random force fields $F_{\eps }$ by
means of (\ref{S5E7}) or (\ref{S5E8}) (or in the more precise forms (\ref%
{S9E9}), (\ref{T8E2}) in Proposition \ref{RandForFieldMovPart}) we need to
assume a sufficiently fast decay of $\Phi \left( s\right) $ as $s\rightarrow
\infty $ or to impose suitable electroneutrality conditions if, say (\ref%
{phiDec_sOne}) holds.

\bigskip

We have seen in \cite{NSV} how to obtain a kinetic limit for tagged
particles moving in stationary random force fields. Using analogous
arguments it is possible to derive kinetic equations describing the
evolution of a tagged particle moving in time-dependent random force fields
with the form (\ref{S5E7}) or (\ref{S5E8}) and (\ref{S4E6}) in the limit $%
\eps \rightarrow 0.$ Some additional assumptions on $L_{\eps }$
are also required, in order to obtain uncorrelated velocity deflections over
distances of the order of the mean free path. The case (\ref{S4E7b})
corresponds to the so-called grazing collisions limit. In this case,
although the limit kinetic equation is a Landau equation, we can interpret
the dynamics as a sequence of weak, Boltzmann-like, binary collisions. In
the case in which $L_{\eps }\gg 1$ the resulting limit kinetic
equation is a linear Landau equation.

\subsection{Scatterer distributions in Rayleigh gases \label{ScattRayl}}

\bigskip

In the case of the dynamics (\ref{eq:RayNew}) the evolution of the tagged
particle cannot be described only by means of the action of the random force
field generated by the scatterers. This can be seen most clearly is the case
in which the mass of the tagged particle is the same as the mass of the
scatterers (or comparable) and the interaction between the tagged particle
and the scatterers are as in the case of the Boltzmann random force fields
(cf. (\ref{S4E4})). Also in the case of Landau interaction potentials (cf. \ref{S4E6}) the
dynamics of the tagged particle depends on whether the scatterers are
affected by the tagged particle or not. This may be seen by studying the
binary interaction (cf.~\cite{NVW}, Section 5.2). 

\bigskip

The information that we need to keep about the system of scatterers in order
to compute the evolution of a tagged particle which evolves according to (%
\ref{eq:RayNew}) is the whole configuration of positions and velocities of the
scatterers at any time (i.e. $\left\{ \left( X_{k}\left( t\right)
,V_{k}\left( t\right) \right) \right\} _{k\in\mathbb{N}}$ as well as their
charges, if there are different types of particles, and the form of the
potentials yielding the interactions. If the tagged particle and the
scatterers have a different mass, the ratio between these masses would also
play a role in the determination of the dynamics. We could also consider
different types of scatterers having different masses. The distribution of
velocities of the different scatterers then plays also a role in the
dynamics of the tagged particle $\left( X\left( t\right) ,V\left( t\right)
\right) .$

\bigskip

\section{Computation of the kinetic timescale and limit equation}\label{sec:TimeScale}

As indicated in the Introduction, the paper \cite{NSV} describes how to
obtain kinetic limits for the dynamics of tagged particles moving in a field
of fixed scatterers (Lorentz Gases). There are basically two different types of kinetic
limits, namely the ones leading to the Boltzmann and the Landau equations. It has been seen in \cite{NSV} that it
is possible to define two time scales $T_{BG}$ and $T_{L}$ which are the
characteristic times for the velocity of the tagged particle to experience a
deflection of order one due to a binary collision and to the accumulation of
many small random deflections due to weak interactions with many particles
respectively.

\smallskip

In the case of more complicated particle dynamics, like (\ref{eq:IntNew}) or (\ref%
{eq:RayNew}) we can also define kinetic time scales. We consider first the case
of a tagged particle moving in a Rayleigh gas, i.e. the dynamics of the
tagged particle is described by means of (\ref{eq:RayNew}). In the case of
systems described by means of (\ref{eq:IntNew}) we will compute the kinetic time
scales approximating the dynamics of a tagged particle by means of (\ref%
{eq:TagFric}). This approximation will be obtained in Section \ref{ss:JustR}.
The kinetic time scale associated to equations with the form (\ref{eq:TagFric})
can be readily obtained and this will provide a method to obtain the kinetic
time scale for systems described by (\ref{eq:IntNew}).

In the case of tagged particles moving in a Rayleigh gas (cf.~(\ref{eq:RayNew}))
we can define a characteristic kinetic time and also the mean free path for
a particle moving in a Rayleigh gas as follows. Suppose that $\left(
X,V\right) $ is the position and velocity of a tagged particle whose
dynamics is given by (\ref{eq:RayNew}) where the initial velocity distribution of
the scatterers is given by the measure $g$ in the case of particles of a
single type, or a set of measures $\left\{ g_{\ell }\right\} _{\ell =1}^{L}$
in the case of particles of different types as discussed in Subsection \ref%
{ScattRayl}. We define the mean free path associated to the family of
potentials $\Phi _{\eps }$ and to the set of particle distributions
as the value $T_{\eps }$ such that:%
\begin{equation}
\mathbb{E}\left[ \left\vert V_{\eps }\left( T_{\eps }\right)
-V_{\eps }\left( 0\right) \right\vert ^{2}\right] =\frac{1}{4}%
\left\vert V_{\eps }\left( 0\right) \right\vert ^{2}  \label{MFP}
\end{equation}%
where we assume that $\left\vert V_{\eps }\left( 0\right) \right\vert 
$ is the characteristic speed of the tagged particle.

\bigskip

It has been seen in \cite{NSV} how to estimate $T_{\eps}$ in the case
of static scatterers, splitting the random force field in the sum of a
Boltzmann part and a Landau part. We can then define a Boltzmann-Grad limit
time scale $T_{BG}$ by means of (\ref{MFP}) for the Rayleigh gas associated
to the Boltzmann part of the potential and the Landau time scale $T_{L}$
which is computed applying (\ref{MFP}) to the Landau part of the potential.
In most of the interaction potentials considered in \cite{NSV} one of the
two time scales is much larger than the other as $\eps\rightarrow0.$
Then, given that in the kinetic limit the deflections are additive we obtain:%
\begin{equation*}
T_{\eps}\sim\min\left\{ T_{BG},T_{L}\right\} \text{ as }%
\eps\rightarrow0.
\end{equation*}

A similar decomposition in Boltzmann part and Landau part can be made in the
case of arbitrary Rayleigh gases. We can then use the same approach to
estimate $T_{\eps }$ as in \cite{NSV} in the kinetic limit $%
\eps \rightarrow 0.$ We then recall the decomposition of the
interaction potentials obtained in \cite{NSV}.

Suppose that $\Phi _{\eps }$ is an interaction potential as it appears in (\ref{eq:RayNew}). The potentials
for which there is a Boltzmann part are characterized by the existence of a
collision length $\lambda _{\eps }$ satisfying $\lim_{\eps
\rightarrow 0}\lambda _{\eps }=0$ and such that:%
\begin{equation*}
\lim_{\eps \rightarrow 0}\Phi _{\eps }\left( \lambda
_{\eps }y\right) =\Psi \left( y\right)
\end{equation*}%
uniformly in compact sets of $\R^3\setminus \left\{ 0\right\} ,$
where $\Psi \left( y\right) \neq 0.$ The distance $\lambda _{\eps }$
is the characteristic distance at which the tagged particle experiences
deflections of order one due to the interaction with one of the scatterers.

We then split $\Phi _{\eps }$ as:%
\begin{equation}
\phi _{B,\eps }\left( x\right) =\Phi _{\eps }\left( x\right)
\eta \left( \frac{\left\vert x\right\vert }{M\lambda _{\eps }}\right)
\ \ ,\ \ \phi _{L,\eps }\left( x\right) =\Phi _{\eps }\left(
x\right) \left[ 1-\eta \left( \frac{\left\vert x\right\vert }{M\lambda
_{\eps }}\right) \right]  \label{S1E5}
\end{equation}%
where $\eta \in C^{\infty }\left( \R^3\right) $ is a radially
symmetric cutoff function satisfying $\eta \left( y\right) =1$ if $%
\left\vert y\right\vert \leq 1,\ \eta \left( y\right) =0$ if $\left\vert
y\right\vert \geq 2,\ 0\leq \eta \leq 1.$ We assume that $M$ is a large
number, which eventually might be sent to infinity at the end of the
argument. It would be also possible to take $M=M_{\eps }$ with $%
M_{\eps }\rightarrow \infty $ as $\eps \rightarrow 0$ at a
sufficiently slow rate to control the transition region between Boltzmann
and Landau collisions. In (\ref{S1E5}), $\phi _{B,\eps }$ stands for
Boltzmann and $\phi _{L,\eps }$ for Landau.

\bigskip

We now compute the characteristic time scales $T_{BG}$ and $T_{L}$ which
yield the expected time to have velocity deflections comparable to the
tagged particle velocity itself for the Boltzmann part of the interaction $%
\phi _{B,\eps }$ and the Landau part of the interaction $\phi
_{L,\eps }$ for a tagged particle moving in a Rayleigh gas. We notice that the time scale for the change of the velocity distribution for a Lorentz Gas with moving scatterers is the same time scale of the Rayleigh Gas, in the case in which the mass of the tagged particle and the scatterers  are comparable.  For this reason, to compute the kinetic time scale for the Rayleigh Gas it is enough to compute the time scale for a Lorentz Gas with moving scatterers. This is the strategy we will use in the sections below. 

\bigskip

\subsection{Computation of $T_{BG}$ for a tagged particle in a Rayleigh gas}

\bigskip

We first compute the characteristic time scale to obtain relevant
deflections for a tagged particle moving in a Rayleigh gas generated by
potentials with the form $\phi _{B,\eps }$ in (\ref{S1E5}). To
compute the deflection time using the definition (\ref{MFP}) would require
some involved computations, since it would require to compute the
contributions to the deflection of having multiple collisions, something
that would require some tedious computations. Therefore, instead of using
the definition (\ref{MFP}) we will compute a simpler quantity, namely the
length of a tube in which the probability of finding one particle is of
order one.

\bigskip

We will restrict ourselves to the case in which the velocity of the tagged
particle experiences deflections of order one at distances of order $%
\eps $ of one scatterer. Suppose that a tagged particle moves at
speed $V.$ We assume that the initial positions of the scatterers $x_{j}$,
as well as their velocities $v_{j}$, denoted as $\omega =\left\{ \left(
x_{k},v_{k}\right) \right\} _{k\in J}$ are determined by means of the
Poisson measure defined in (\ref{S4E2}). Suppose without loss of generality
that the initial position of the tagged particle is the origin. Since its
initial velocity is $V$ we have that the position of the tagged particle at
time $t$ is $Vt.$ We then define $T_{BG}$ in the following way:%
\begin{equation*}
T_{BG}=\inf \left\{ \tau :\mathbb{P}\left( \exists k:\min_{0\leq s\leq \tau
}\left\vert Vs-\left( x_{k}+v_{k}s\right) \right\vert \leq \eps
\right) \geq \frac{1}{2}\right\} .
\end{equation*}%
To compute $T_{BG}$, we compute the probability above for each $\tau >0$. To
this end, we decompose the velocity space ${\mathbb{R}}^{3}$ into disjoint
cubes $Q_{\ell }$: 
\begin{equation*}
{\mathbb{R}}^{3}=\cup _{\ell \in {\mathbb{N}}}Q_{\ell },\quad Q_{\ell
}=\{v\in {\mathbb{R}}^{3}:x=c_{\ell }+p,p\in \lbrack 0,{\eps }/\tau
)^{3}\},
\end{equation*}%
for some appropriately chosen centers $c_{\ell }\in {\mathbb{R}}^{3}$. Then
for $\tau >0$ we have 
\begin{align*}
\mathbb{P}\left( \exists k:\min_{0\leq s\leq \tau }\left\vert Vs-\left(
x_{k}+v_{k}s\right) \right\vert \leq \eps \right) & =\sum_{\ell \in {%
\mathbb{N}}}\mathbb{P}\left( \exists k:\min_{0\leq s\leq \tau }\left\vert
Vs-\left( x_{k}+v_{k}s\right) \right\vert \leq \eps ,v_{k}\in Q_{\ell
}\right) \\
& \sim \sum_{\ell \in {\mathbb{N}}}\mathbb{P}\left( \exists k,s\in \lbrack
0,\tau ]:|x_{k}-s(V-c_{\ell })|\leq {\eps },v_{k}\in Q_{\ell }\right)
\\& \sim \tau \lambda _{\eps }^{2},
\end{align*}%
where $\lambda _{\eps }$ is the collision length. Therefore, the
Boltzmann-Grad timescale $T_{BG}$ is given by 
\begin{equation*}
T_{BG}\sim \lambda _{\eps }^{-2},
\end{equation*}%
as in the Lorentz gas with fixed obstacles, see \cite{NSV}. We notice that in the case of potentials of the form \eqref{S4E4}, namely $\Phi _{\eps }\left( x\right) =\Phi \left( \frac{\left\vert x\right\vert }{\eps }\right)$ with $\Phi \left( s\right)$ compactly supported, the collision length $\lambda_{\eps}=\eps$.

\bigskip

\subsection{Computation of $T_{L}$ for a tagged particle in a Rayleigh gas \label{CharTimeRayl}}

\bigskip

We now compute the characteristic time in which the deflections of the
velocity are of order one for one tagged particle moving in a time-dependent
Landau random force field with the form (\ref{S5E7}) (or (\ref{S5E8}))
generated by potentials with the form (\ref{S4E6}) with $\Phi \left(
r\right) \sim \frac{1}{r^s}$ as $r\rightarrow \infty$, $s\geq 1$. To this end we compute
the variance of the velocity deflections experienced by a particle moving in
a straight line during the time $T.$ This deflection is given by:%
\begin{equation}\label{def:deflection}
d_{\eps }\left( T;\omega \right) =\int_{0}^{T}F_{\eps }\left(
Vt,t;\omega \right) dt
\end{equation}%
where we assume that the initial position of the tagged particle is $X=0$
and the velocity is $V.$ We have $\mathbb{E}\left[ d_{\eps }\left(
T\right) \right] =0.$ We now compute:%
\begin{equation}
\mathbb{E}\left[ d_{\eps }\left( T\right) \otimes d_{\eps
}\left( T\right) \right] =\int_{0}^{T}dt_{1}\int_{0}^{T}dt_{2}\mathbb{E}%
\left[ F_{\eps }\left( Vt_{1},t_{1};\omega \right) \otimes
F_{\eps }\left( Vt_{2},t_{2};\omega \right) \right] \ . \label{T8E3}
\end{equation}

In order to indicate how to compute (\ref{T8E3}) we consider the case of
random force fields (\ref{S5E7}) with $\Phi =\Phi _{\eps }$ given by:%
\begin{equation*}
\Phi _{\eps }\left( x\right) =\eps \Phi \left( \frac{x}{%
L_{\eps }}\right)
\end{equation*}%
where $\Phi \left( \xi \right) $ decreases sufficiently fast as $\left\vert
\xi \right\vert \rightarrow \infty ,$ for instance as in (\ref{phiDec_slarge}%
). In order to compute (\ref{T8E3}) we approximate it assuming that we have $%
N$ particles independently distributed in $B_{R_{N}}\left( 0\right) \times 
\R^3\subset \ \mathbb{R}^{6}$ with $N=\frac{4\pi \left(
R_{N}\right) ^{3}}{3}$ with the probability density $\frac{1}{\left\vert
B_{R_{N}}\left( 0\right) \right\vert }g\left( v\right) .$ (It would be also
possible to use a macrocanonical distribution, but the result is equivalent
in the limit $N\rightarrow \infty $). We will denote as $d_{\eps
}^{N}\left( T\right) $ the corresponding deflections and:
\begin{equation*}
F\left( x,\tau ;\omega ;g\right) =-\sum_{k=1}^{N}\nabla \Phi _{\eps
}\left( x-x_{k}-v_{k}\tau \right)\ .
\end{equation*}
We now compute \eqref{T8E3} and using the equation above we obtain
\begin{align*}
&\mathbb{E}[ d_{\eps }^{N}( T) \otimes d_{\eps }^{N}\left(
T\right) ] =\sum_{k,j=1}^{N}\int_{0}^{T}dt_{1}\int_{0}^{T}dt_{2}%
\mathbb{E}\left[ \nabla \Phi _{\eps }\left(
Vt_{1}-x_{k}-v_{k}t_{1}\right) \otimes \nabla \Phi _{\eps }\left(
Vt_{2}-x_{j}-v_{j}t_{2}\right) \right] \\
&=\sum_{k=1}^{N}\int_{0}^{T}dt_{1}\int_{0}^{T}dt_{2}\mathbb{E}\left[ \nabla
\Phi _{\eps }\left( Vt_{1}-x_{k}-v_{k}t_{1}\right) \otimes \nabla
\Phi _{\eps }\left( Vt_{2}-x_{k}-v_{k}t_{2}\right) \right] \\
&=\frac{N}{\left\vert B_{R_{N}} \right\vert }\int_{0}^{T}dt_{1}%
\int_{0}^{T}dt_{2}\int_{B_{R_{N}} }dy\int_{\R^3}dvg\left( v\right)
\nabla \Phi _{\eps }\left( \left( V-v\right) t_{1}-y\right) \otimes
\nabla \Phi _{\eps }\left( \left( V-v\right) t_{2}-y\right) \\
&=\int_{0}^{T}dt_{1}\int_{0}^{T}dt_{2}\int_{B_{R_{N}}\left( 0\right)
}dy\int_{\R^3}dvg(v) \nabla \Phi _{\eps }\left( \left(
V-v\right) t_{1}-y\right) \otimes \nabla \Phi _{\eps }\left( \left(
V-v\right) t_{2}-y\right).
\end{align*}
Taking the limit $N\rightarrow \infty $ we obtain:%
\begin{align*}
&\mathbb{E}[ d_{\eps }( T) \otimes d_{\eps }\left( T\right) ]
\\&=\int_{0}^{T}dt_{1}\int_{0}^{T}dt_{2}\int_{\R^3}dy\int_{\mathbb{R}%
^{3}}dvg\left( v\right) \nabla \Phi _{\eps }\left( \left( V-v\right)
t_{1}-y\right) \otimes \nabla \Phi _{\eps }\left( \left( V-v\right)
t_{2}-y\right) \\
&=\frac{\eps ^{2}}{\left( L_{\eps }\right) ^{2}}%
\int_{0}^{T}dt_{1}\int_{0}^{T}dt_{2}\int_{\R^3}dy\int_{\mathbb{R}%
^{3}}dvg\left( v\right) \nabla _{\xi }\Phi \left( \frac{y-\left( V-v\right)
\left( t_{1}-t_{2}\right) }{L_{\eps }}\right) \otimes \nabla _{\xi
}\Phi _{\eps }\left( \frac{y}{L_{\eps }}\right) \\
&=\eps ^{2}\left( L_{\eps }\right)
^{2}\int_{0}^{T}dt_{1}\int_{-\frac{t_{1}}{L_{\eps }}}^{\frac{T-t_{1}}{%
L_{\eps }}}K\left( V;\tau \right) d\tau
\end{align*}%
where:%
\begin{equation*}
K\left( V;\tau \right) =\int_{\R^3}d\xi \int_{\mathbb{R}%
^{3}}dvg\left( v\right) \nabla _{\xi }\Phi \left( \xi +\left( V-v\right)
\tau \right) \otimes \nabla _{\xi }\Phi \left( \xi \right).
\end{equation*}

We will now assume that $T\gg L_{\eps }.$ Then, for most of the
values of $t_{1}$ in the integral we have that $\frac{T-t_{1}}{%
L_{\eps }}\gg 1$ and $\frac{t_{1}}{L_{\eps }}\gg 1.$ On the
other hand, the function $K\left( V;\tau \right) $ is integrable in $\tau $
under the assumption (\ref{phiDec_slarge})-(\ref{phiDec_sInt}), namely for $s>1$. There are some boundary effects
in the integrals if $t_{1}$ or $\left( T-t_{1}\right) $ are of order $%
L_{\eps }.$ We then obtain the approximation:%
\begin{equation*}
\mathbb{E}\left[ d_{\eps }\left( T\right) \otimes d_{\eps
}\left( T\right) \right] \sim \eps ^{2}\left( L_{\eps }\right)
^{2}T\int_{-\infty }^{\infty }K\left( V;\tau \right) d\tau.
\end{equation*}
Hence, the characteristic time scale is 
$T=\frac{1}{\eps ^{2}\left( L_{\eps }\right) ^{2}}$ 
and the diffusion coefficient for the tagged particle is the matrix 
$
D\left( V\right) =\int_{-\infty }^{\infty }K\left( V;\tau \right) d\tau
$ 
which is diagonal if $g$ is isotropic.  
Notice that in order to obtain a
dynamics without correlations (and also in order to have self-consistency of
the previous argument), we need to have $T\gg L_{\eps },$ i.e.:%
\begin{equation*}
\eps ^{2}\left( L_{\eps }\right) ^{3}\ll 1\ .
\end{equation*}

It is interesting to compute the characteristic time in the case of
potentials decreasing as Coulombian potentials $\Phi(r)\sim \frac 1 r$ as $r\to\infty$. In this case we need to take
at least two types of charges in order to have electroneutrality. We
consider a family of rescaled potentials $\Phi_\ep(x)=\ep \Phi(x)$ and random force fields that are the limit as $N\rightarrow \infty $ of
fields with the form:%
\begin{equation*}
F\left( x,\tau ;\omega ;g\right) =-\eps \sum_{k=1}^{N}\nabla \Phi
\left( x-x_{k}-v_{k}\tau \right) +\eps \sum_{k=1}^{N}\nabla \Phi
\left( x-y_{k}-w_{k}\tau \right)
\end{equation*}%
where the particles $\left\{ \left( x_{k},v_{k}\right) \right\} $ and $%
\left\{ \left( y_{k},w_{k}\right) \right\} $ are chosen independently and $%
\Phi $ is a smooth function satisfying:%
\begin{equation*}
\Phi \left( x\right) \sim \frac{1}{\left\vert x\right\vert }\text{ as }%
\left\vert x\right\vert \rightarrow \infty \text{ and }\nabla \Phi \left(
x\right) \sim -\frac{x}{\left\vert x\right\vert ^{3}}\text{ as }\left\vert
x\right\vert \rightarrow \infty.
\end{equation*}

We can then compute the variance of the deflections arguing as above. We
then obtain:%
\begin{align*}
\mathbb{E}\left[ d_{\eps }\left( T\right) \otimes d_{\eps
}\left( T\right) \right] 
&=2\eps ^{2}\int_{0}^{T}dt_{1}\int_{-t_{1}}^{T-t_{1}} ds K\left(
V;s\right),
\end{align*}%
where:%
\begin{equation*}
K\left( V;s\right) =\int_{(\mathbb{R}^3)^2}dy dv g\left( v\right) \nabla
\Phi \left( y+\left( v-V\right) s\right) \otimes \nabla \Phi \left( y\right).
\end{equation*}

Notice that the function $K\left( V;s\right) $ is well defined for each $%
s\in \mathbb{R}$. We can compute the asymptotics of $K\left( V;s\right) $ as 
$\left\vert s\right\vert \rightarrow \infty .$ We can assume without loss of
generality that $s>0$ since the case $s<0$ would be similar. We rescale $%
y=s\xi$, then: 
\begin{equation*}
K\left( V;s\right) =s^{3}\int_{\R^3}d\xi \int_{\mathbb{R}%
^{3}}dvg\left( v\right) \nabla _{y}\Phi \left( s\left( \xi +\left(
v-V\right) \right) \right) \otimes \nabla \Phi \left( s\xi \right).
\end{equation*}

Taking the limit as $s\rightarrow \infty $ and using the asymptotics of $%
\nabla \Phi \left( y\right) $ we obtain:
\begin{equation}
K\left( V;s\right) \sim \frac{1}{s}\int_{\R^3}d\xi \int_{\mathbb{R}%
^{3}}dvg\left( v\right) \frac{\left( \xi +\left( v-V\right) \right) }{%
\left\vert \xi +\left( v-V\right) \right\vert ^{3}}\otimes \frac{\xi }{%
\left\vert \xi \right\vert ^{3}}\ .  \label{Kdef}
\end{equation}

This integral exists. Indeed, the integrals on $\xi $ are well defined if $%
v\neq V.$ On the other hand, if $v\rightarrow V$ we obtain a divergence like 
$\frac{1}{\left\vert v-V\right\vert }.$ We will assume that $g\left(
v\right) $ is smooth enough to have integrability of $\int_{\R^3}dv%
\frac{g\left( v\right) }{\left\vert v-V\right\vert }.$ Then the integral is
well defined and we obtain the asymptotics:%
\begin{equation}
K\left( V;s\right) \sim \frac{\Gamma \left( V,\func{sgn}\left( s\right)
\right) }{s}\text{ as }\left\vert s\right\vert \rightarrow \infty
\label{Gammadef}
\end{equation}%
for some function $\Gamma :\R^3\times \left\{ -1,+1\right\}
\rightarrow M_{3}\left( \R^3\right) .$ We then obtain:%
\begin{equation}
\mathbb{E}\left[ d_{\eps }\left( T\right) \otimes d_{\eps
}\left( T\right) \right] \sim 2\eps
^{2}\int_{0}^{T}dt_{1}\int_{-t_{1}}^{T-t_{1}}ds \frac{\Gamma \left( V,\func{sgn}%
\left( s\right) \right) }{\left\vert s\right\vert } \sim D\left( V\right)
\eps ^{2}T\log \left( T\right)  \label{DiffRay}
\end{equation}%
where $D\left( V\right) \in M_{3}\left( \R^3\right) $ is a
nonnegative matrix.

\bigskip

This yields the characteristic time scale:%
\begin{equation}
T_L
\sim \frac{1}{\eps ^{2}\log \left( \frac{1}{%
\eps }\right) }  \label{TimeCoulLog}
\end{equation}%
which is exactly the same time scale obtained in \cite{NSV} for random force
fields generated by static distributions of particles (Lorentz gases).  
The logarithmic term $\log \left( \frac{1}{\eps }\right) $ is a well
known correction appearing in systems with Coulombian interactions. It is
usually termed as Coulombian logarithm (cf.~\cite{LL2}).

%
%

\bigskip


\section{Kinetic limits for approximate models \label{KineLimit}}

In this Section we show how to approximate the dynamics of the systems \textcolor{blue}{\eqref{eq:IntNew}-}\eqref{eq:RayNew} by means of an equation of the form (\ref{eq:TagFric}), at
least during small macroscopic times in suitable scaling limits. The method
can be applied in systems in which the range of the potential is much larger
than the average distance between particles. The rationale behind the method
is to decompose the force made by the scatterers on the tagged particle in a
friction term, which is the reaction to the force made by the tagged
particle on the scatterers and yields their deflection, plus a random force
term which is the sum of the forces produced by the scatterers in the tagged
particle assuming that they are not affected by it. The friction term is due
to the fact that the tagged particle is moving against the mean velocity of
the surrounding medium.

We will restrict our analysis
to three types of potentials, namely weak potentials with a finite range
much larger than the distance between particles, potentials behaving for
large distances as Coulomb potentials and finally the case of so-called
``grazing collisions" in which the interactions between particles are very
weak and have a range smaller or similar to the particle distance.

\subsection{Derivation of the kinetic equations for Rayleigh gases and interacting particle systems by means of the evolution of a tagged particle} \label{KinEquInteractCase}

Our strategy for obtaining, in a suitable scaling limit, kinetic equations from particle systems, both for Rayleigh gases and interacting particle systems, consists in approximating the evolution of a tagged particle by:
\begin{equation} \label{eq:Langevin}
\frac{dx}{dt}=w\ ,\ \ \ \frac{dw}{dt}=\eta_g \left( t;w\right) -\Lambda
_{g}\left( w\right) \ , 
\end{equation}
where $\eta_g$ is a Gaussian noise with
\begin{align}
	\E[\eta_g(t,w)]&= 0,\\
	\E [\eta_g(t_1,w)\eta_g(t_2,w)] &= D_g(w) \delta(t_1-t_2).
\end{align}
The function $g$ describes the velocity distribution of scatterers in the Rayleigh gas case, and the initial one-particle function in the case of interacting particle systems. 
Therefore, if we denote the probability density describing the distribution
of the tagged particle as $f\left( x,w,t\right) $ we obtain the following
evolution equation for it%
\begin{equation}
\partial _{t}f+w\cdot \partial _{x}f=\partial _{w}\cdot \left( \frac{1}{2}%
D_{g}\left( w\right) \partial _{w}f+\Lambda _{g}\left( w\right) f\right) .
\label{KinFastDec}
\end{equation}

In the case of interacting particle systems, an approximation of the form \eqref{eq:Langevin} is valid only for infinitesimal macroscopic times with $g=f_0$.
For positive macroscopic times, we need to take into account the change of the one-particle distribution function $f(t,v)$. The kinetic equation can then be obtained through the closure assumption $g(t,x,v)=f(t,x,v)$. Therefore, the kinetic equation reads:
\begin{equation}
\partial _{t}f+w\cdot \partial _{x}f=\partial _{w}\cdot \left( \frac{1}{2}%
D_{f}\left( w\right) \partial _{w}f+\Lambda _{f}\left( w\right) f\right) .
\label{Nonlinear}
\end{equation}
Notice however, that the dependence of $\Lambda_g$ and $D_g$ in the case of interacting particle systems is in general different from the form of the same operators for the Rayleigh gas.

\paragraph{Rayleigh gases}

We first present the kinetic equations and the corresponding scaling
limits that we will obtain in this Section for a
tagged particle evolving according to the equations (\ref{eq:RayNew}) (i.e.~Rayleigh gases).

\smallskip 

\begin{itemize}
\item In the case of potentials with the form%
\begin{equation}\label{pot1}
\Phi _{\eps }\left( x\right) =\eps \Phi \left( \frac{%
\left\vert x\right\vert }{L_{\eps }}\right) \ \ \text{with\ \ }%
L_{\eps }\gg 1
\end{equation}%
and $\Phi \left( s\right) $ decreasing faster than $\frac{1}{s^{\alpha }}$
with $\alpha >1$ (cf. (\ref{S4E6}), (\ref{S4E7})) we have seen that the
distribution of particle velocities $f\left( x,w,t\right) $ evolves
according to the equation 
\begin{equation*}
\partial _{t}f+v\cdot \partial _{x}f=\partial_{v}\cdot \left( \frac{1}{2}%
D_{g}\left( v\right) \partial _{v}f+\Lambda _{g}\left( v\right) f\right) 
\end{equation*}%
where the precise form of the diffusion coefficient $D_{g}\left( v\right) $ and the friction term $\Lambda _{g}\left(
v\right) $ will be given in {Theorem \ref{thm:FiniteRange}}. The formula relating the microscopic $%
\tilde{t}$ and macroscopic variable $t$ is $t=\left( \theta _{\eps
}\right) ^{2}\tilde{t}$ with $\theta _{\eps }=\eps \left(
L_{\eps }\right) ^{\frac{3}{2}}.$
\end{itemize}

\bigskip 

\begin{itemize}
\item In the case of Coulombian potentials, i.e potentials with one of the
forms%
\begin{equation}
\Phi _{\eps }\left( x\right) =\Phi \left( \frac{\left\vert
x\right\vert }{\eps }\right) \ \ \text{or\ \ }\phi _{\eps
}\left( x\right) =\eps \Phi \left( \left\vert x\right\vert \right) 
\label{CoulPot}
\end{equation}%
with $\Phi \left( s\right) \sim \frac{1}{s}$ as $s\rightarrow \infty $ (cf. (%
\ref{S8E5})) we obtain the following kinetic equation  
\begin{equation*}
\partial _{t}f+v\cdot \partial _{x}f=\partial _{v}\cdot \left( \frac{1}{2}%
D_{g}\left( v\right) \partial _{v}f+\Lambda _{g}\left( v\right) f\right) 
\end{equation*}%
where the precise form of the diffusion coefficient $D_{g}\left( v\right) $ and the friction term $\Lambda _{g}\left(
v\right) $ will be given in Theorem \ref{thm:Coulomb}. The relation between the microscopic time
scale $\tau $ and the macroscopic time scale is given by $\tau =\frac{t}{%
2\eps ^{2}\log \left( \frac{1}{\eps }\right) }.$
\end{itemize}

\begin{itemize}
\item In the case of potentials with the form (weak coupling case)%
\begin{equation}
\Phi _{\eps }\left( x\right) =\eps \Phi \left( \frac{%
\left\vert x\right\vert }{\ell _{\eps }}\right) \ \ ,\ \ \ell
_{\eps }\lesssim 1  \label{WeakCoupling}
\end{equation}%
where $\Phi \left( s\right) $ decreases faster than $\frac{1}{s^{\alpha }}$
with $\alpha >1$ we can approximate the distribution of
particle velocities $f$ by means of the solutions of the kinetic equation 
\begin{equation*}
\partial _{t}f+v\cdot \partial _{x}f=\partial _{v}\cdot \left( \frac{1}{2}%
D_{g}\left( v\right) \partial _{v}f+\Lambda _{g}\left( v\right) f\right) 
\end{equation*}%
where the precise form of the diffusion coefficient $D_{g}\left( v\right) $ and the friction term $\Lambda _{g}\left(
v\right) $ will be given in Theorem \ref{thm:Grazing}.
\end{itemize}

\paragraph{Interacting particle systems}

We now present the precise scaling limit in which a system of particles
described by the system of equations (\ref{eq:IntNew}) can be approximated by
means of Landau and Balescu-Lenard equations.
\smallskip 

\begin{itemize}
	\item In the case of interaction potentials with the form%
	\begin{equation*}
	\Phi _{\eps }\left( x\right) =\eps \Phi \left( \frac{%
		\left\vert x\right\vert }{L_{\eps }}\right) \ \ \text{with\ \ }%
	\eps \left( L_{\eps }\right) ^{3}\rightarrow 1
	\end{equation*}%
	and $\Phi \left( s\right) $ decreasing faster than $\frac{1}{s^{\alpha }}$
	with $\alpha >1$ (cf. (\ref{S4E6}), (\ref{S4E7})) the kinetic equation
	yielding the evolution of the distribution function $f$ is given by the
	Balescu-Lenard equation 
	\begin{equation*}
	\partial _{t}f+v\cdot \partial _{x}f=\partial _{v}\left( \frac 1 2 D_{f}\left(
		v\right) \partial _{v}f +\Lambda_{f}\left( v\right) f\right) 
	\end{equation*}%
	where the precise form of the diffusion coefficient $D_{f}\left( v\right)$ and the friction term $\Lambda_{f}\left(v\right)$ will be given in Theorem  				\ref{thm:FiniteRange} with $f=g$. The formula relating the
	microscopic time scale $\tilde{t}$ with the macroscopic time scale $t$ is $%
	\tilde{t}=\eps ^{2}\left( L_{\eps }\right) ^{3}t.$ 
\end{itemize}

\begin{itemize}
	\item In the case of Coulombian potentials as in (\ref{CoulPot}) (cf. also (%
	\ref{S8E5})) the evolution of the distribution function $f$ is given by the
	Landau equation 
	\begin{equation*}
	\partial _{t}f+v\cdot \partial _{x}f=\partial _{v}\left(\frac{1}{2} D_{f}\left(
	v\right) \partial _{v}f+\Lambda_{f}\left( v\right) f\right) 
	\end{equation*}%
	 where the precise form of the diffusion coefficient $D_{f}\left( v\right)$ and the friction term $\Lambda_{f}\left(v\right)$ will be given in Theorem \ref{thm:Coulomb} with $f=g$. 
	 The microscopic time scale $%
	\tau $ and the macroscopic time scale are related by means of the formula $%
	\tau =\frac{t}{2\eps ^{2}\log \left( \frac{1}{\eps }\right) }.$ The fact that the kinetic equation for Coulombian potentials is the Landau equation instead of the Balescu-Lenard equation, due to the presence of the Coulombian logarithm, has been originally noticed by Landau and Lenard (cf.~\cite{LL2, Le}).
\end{itemize}

\begin{itemize}
	\item We consider one version of the weak coupling case with potentials of
	the form (\ref{WeakCoupling}) with $\ell _{\eps }=1.$ The kinetic
	equation describing the evolution of the particle distribution $f$ is the
	following version of the Landau equation
	\begin{equation*}
	\partial _{t}f+v\cdot \partial _{x}f=\partial _{v}\left( \frac{1 }{2 }D_{f}(v) %
	\partial _{v}f+\Lambda_{f}\left( v\right) f\right) 
	\end{equation*}
	where the precise form of the diffusion coefficient $D_{f}\left( v\right)$ and the friction term $\Lambda_{f}\left(v\right)$ will be given in Theorem \ref{thm:Grazing} with $f=g$. 
	The relation between the microscopic time scale $\tilde{t}
	$ and the macroscopic time scale $t$ is $t=\eps ^{2}L_{\eps }%
	\tilde{t}.$
\end{itemize} 

\begin{remark}
	A few rigorous results deriving kinetic equations of Landau or Balescu-Lenard form for interacting particle systems have been recently obtained. In these results the particles interact by means of weak potentials with a range comparable or larger than the average particle distance. See for instance \cite{BPS, Du, DuS} for some partial results for the full particle system and the time of validity of the kinetic equation is shorter than the macroscopic time scale. In \cite{VW1, W2} the kinetic equation is derived for macroscopic times of order one but the particle system is approximated by a truncated BBGKY hierarchy.  
\end{remark}



\subsection{Main results}

We will now provide the precise form of the diffusion coefficient for a tagged particle in a Rayleigh gas and the diffusion coefficient for an interacting particle system. To this end, we consider approximate equations for the fluctuations of the particle density. The justification for this approximation will be given in Section~\ref{ss:JustR}. Both the approximation and the resulting diffusion coefficient depend on the scaling and choice of potential, therefore we consider separately potentials with finite range, Coulomb interaction, and grazing collision limits.

In the following, for a function or measure $h(x,v)$ on the phase space, we denote the
associated spatial density by $\rho[h]$ by: 
\begin{equation}
\rho(x)= \rho [ h]\left( x\right) =\int_{\mathbb{R}^{3}}h\left( x,w\right) dw.
\label{eq:Spdensity}
\end{equation}

\subsubsection{Particle interactions with finite range much larger than the
	particle distance } \label{sec:finRanTh}

Let $\theta>0$, we consider the following system that, as we will justify in Section \ref{ss:JustR}, is an approximation of the original particle systems for times scales much shorter than the mean free flight times. More precisely, we argue that on an intermediate scale $\tilde{t}=\frac{\tau }{L_{\eps }}.$ and  $\xi =\frac{X}{L_{\eps }}$ between the microscopic scale and the kinetic scale, the empirical measure can be approximated by a Gaussian process. A careful linearization then leads to a system of the form
\begin{equation}
\frac{d\xi }{d\tilde{t}}=V_0\ \ ,\ \ \frac{dV}{d\tilde{t}}=\theta F_g(\xi, \tilde{t};\omega)  \label{eq:MotionFiniteCombined},
\end{equation}
where
\begin{equation}
F_g(\xi, \tilde{t};\omega):=-\int_{\R^3}\nabla _{\xi }\Phi \left( \xi -\eta
\right) \tilde{\rho}\left( \eta ,\tilde{t}\right) d\eta .  \label{eq:ForceFiniteCombined}
\end{equation}
Here $\tilde{\rho}\left( \eta ,\tilde{t}\right)=\tilde{\rho}\left[\zeta\left( \cdot,\tilde{t}\right)\right](\eta)$ with $\zeta$ a Gaussian noise defined as follows.  
Let $N\left(\cdot,\cdot\right)$ be a Gaussian random field such that
\begin{equation}
\mathbb{E}\left[ N\left( y,w\right) \right] =0\ \ ,\ \ \mathbb{E}\left[
N\left( y_{a},w_{a}\right) N\left( y_{b},w_{b}\right) \right]
=g\left( w_{a}\right) \delta \left( y_{a}-y_{b}\right) \delta \left(
w_{a}-w_{b}\right).  \label{S7E1}
\end{equation}
An approximation of this form can (formally) be obtained both for particle systems of the form  \eqref{eq:RayNew} describing a Rayleigh Gas and interacting particle systems \eqref{eq:IntNew}.
The Gaussian noise $\zeta$ in the Rayleigh-case is given by
\begin{align}
(\partial _{\tilde{t}}+w\cdot \nabla_{y})\zeta \left( y,w,\tilde{t}\right)
-\theta \nabla _{y}\Phi \left( y-\xi \right) \cdot \nabla
_{w}g\left( w\right) & =0,\quad \zeta \left( y,w,0\right) =N\left(
y,w\right) , \label{eq:zetaFiniteRayleigh} 
\end{align}
whereas in the interacting case it is given by
\begin{equation}
\begin{aligned}
(\partial _{\tilde{t}}+w\cdot \nabla _{y})\zeta \left( y,w,\tilde{t}\right) -%
\left[ \theta \nabla
_{y}\Phi \left( y-\xi \right) +\sigma (\nabla _{y}\Phi \ast \tilde{\rho%
})\left( y,\tilde{t}\right) d\eta \right] \cdot \nabla _{w}g\left( w\right)
&=0,\\ \zeta \left( y,w,0\right) &=N\left(
y,w\right)  .
\end{aligned} \label{eq:zetaFiniteInteracting} 
\end{equation}

We obtain the following Theorem. 
\begin{theorem}\label{thm:FiniteRange} 
We consider the system \eqref{eq:MotionFiniteCombined}-\eqref{eq:ForceFiniteCombined}. Let the interaction
potential be a radial function $\Phi \left( s\right)\sim \frac{1}{s^{a}}$  with $a>1$ as $s\rightarrow \infty$ (cf.~\eqref{phiDec_slarge}, \eqref{phiDec_sInt}).  Under the following rescaling:
\begin{equation*}
t=\theta^{2}\tilde{t}\ \ ,\ \ x= \theta^{2}\xi\ \ ,\ \ v=V_0
\end{equation*}
and setting $F_g(\xi,\tilde{t})=\mathcal{F}_g(x,t)$ we then obtain for $t_1>0$, $t_2>0$
\begin{align*}
&\ \lim_{\theta\to 0}\mathbb{E}\left[ \mathcal{F}_g(vt, t;\omega)\right]=-\Lambda_{g}(v), \\  
&\ \lim_{\theta\to 0}\mathbb{E}\left[ (\mathcal{F}_{g}\left( v t_1,{t}_{1}\right)+\Lambda_{g}(v)) \otimes
(\mathcal{F}_{g}\left( v t_2, {t}_{2}\right)+\Lambda_{g}(v)) \right] 
=D_{g}\left( v\right) \delta \left( t_{1}-t_{2}\right) . 
\end{align*}
More precisely, 
\begin{itemize}
\item[(I)]{\textbf{Rayleigh Gas case}}\\ 
The diffusion coefficient $D_{g}$ and the friction term $\Lambda_g $ are given by:%
\begin{align}
D_{g}\left( v\right) &=\pi \int_{\R^3}\int_{\R^3}  \left( k\otimes k\right) |\hat{\Phi}(k)|^2 \delta(k\cdot (v-w)) g(w) \ud{w}\ud{k} ,  \label{eqthm:Rayleigh} \\
\Lambda_g(v)&= -\pi \int_{\R^{3}} \int_{\R^{3}} (k\otimes k) |\hat{\Phi}(k)|^2 \delta (k\cdot(v-w))\nabla g(w) \ud{w} \ud{k} 
\end{align}
\item[(II)]{\textbf{Interacting particle case}}\\
We assume that the function $G_{\sigma }\left( y,w,w_{0},\tilde{t}\right) $ defined as the solution of the problem \eqref{T2E2}-\eqref{T2E4} below is such that $\left\vert \sup_{y,w_0}\int_{\R^3}G_{\sigma }\left( y,w,w_{0},\tilde{t}\right)dw\right\vert \leq e^{-\kappa\tilde{t}}$, $\kappa>0$. Then, the diffusion coefficient $D_{g}$ and the friction term $\Lambda_g $ are given by:
\begin{align}
D_{g}\left(v\right) &=\pi \int_{\R^{3}}   \int_{\R^{3}}%
\frac{\left( k\otimes k\right) \left\vert \hat{\Phi}(k)
	\right\vert ^{2}\delta(k\cdot (v-w))}{\left\vert \Delta _{\sigma }\left( k,-i\left( k\cdot w\right)
	\right) \right\vert ^{2}} g\left( w\right)\ud{k}\ud{w}  \ , \label{eqthm:nonlin1}\\
\Lambda_g(v)&= -\pi \int_{\R^3}\int_{\R^3}
 \frac{ (k \otimes k) \left\vert \hat{\Phi}\left( k\right)
 	\right\vert ^{2}\delta(k\cdot (v-w))}{|\Delta_\sigma(k,-i\left( k\cdot w\right)|^2}\nabla g(w) \ud{k} \ud{w} . \label{eqthm:nonlin2}
\end{align}
\end{itemize}
\end{theorem}

\begin{remark}
The exponential decay of the function $G_{\sigma}$ in Theorem \ref{thm:FiniteRange} is closely related to the stability of a spatially homogeneous Vlasov medium with underlying distribution of particle velocities $g(w)$. The physical meaning and the relevance of the stability of the Vlasov medium, as well as of the dielectric function will be discussed in Appendix \ref{app:B}.
\end{remark}

\begin{remark}
	We will see in the proof that the particle density around the tagged particle stabilizes 
	to a stationary distribution (in a coordinate system moving with the tagged
	particle) in the microscopic time scale, which is much shorter than the
	macroscopic time scale. This was noticed first by Bogoliubov. 
\end{remark}

\begin{remark}
	The earliest study of systems with long range interactions can be found in
	\cite{Bo}. The approach introduced there is based on the BBGKY hierarchy,
	arguing that the truncated correlation function $g_2$ and truncated
	correlations of higher order stabilize on a much shorter timescale than the
	one-particle distribution function $f_1$. In our approach this corresponds
	to the stabilization of the function $\Lambda_g$, as well as of the covariance of the noise $D_g$ to
	a stationary process on the short time scale introduced by Bogoliubov.
\end{remark}

\bigskip

\subsubsection{Coulombian interactions \label{RaylCoulomb}}

We now consider the case of interaction potentials decreasing as the Coulomb potential at large
distances. We recall that in this case, in order to be able to define the
random force field we need to impose suitable electroneutrality conditions.
To be concise, we restrict our analysis to the case in which we have only
two types of charges having opposite signs, namely we assume that they have charges $+1$
and $-1$ respectively, but it would be possible to
adapt the arguments to more complicated charge distributions. We will
consider the following types of interaction potentials:%
\begin{equation}
\Phi _{\eps }\left( X\right) =\Phi \left( \frac{\left\vert
X\right\vert }{\eps }\right) \ \ \text{or \ }\Phi _{\eps
}\left( X\right) =\eps \Phi \left( \left\vert X\right\vert \right),
\label{S8E5}
\end{equation}%
where $\Phi $ is a smooth function which behaves for large values as $\Phi
\left( s\right) \sim \frac{1}{s}$ as $s\rightarrow \infty$ (cf.~\eqref{phiDec_sOne}), as well as the
corresponding asymptotic formulas for the derivatives of $\Phi .$ The first
type of potential in (\ref{S8E5}) includes for instance the Landau part
associated to Coulomb potentials (cf. (\ref{S1E5})).

Following the same reasoning of the previous Section we are left to consider the following approximating system of equations:
\begin{equation}\label{eq:appCoul}
\frac{d\xi }{d\tilde{t}}=V_0\ \ ,\ \ \frac{dV}{d\tilde{t}}=\left(
L_{\eps }\right) ^{\frac{3}{2}}F_g(\xi,\tilde{t};\omega), 
\end{equation}
where 
\begin{equation}
F_g(\xi,\tilde{t};\omega):=-\int_{\R^3}\int_{\R^3}\nabla _{\xi }%
\tilde{\phi}_{\eps }\left( \xi -\eta \right) \left[ \zeta^+
_{\eps }\left( \eta ,w,\tilde{t}\right) -\zeta^- _{\eps
}\left( \eta ,w,\tilde{t}\right) \right] d\eta dw.   \label{T1E3bis}
\end{equation}
Here
\begin{equation}\label{eq:tildephiep}
\tilde{\phi}_{\eps}\left( \xi\right) = \Phi _{\eps }\left( L_\ep \xi\right) \quad \text{with}\ \ L_\ep \gg 1
\end{equation} 
where $\Phi _{\eps }$ is as in \eqref{S8E5} with $\Phi \left( s\right) \sim \frac{1}{s}$ as $s\rightarrow \infty$, i.e. \eqref{phiDec_sOne} holds for some $A>0$. 
\smallskip

In the so-called Rayleigh gas case, the evolution of $\zeta_\ep$ is given by
\begin{align}  \label{T1E4a}
\partial _{\tilde{t}}\zeta^\pm _{\eps }\left( y,w,\tilde{t}\right)
+w\cdot \nabla _{y}\zeta^\pm _{\eps }\left( y,w,\tilde{t}\right)
\mp  L_\eps^\frac32 \nabla _{y}\tilde{\phi}%
_{\eps }\left( y-\xi\right) \cdot \nabla _{w}g^\pm\left(
w\right) & =0 ,
\end{align}
whereas in the interacting case the evolution is given by
\begin{align}
	& \partial _{\tilde{t}}\zeta _{\eps }^{\pm }\left( y,w,\tilde{t}%
	\right) +w\cdot \nabla _{y}\zeta _{\eps }^{\pm }\left( y,w,\tilde{t}%
	\right)\nonumber \\
	& - L_{\eps }^{\frac{3}{2}}\left[ L_\eps^2 \int_{\R^3}\nabla \Phi \left( L_{\eps
	}\left\vert y-\eta \right\vert \right) (\tilde{\rho}_{\eps }^{\pm }-%
	\tilde{\rho}_{\eps }^{\mp })(\eta ,\tilde{t})d\eta \pm \eps
	L_{\eps }\nabla \Phi \left( L_{\eps }\left\vert y-\xi
	\right\vert \right) \right] \cdot \nabla g^\pm\left( w\right) =0 \label{A1E5}
\end{align}
with random initial data $\zeta^\pm _{\eps }(y,w,0)=N^\pm(y,w)$
satisfying:%
\begin{align}
\mathbb{E}\left[ N^\pm\left( y,w\right) \right] = 0, \quad ,& \mathbb{E}%
\left[ N^+ _{\eps }\left( y_{a},w_{a}\right) N^-\left(
y_{b},w_{b}\right) \right] =0  \label{T1E1} \\
\mathbb{E}\left[ N^\pm\left( y_{a},w_{a}\right) N^\pm _{\eps }\left(
y_{b},w_{b}\right) \right] & =g^\pm\left( w_{a}\right) \delta \left(
y_{a}-y_{b}\right) \delta \left( w_{a}-w_{b}\right) .  \label{T1E2}
\end{align}
For simplicity, in what follows, we will assume $g^+=g^{-}=g$. 

\begin{theorem}\label{thm:Coulomb}
We consider the system \eqref{eq:appCoul}-\eqref{T1E3bis}. Let the interaction
potential $\tilde{\phi}_\ep$ be defined as in \eqref{eq:tildephiep}. We consider the following rescaling:
\begin{equation*}
t=\theta_\eps^2 \tilde{t}\ \ ,\ \ x= \theta_\eps^2 \xi\ \ ,\ \ v=V_0
\end{equation*}
where
\begin{equation}
\theta_\ep^2 =\frac{L_\ep}{T_{\ep}}\quad \text{with}\quad T_{\eps }= \frac{1}{2\eps ^{2}\log \left( \frac{1}{%
\eps }\right) }\ \ \text{as\ \ }\eps \rightarrow 0.
\label{TCoulRay}
\end{equation}
Setting $F_g(\xi,\tilde{t})=\mathcal{F}_g(x,t)$ we then obtain 
\begin{align*}
&\ \lim_{\theta\to 0}\mathbb{E}\left[ \mathcal{F}_g(vt, t;\omega)\right]=-\Lambda_{g}(v), \\  
&\ \lim_{\theta\to 0}\mathbb{E}\left[ (\mathcal{F}_{g}\left( v t_1,{t}_{1}\right)+\Lambda_{g}(v)) \otimes
(\mathcal{F}_{g}\left( v t_2, {t}_{2}\right)+\Lambda_{g}(v)) \right] 
=D_{g}\left( v\right) \delta \left( t_{1}-t_{2}\right) .
\end{align*}
More precisely, 
\begin{itemize}
\item[(I)]{\textbf{Rayleigh Gas case}}\\ 
Let $L_\ep=\ep^{-\beta}$ with $\beta\in(0,2)$. Then, the diffusion coefficient $D_{g}$ and the friction term $\Lambda_g $ are given by: 
\begin{align}
D_{g}\left( v\right) &=\int_{0}^{\infty }\frac{ds}{s^{3}}\int_{\mathbb{R}%
^{3}}d\eta _{1}\int_{\R^3}d\eta _{2}\nabla _{\eta }\Phi \left(
\eta _{1}\right) \otimes \nabla _{\eta }\Phi \left( \eta _{2}\right) g\left(
v+\frac{\eta _{1}-\eta _{2}}{s}\right) \ ,  \label{Coulomb:RayleighRandom}\\
\Lambda_{g}(v) &= -\int_{\left( \mathbb{R}^{3}\right)
^{2}}dwdy \frac{1}{\left\vert y\right\vert ^{2}\left\vert y+\left(
w-v\right) \right\vert ^{2}}\left( \nabla _{w}g\left( w\right) \cdot \frac{%
y+\left( w-v\right) }{\left\vert y+\left( w-v\right) \right\vert }\right) 
\frac{y}{\left\vert y\right\vert }, \label{Coulomb:RayleighFriction}
\end{align}
\item[(II)]{\textbf{Interacting particle case}}\\ 
Let $L_\ep=\ep^{- \frac 1 2 }$, namely $L_\ep$ of the order of the Debye length scale.  
We assume the Penrose stability condition of the medium, i.e. that the
function $\Delta _{\eps }\left( k,\cdot \right)$ defined as in \eqref{def:dielectric} below is analytic in $\left\{ \func{Re}\left( z\right) >0\right\} $ for each $k\in \R^3.$ 
Then, the diffusion coefficient $D_{g}$ and the friction term $\Lambda_g$ are given by: 
\begin{align}
	 D_{g}\left( v\right)= & \frac{\pi A^2}{2} \int_{\mathbb{R}%
		^{3}}dw_{0} \frac{g(w_{0})}{|v-w_{0}|}\left( I-\frac{(v-w_{0})\otimes
		(v-w_{0})}{|v-w_{0}|^{2}}\right), \label{Coulomb:InterRandom}\\
	\Lambda_{g}(v) =& -\frac{\pi A^2}{2} \int_{\mathbb{R}%
	^{3}}dw_{0} \frac{\nabla g(w_{0})}{|v-w_{0}|}\left( I-\frac{(v-w_{0})\otimes
	(v-w_{0})}{|v-w_{0}|^{2}}\right), \label{Coulomb:InterFriction}
\end{align}
 where $A$ is as in \eqref{phiDec_sOne}. 

\end{itemize}
\end{theorem}

\begin{remark}
We notice that the function $\Delta _{\eps }\left( k,\cdot \right)$ has additional analyticity properties due to the fact that \eqref{def:dielectric} is obtained by means of a Coulombian potential. These properties will be discussed in detail in the proof (cf.~Section \ref{RaylCoulomb}).
\end{remark}

\begin{remark}
We observe that a difference between the case of Coulomb potentials and the case of 
potentials considered in Subsection \ref{sec:finRanTh} is that in this second
case, there is a well defined friction coefficient acting over a tagged
particle moving at speed $V=V_0$ (cf.~\eqref{eqthm:nonlin1}) and a well defined stationary random force field 
acting on the tagged particle (cf. (\ref{eq:ForceFiniteCombined}), (\ref{S7E1})). Both the
friction term and the stationary random force field stabilize to their asymptotic behaviour in large
microscopic times. In the case of Coulombian interactions such stabilization
of the friction and the random force field in macroscopic times does not take place. On
the contrary, they increase logarithmically. It is well known that there is
not a large difference in practice between a magnitude converging to a value
or increasing logarithmically. Nevertheless, the main consequence of this
logarithmic behaviour will be the onset in the macroscopic time scale of the
Coulombian logarithm as it might be expected in a system with Coulombian
interactions.

Finally, it is interesting to remark that there is a difference between the
interacting particle systems and the case of Rayleigh gases
with Coulombian interactions. In
this second case, some of the logarithmic factors are due to the fact that
all the dyadic scales between the cutoff length for the Coulombian potential
and the macroscopic scale yield contributions of a similar size in the
computation of the friction term and the noise term. (See the discussion in
Subsection \ref{ThresholdCoulomb} below about the contributions due to dyadics for
Coulombian potentials). In the case of interacting particle systems the
"dyadic" scales contributing to the friction and noise term are those
between the particle size and the Debye length. Due to this, some of the
numerical factors appearing in the formula of the macroscopic scale are
different for Rayleigh gases and interacting particle systems. \end{remark}


\subsubsection{Grazing Collisions} \label{ss:RayGrazing}

 
We now consider the case which is usually termed with the name of grazing collisions. This corresponds 
to the case of
particles which interact weakly with a tagged particle, with an interaction
range smaller or similar to the average particle distance. 
We study the Rayleigh gas dynamics (cf.~(\ref{eq:RayNew}))  as well as the interacting particles dynamics (cf.~\eqref{eq:IntNew}). 
We will assume that there is only one type of
scatterers. We assume that the interaction potential $\Phi _{\eps }$ is given as in \eqref{WeakCoupling}, i.e. $\Phi _{\eps }\left( X\right) =\eps \Phi \left( \frac{%
\left\vert X\right\vert }{\ell _{\eps }}\right)$ with $\ell _{\eps }\rightarrow \sigma,$ for some $\sigma \geq 0$. 
Here $\Phi =\Phi \left( s\right) $ is a smooth function which decreases
sufficiently fast as $s\rightarrow \infty ,$ together with 
the asymptotic formulas for its derivatives.  We can assume for
instance that $\Phi \left( s\right) \sim \frac{1}{s^{\alpha }}$ as $%
s\rightarrow \infty $ with $\alpha >1,$ or that $\Phi \left( s\right) $
decreases exponentially. 

Following the same strategy of the previous sections we consider the following approximating system for both Rayleigh Gases and interacting particle systems: 
\begin{equation}
\frac{d\xi }{d\tilde{t}}=V_0\ \ ,\ \ \frac{dV}{d\tilde{t}}=-\left(
L_{\eps }\right) ^{3} F_g(\xi,\tilde{t};\omega) .  \label{T4E7}
\end{equation}
where 
\begin{equation}
F_g(\xi,\tilde{t};\omega):=\int_{\R^3}\int_{\R^3}\nabla
_{\xi }\tilde{\phi}_{\eps }\left( \xi -\eta \right) \zeta
_{\eps }\left( \eta ,w,\tilde{t}\right) dwd\eta \label{T4E7bis}.
\end{equation}
Let 
\begin{equation}\label{phitildeepgraz}
\Phi_{\eps}\left( X\right) =\tilde{\phi}_{\eps}\left( \frac{X}{%
	L_{\eps}}\right) =\tilde{\phi}_{\eps}\left( \xi\right) \ ,
\end{equation}
and 
\begin{equation}
(\partial_{\tilde{t}}+w\cdot\nabla_y)\zeta_{\eps}\left( y,w,\tilde{t}%
\right) -\left( L_{\eps}\right) ^{\frac{3}{2}}\nabla_{y}\tilde{\phi}%
_{\eps }\left( y-\xi\right) \cdot\nabla_{w}g\left( w\right) =0, \,\,
\zeta_{\eps}\left( y,w,0\right) =N_{\eps}\left( y,w\right)
\label{T4E4}
\end{equation}
where the (random) initial data is characterized by: 
\begin{align}
\mathbb{E}\left[ N_{\eps}\left( y,w\right) \right] &=0  
\nonumber \\
\mathbb{E}\left[ N_{\eps}\left( y_{a},w_{a}\right) N_{\eps
}\left( y_{b},w_{b}\right) \right] &=\frac{g\left( w_{a}\right) }{\left(
L_{\eps}\right) ^{3}}\delta\left( y_{a}-y_{b}\right) \delta\left(
w_{a}-w_{b}\right).  \label{T4E6}
\end{align}

\begin{theorem}\label{thm:Grazing}
Let the interaction
potential $\tilde{\phi}_\eps $ be as in \eqref{phitildeepgraz} and $L_\ep=\ep^{-\beta}$ with $\beta\in(0,2)$. We consider the following rescaling:
\begin{equation*}
t=\big(\ep\ell_{\ep}\big)^2 L_{\ep}\tilde{t}\ \ ,\ \ x=\big(\ep\ell_{\ep}\big)^2 L_{\ep} \xi\ \ ,\ \ v=V_0.
\end{equation*}
Setting $F_g(\xi,\tilde{t})=\mathcal{F}_g(x,t)$ we then obtain the following formulas for the diffusion coefficient $D_{g}$ and the friction term $\Lambda_g $:
\begin{align}
	D_g(v)=& \int_{0}^{\infty }ds\int_{\R^3}dy\int_{\R^3}dw\nabla
	_{y}\Phi \left( \left\vert y+(w_0-w)s\right\vert \right) \otimes \nabla _{y}\Phi
	\left( \left\vert y\right\vert \right) g\left( w_0\right) ,
\end{align}
\begin{equation}
\Lambda_{g}\left( v\right) =-\int_{\mathbb{R}^{3}}dy\int_{\mathbb{R}%
^{3}}dw\nabla _{y}\Phi \left( \left\vert y\right\vert \right) \nabla
_{w}g\left( w\right) \cdot \int_{-\infty }^{0}\nabla _{y}\Phi \left(
\left\vert y+\left( v-w\right) s\right\vert \right) ds . \label{T4E8a}
\end{equation}%


\end{theorem}

\subsubsection*{The threshold between short and long range potentials in the
decay $\frac{1}{\left\vert x\right\vert }.$}\label{ThresholdCoulomb} 

Something that was seen in \cite{NSV} is that in the case of Lorentz gases
and interaction potentials with the form $\frac{\eps }{\left\vert
x\right\vert ^{s}}$ the threshold which separates between short range and
long range potentials corresponds to $s=1,$ i.e. to potentials decreasing
like Coulombian potentials. Actually, the same separation between long range
and short range potentials holds in the case of Rayleigh gases.

The rationale behind this difference between short and long range potentials
can be easily understood at the physical level (cf.~\cite{BT} for an
explanation of the onset of the Coulombian logarithm in the case of
Coulombian potentials). Suppose that we consider a tagged particle moving
along a rectilinear path which will be denoted as $\ell $ during a length $%
L. $ We consider the deflections of the tagged particle induced by a set of
scatterers distributed according to the Poisson distribution with an average
distance $1.$ More precisely, we distinguish the deflections produced by the
scatterers at dyadic distances, i.e. at a distance $2^{n}$ and $2^{n+1}$
from $\ell $ for each $n\in \mathbb{Z}$, where we assume that $2^{n}$ is
smaller than the range of the potential. We will denote the distances
between $2^{n}$ and $2^{n+1}$ for each $n$ as a "dyadic". The deflection
experienced by the tagged particle is a random variable with zero average,
due to the fact that the scatterers are symmetrically distributed with
respect to the line $\ell .$ On the other hand if we estimate the variance
of the deflections produced by the particles at distances between $2^{n}$
and $2^{n+1}$ it can be readily seen that the magnitude of this variance
decreases exponentially with $n$ if $s>1$ and increases exponentially with $%
n $ (as long as $2^{n}\lesssim L$) for $s<1$. The contribution to the
deflections of the scatterers at distances larger than $L$ is negligible if $%
s>\frac{1}{2}$ (cf. \cite{NSV} for a detailed proof of this in the case of
static scatterers). Notice that for these ``far-away" scatterers the change
of angle subtended by the tagged particle is very small, i.e. we can say
that no collision is taking place. In the case of $s=1$ the magnitude of the
deflections is the same for all the values of $n$ as long as $2^{n}\lesssim
L $. Actually, this is the reason for the onset of the Coulombian logarithm,
namely the fact that we need to add the deflections produced by the
different dyadic cylinders and the number of these cylinders is of order $%
\log \left( L\right) $, if we assume that the potential $\frac{1}{\left\vert
x\right\vert }$ is cut at distances of order $\left\vert x\right\vert \simeq
1.$

It turns out that a similar picture takes place in the case of moving
scatterers, i.e. for Rayleigh gases. The onset of the logarithmic correction
for potentials behaving like Coulomb for large distances has been seen
above. Seemingly nonkinetic models, with long range correlations due to the
long range of the potentials arise for $s\in \left( \frac{1}{2},1\right) $
as it happens in the Lorentz gases considered in \cite{NSV}, in spite of the
fact that for Rayleigh gases the scatterers move. However, we will not
continue the study of this case in this paper. (Nevertheless, the case $s\in
\left( \frac{1}{2},1\right) $ is discussed in \cite{NVW} Section 3.3, 
in the case of interacting particle systems).

\section{Justification of the approximation by a tagged particle in an evolving Gaussian field}\label{ss:JustR}

The goal of this section is to provide a formal derivation of the Langevin equation \eqref{eq:TagFric} used to formulate Theorems \ref{thm:FiniteRange}, \ref{thm:Coulomb}, \ref{thm:Grazing}, starting from the Newton equations for the Rayleigh Gas (cf.~\eqref{eq:RayNew}) as well as from the Newton equations for the interacting particle system (cf.~\eqref{eq:IntNew}), at least during small macroscopic times, in suitable scaling limits.

\subsection{The case of interactions with finite range much larger than the particle distance} \label{ssec:Jfinrange}
\smallskip

The goal of this section is to justify the system of equations \eqref{eq:MotionFiniteCombined}-\eqref{eq:zetaFiniteRayleigh} starting from the Newton equations for the Rayleigh Gas (cf.~\eqref{eq:RayNew}). We also justify the system of equations \eqref{eq:MotionFiniteCombined}-\eqref{S7E1}, \eqref{eq:zetaFiniteInteracting} starting from the Newton equations for the interacting particle system (cf.~\eqref{eq:IntNew}). \bigskip

\noindent\textbf{Rayleigh Gas case}

\noindent Suppose that the position and velocity of a tagged particle $\left(
X,V\right) $ evolve according to (\ref{eq:RayNew}) where the interaction
potential $\Phi _{\eps }\left( x\right) $ is as in (\ref{S4E6}), (\ref%
{S4E7}). We will assume that $\Phi \left( s\right) $ decreases fast enough
as $s\rightarrow \infty ,$ say exponentially, although in space dimension
three a decay like $\frac{1}{s^{a}}$ with $a>1$ would be enough. As a first
step we approximate the distribution of scatterers as the sum of a constant
density plus some Gaussian fluctuations in some suitable topology. To this
end we introduce a new variable $y_{k}=\frac{Y_{k}}{L_{\eps }}$ which
will be useful to describe the system on a scale where this Gaussian
approximation is valid, but that is smaller than the mean free path. In
order to keep the particle velocities of order one we introduce a new time
scale by means of $\tilde{t}=\frac{\tau }{L_{\eps }}.$ We write also $%
\xi =\frac{X}{L_{\eps }}.$ Then (\ref{eq:RayNew}) becomes:%
\begin{equation}  \label{S6E1}
\begin{aligned} \frac{d\xi }{d\tilde{t}}& =V\ \ ,\ \
\frac{dV}{d\tilde{t}}=-\eps \sum_{j\in S}\nabla _{\xi }\Phi \left( \xi -y_{j}\right) \\ \frac{dy_{k}}{d\tilde{t}}& =W_{k}\ \ ,\ \
\frac{dW_{k}}{d\tilde{t}}=-\eps \nabla _{y}\Phi \left( y_{k}-\xi \right) \ \ ,\ \ k\in S \ . \end{aligned}
\end{equation}

The goal is to approximate the dynamics of the scatterers by means of a
continuous density. To this end we introduce the following particle density
in the phase space:%
\begin{equation}  \label{eq:feps}
f_{\eps }\left( y,w,\tilde{t}\right) =\frac{1}{\left( L_{\eps
}\right) ^{3}}\sum_{k}\delta \left( y-y_{k}\right) \delta \left(
w-W_{k}\right).
\end{equation}

We can then rewrite the first two equations of (\ref{S6E1}) as:%
\begin{equation}
\frac{d\xi }{d\tilde{t}}=V\ \ ,\ \ \frac{dV}{d\tilde{t}}=-\eps \left(
L_{\eps }\right) ^{3}\int_{\R^3}\nabla _{\xi }\Phi \left(
\xi -\eta \right) \rho _{\eps }\left( \eta ,\tilde{t}\right) d\eta
\label{S6E2a}
\end{equation}%
where $\rho _{\eps }(y,\tilde{t})=\rho \lbrack f_{\eps }(\cdot
,\tilde{t})](y)$ is the spatial density. On the other hand, the second set of equations of (\ref{S6E1}) implies
that:%
\begin{equation}
\partial _{\tilde{t}}f_{\eps }\left( y,w,\tilde{t}\right) +w\cdot
\nabla _{y}f_{\eps }\left( y,w,\tilde{t}\right) -\eps \nabla
_{y}\Phi \left( y-\xi \right) \cdot \nabla _{w}f_{\eps }\left( y,w,%
\tilde{t}\right) =0 \ .  \label{S6E3}
\end{equation}

We have then reformulated (\ref{S6E1}) as (\ref{S6E2a}), (\ref{S6E3}).

We can now take formally the limit $\eps \rightarrow 0.$ To this end,
notice that $f_{\eps }\left( y,w,0\right) $ is of order one and it
converges in the weak topology to $g\left( w\right) .$ In order to obtain
the evolution for different rescalings of $L_{\eps }$ with $%
\eps $ we compute the asymptotic behaviour in law of $f_{\eps
}\left( y,w,0\right) $ as $\eps \rightarrow 0.$ The following
Gaussian approximation will be used repeatedly in the
following: 
\begin{equation}
\mathbb{E}\left[ \left( f_{\eps }\left( y_{a},w_{a},0\right) -g\left(
w_{a}\right) \right) \left( f_{\eps }\left( y_{b},w_{b},0\right)
-g\left( w_{b}\right) \right) \right] =\frac{g\left( w_{a}\right) }{\left(
L_{\eps }\right) ^{3}}\delta \left( y_{a}-y_{b}\right) \delta \left(
w_{a}-w_{b}\right) .  \label{S6E4}
\end{equation}

This approximation will be justified in Appendix A where it will be seen how
to derive it from the empirical densities associated to particle
distributions given by the Poisson measure.

Assuming (\ref{S6E4}), it is natural to look for solutions of (\ref{S6E3})
with the form:%
\begin{equation}
f_{\eps }\left( y,w,\tilde{t}\right) =g\left( w\right) +\frac{1}{%
\left( L_{\eps }\right) ^{\frac{3}{2}}}\zeta _{\eps }\left(
y,w,\tilde{t}\right).  \label{S6E4a}
\end{equation}

Then, using the fact that the contribution to the integral $\int_{\mathbb{R}%
^{3}}\nabla _{\xi }\Phi \left( \xi -\eta \right) \rho _{\eps }\left(
\eta ,\tilde{t}\right) d\eta $ due to the term $g\left( w\right) $ vanishes,
we obtain that $\zeta _{\eps }\left( y,w,\tilde{t}\right) $ solves
the following problem:

\begin{align}
\frac{d\xi }{d\tilde{t}}&=V\ \ ,\ \ \frac{dV}{d\tilde{t}}=-\eps
\left( L_{\eps }\right) ^{\frac{3}{2}}\int_{\R^3}\nabla
_{\xi }\Phi \left( \xi -\eta \right) \tilde{\rho}_{\eps }\left( \eta ,%
\tilde{t}\right) d\eta,  \label{S6E5} \\
\partial _{\tilde{t}}\zeta _{\eps }\left( y,w,\tilde{t}\right)
+&w\cdot \nabla _{y}\zeta _{\eps }\left( y,w,\tilde{t}\right)
-\eps \left( L_{\eps }\right) ^{\frac{3}{2}}\nabla _{y}\Phi
\left( y-\xi \right) \cdot \nabla _{w}\left( g\left( w\right) +\frac{\zeta
_{\eps }\left( y,w,\tilde{t}\right)}{\left( L_{\eps }\right) ^{%
\frac{3}{2}}} \right) =0 ,  \label{S6E6}
\end{align}
where $\tilde{\rho}_{\eps }\left( y,\tilde{t}\right) = \rho[%
\zeta_{\eps}(\cdot,\tilde{t})](y)$ is the associated spatial density.

Notice that we assume that the potential $\Phi $ decreases sufficiently fast
to guarantee that the integral $\int_{\R^3}\nabla _{\xi }\Phi
\left( \xi -\eta \right) \tilde{\rho}_{\eps }\left( \eta ,\tilde{t}%
\right) d\eta $ is well defined for the random field 
$\zeta
_{\eps }\left( y,w,\tilde{t}\right) $ given by (\ref{S6E4}), (\ref%
{S6E4a}). In the case of potentials $\Phi $ decreasing as some power laws it
is possible to give a meaning to this integral by means of a limit procedure.

Then, using the fact that $L_{\eps }\rightarrow \infty $ we obtain
the following limit problem formally. We write 
\begin{equation}
\theta _{\eps }=\eps \left( L_{\eps }\right) ^{\frac{3}{%
2}}.  \label{thetaRay}
\end{equation}

Given that the range of the interaction potentials is of order $\left\vert
y\right\vert \approx 1$ we need to have $\theta _{\eps }\rightarrow 0$
as $\eps \rightarrow 0$ in order to obtain a kinetic limit. On the
other hand, since $L_{\eps }\rightarrow \infty ,$ making the Gaussian
approximation $\zeta _{\eps }\rightharpoonup \zeta ,$ we can
approximate the problem (\ref{S6E5}), (\ref{S6E6}) as: 
\begin{equation}
\frac{d\xi }{d\tilde{t}}=V_0\ \ ,\ \ \frac{dV}{d\tilde{t}}=-\theta
_{\eps }\int_{\R^3}\nabla _{\xi }\Phi \left( \xi -\eta
\right) \tilde{\rho}\left( \eta ,\tilde{t}\right) d\eta  \label{S6E7}
\end{equation}%
where $\tilde{\rho}\left( y,\tilde{t}\right) =\rho \lbrack \zeta (\cdot ,%
\tilde{t})](y)$ and 
\begin{align}
(\partial _{\tilde{t}}+w\cdot \nabla) _{y}\zeta \left( y,w,\tilde{t}\right)
-\theta _{\eps }\nabla _{y}\Phi \left( y-\xi \right) \cdot \nabla
_{w}g\left( w\right) & =0,\quad \zeta \left( y,w,0\right) =N\left(
y,w\right) .  \label{S6E8}
\end{align}%
Using (\ref{S6E4}) we obtain:%
\begin{equation}
\mathbb{E}\left[ N\left( y,w\right) \right] =0\ \ ,\ \ \mathbb{E}\left[
\left( N\left( y_{a},w_{a}\right) \right) N\left( y_{b},w_{b}\right) \right]
=g\left( w_{a}\right) \delta \left( y_{a}-y_{b}\right) \delta \left(
w_{a}-w_{b}\right).  \label{S7E1bis}
\end{equation}
This approximation is valid for infinitesimally small macroscopic times, where the approximation $V(t)\approx V_0$ is valid. This justifies the set of equations \eqref{eq:MotionFiniteCombined}-\eqref{eq:zetaFiniteRayleigh}. 
\bigskip

\noindent\textbf{Interacting particle system}

\noindent We now justify the system of equations \eqref{eq:MotionFiniteCombined}-\eqref{S7E1}, \eqref{eq:zetaFiniteInteracting} starting from the Newton equations for the interacting particle system (cf.~\eqref{eq:IntNew}), at least during small macroscopic times in suitable scaling limits.

We begin considering the case of interaction potentials with the form (\ref%
{S4E6}), (\ref{S4E7}). We consider the dynamics of a distinguished particle $%
\left( X,V\right) $ in an interacting particle system. Denoting as $\left(
Y_{k},W_{k}\right) ,\ k\in S$ the position and velocity of the remaining
particles of the system, using the change of variables (\ref{S8E6a}) as well
as (\ref{eq:IntNew}) we obtain:%
\begin{align}
\frac{d\xi }{d\tilde{t}}& =V\ \ ,\ \ \frac{dV}{d\tilde{t}}=-\eps
\sum_{j\in S}\nabla _{\xi }\Phi \left( \xi -y_{j}\right)  \label{S9E2} \\
\frac{dy_{k}}{d\tilde{t}}& =W_{k}\ \ ,\ \ \frac{dW_{k}}{d\tilde{t}}%
=-\eps \nabla _{y}\Phi \left( y_{k}-\xi \right) -\eps
\sum_{j\in S}\nabla _{y}\Phi \left( y_{k}-y_{j}\right) \ \ ,\ \ k\in S  \ .
\notag
\end{align}

In order to approximate this system by an evolution equation of the form (%
\ref{eq:TagFric}), first recall the empirical measure $f_{\eps}$ defined in %
\eqref{eq:feps}. Then, the first equation of (\ref{S9E2}) can be rewritten
as:

\begin{equation}
\frac{d\xi}{d\tilde{t}}=V\ \ ,\ \ \frac{dV}{d\tilde{t}}=-\eps\left(
L_{\eps}\right) ^{3}\int_{\R^3}\nabla_{\xi}\Phi\left(
\xi-\eta\right) \rho_{\eps}\left( \eta,\tilde{t}\right) d\eta
\label{S9E3}
\end{equation}
with $\rho_{\eps}\left( y,\tilde{t}\right) =\rho[f_{\eps}(%
\cdot,\tilde{t})](y)$ (cf. \eqref{eq:Spdensity}). 
The second set of equations in (\ref{S9E2}) becomes:%
\begin{equation}
\frac{dy_{k}}{d\tilde{t}}=W_{k}\ \ ,\ \ \frac{dW_{k}}{d\tilde{t}}%
=-\eps\nabla_{y}\Phi\left( y-\xi\right) -\eps\left(
L_{\eps}\right)^3(\nabla_{y}\Phi *\rho_{\eps})(y_k,\tilde{t})
\ ,\ \ k\in S \ .  \label{S9E4}
\end{equation}

Using (\ref{S9E4}) we can derive, arguing as in Subsection \ref%
{RaylCompSuppPot}, the following evolution equation for the particle density 
$f_{\eps}:$ 
\begin{equation}
\partial_{\tilde{t}}f_{\eps}\left( y,w,\tilde{t}\right) +w\cdot
\nabla_{y}f_{\eps}\left( y,w,\tilde{t}\right) - {\eps} \left[
\nabla_{y}\Phi\left( y-\xi\right) +\left( L_{\eps }\right) ^{3}
(\nabla_{y}\Phi * \rho_{\eps})( y,\tilde{t})\right] \cdot\nabla
_{w}f_{\eps}\left( y,w,\tilde{t}\right) =0 \ . \label{S9E5}
\end{equation}

Using the change of variables (\ref{S6E4a}), we can rewrite (\ref{S9E5}) as:%
\begin{align*}
(\partial _{\tilde{t}}& +w\cdot \nabla _{y})\zeta _{{\eps }}\left(
y,w,\tilde{t}\right) \\
-& \left[ \eps \left( L_{\eps }\right) ^{\frac{3}{2}}\nabla
_{y}\Phi \left( y-\xi \right) +\eps \left( L_{\eps }\right)
^{3}(\nabla _{y}\Phi \ast \tilde{\rho}_{\eps })(y,\tilde{t})\right]
\cdot \nabla _{w}\left( g\left( w\right) +\frac{1}{\left( L_{\eps
	}\right) ^{\frac{3}{2}}}\zeta _{\eps }\left( y,w,\tilde{t}\right)
\right) =0
\end{align*}%
where $\tilde{\rho}\left( y,\tilde{t}\right) =\rho \lbrack \zeta (\cdot ,%
\tilde{t})](y)$. Neglecting the term $\frac{1}{\left( L_{\eps
	}\right) ^{\frac{3}{2}}}\zeta _{\eps }$ which might be expected to be
small compared with $g$ we obtain the following approximation for the
fluctuations of the density of scatterers:%
\begin{equation}
\left( \partial _{\tilde{t}}+w\cdot \nabla _{y}\right) \zeta _{\eps
}\left( y,w,\tilde{t}\right) -\left[ \eps \left( L_{\eps
}\right) ^{\frac{3}{2}}\nabla _{y}\Phi \left( y-\xi \right) +\eps
\left( L_{\eps }\right) ^{3}(\nabla _{y}\Phi \ast \tilde{\rho}%
_{\eps })\left( y,\tilde{t}\right) \right] \cdot \nabla _{w}g\left(
w\right) =0.  \label{S9E5a}
\end{equation}

The term $\tilde{\rho}_{\eps }$ has the same order of magnitude as $%
\zeta _{\eps }.$ Then, the terms $w\cdot \nabla _{y}\zeta
_{\eps }\left( y,w,\tilde{t}\right) $ and $\eps \left(
L_{\eps }\right) ^{3}\left[ \int_{\R^3}\nabla _{y}\Phi
\left( y-\eta \right) \tilde{\rho}_{\eps }\left( \eta ,\tilde{t}%
\right) d\eta \right] \cdot \nabla _{w}g\left( w\right) $ yield a comparable
contribution if $\eps \left( L_{\eps }\right) ^{3}$ is of
order one. To consider the general case, we will assume that $\eps
\left( L_{\eps }\right) ^{3}\rightarrow \sigma $ as $\eps
\rightarrow 0,$ where $\sigma $ can be zero, a positive number or infinity.
(In this last case the dependence of $\sigma $ on $\eps $ must be
preserved, but it will not be explicitly written for the sake of
simplicity). We can then expect that $\zeta _{\eps }\sim \zeta $ as $%
\eps \rightarrow 0$ where $\zeta $ solves the equation \eqref{eq:zetaFiniteInteracting}.

\color{black}

\bigskip

\subsection{The case of Coulombian interactions}\label{ssec:JCoulomb}
The goal of this section is to justify the system of equations \eqref{eq:appCoul}-\eqref{T1E4a} starting from the Newton equations for the Rayleigh Gas (cf.~\eqref{eq:RayNew}). We also justify the system of equations \eqref{eq:appCoul}-\eqref{T1E3bis}, \eqref{A1E5} starting from the Newton equations for the interacting particle system (cf.~\eqref{eq:IntNew}). \bigskip

\textbf{Rayleigh Gas case}\\
\noindent
We now consider the evolution of a tagged particle $\left( X,V\right) $
evolving according to the Rayleigh gas dynamics (cf. (\ref{eq:RayNew})) with
interaction potentials decreasing as the Coulomb potential at large
distances. We recall that in this case, in order to be able to define the
random force field we need to impose suitable electroneutrality conditions.
To be concise, we restrict our analysis to the case in which we have only
two types of charges, having opposite signs, but it would be possible to
adapt the arguments to more complicated charge distributions. We will
consider the following types of interaction potentials:%
\begin{equation}
\Phi _{\eps }\left( X\right) =\Phi \left( \frac{\left\vert
X\right\vert }{\eps }\right) \ \ \text{or \ }\phi _{\eps
}\left( X\right) =\eps \Phi \left( \left\vert X\right\vert \right),
\label{S8E5bis}
\end{equation}%
where $\Phi $ is a smooth function which behaves for large values as $\Phi
\left( s\right) \sim \frac{1}{s}$ as $s\rightarrow \infty ,$ as well as the
corresponding asymptotic formulas for the derivatives of $\Phi .$ The first
type of potential in (\ref{S8E5bis}) includes for instance the Landau part
associated to Coulomb potentials (cf. (\ref{S1E5})).

We will assume that there are two types of scatterers, having charges $+1$
and $-1$ respectively. We will denote them as $\left\{ Y^\pm_{k}\right\}
_{k\in S}$ respectively. Then, the evolution equation for the tagged
particle is described by the following set of equations: 
\begin{align}
\frac{dX}{d\tau} & =V\ \ ,\ \ \frac{dV}{d\tau}=-\sum_{j\in
S}\nabla_{X}\Phi_{\eps}\left( X-Y^+_{j}\right) +\sum_{j\in\tilde{S}%
}\nabla_{X}\Phi_{\eps}\left( X-Y^-_{j}\right)  \notag \\
\frac{dY^\pm_{k}}{d\tau} & =W^\pm_{k}\ \ ,\ \ \frac{dW^\pm_{k}}{d\tau}%
=\mp\nabla_{Y}\Phi_{\eps}\left( Y^\pm_{k}-X\right) \ \ ,\ \ k\in S.
\label{S8E6}
\end{align}

Arguing as in Subsection \ref{RaylCompSuppPot}, we can approximate these
equations replacing the particle distributions by 
random fields on the
phase space. We have two different types of charges and therefore we need to
introduce two different densities. We define new variables by means of:%
\begin{equation}
y^\pm_{k}=\frac{Y^\pm_{k}}{L_{\eps}}\ \ ,\ \ \xi=\frac{X}{%
L_{\eps}}\ \ ,\ \ \tilde {t}=\frac{\tau}{L_{\eps}} \ .
\label{S8E6a}
\end{equation}
We will then denote as $%
f^\pm_{\eps }\left( y,w,\tilde{t}\right)$ the empirical densities defined as in \eqref{eq:feps}. 

In this case there is no canonical choice of the characteristic length $%
L_{\eps }.$ The first type of potentials in (\ref{S8E5bis}) has a
characteristic length $\eps $ and the second one has a characteristic
length of order one. However, on these length scales we cannot approximate
the distributions by continuous densities. In the case of interacting
particle systems, which will be considered right after, 
there will be a natural choice of length scale $L_{\eps },$ namely the so-called
Debye screening length. However, in the case of the Rayleigh gas under
consideration here such a natural choice does not exists since the power law
does not have any characteristic length. We will then assume a length scale $%
L_{\eps }$ such that: 
\begin{equation*}
L_{\eps }=\ep^{-\beta}, \ \ \beta\in (0,2).
\end{equation*} 
We define also a rescaled potential $\tilde{\phi}$ by means of:%
\begin{equation}
\Phi_{\eps}\left( X\right) =\tilde{\phi}_{\eps}\left( \frac{X}{%
L_{\eps}}\right) =\tilde{\phi}_{\eps}\left( \xi\right).
\label{S8E5a}
\end{equation}
Then, the equations for the tagged particle in (\ref{S8E6}) can be rewritten
as:%
\begin{equation}
\frac{d\xi }{d\tilde{t}}=V\ \ ,\ \ \frac{dV}{d\tilde{t}}=\left(
L_{\eps }\right) ^{3}\int_{\R^3}\nabla _{\xi }\tilde{\phi}%
_{\eps }\left( \xi -\eta \right) \left[ \rho^-_{\eps }\left(
\eta ,\tilde{t}\right) -\rho^+ _{\eps }\left( \eta ,\tilde{t}\right) %
\right] d\eta ,  \label{S8E7a}
\end{equation}%
where $\rho^\pm _{\eps }\left( y,\tilde{t}\right) = \rho[%
f_{\eps}^\pm(\cdot,\tilde{t})]$. On the other hand the equations for
the scatterers in (\ref{S8E6}) imply: 
\begin{equation}
\partial_{\tilde{t}}f^\pm_{\eps}\left( y,w,\tilde{t}\right) +w\cdot
\nabla_{y}f^\pm_{\eps}\left( y,w,\tilde{t}\right) \mp\nabla_{y}\tilde{%
\phi }_{\eps}\left( y-\xi\right)
\cdot\nabla_{w}f^\pm_{\eps}\left( y,w,\tilde{t}\right) =0 .
\label{S8E7}
\end{equation}
Suppose that the distribution of velocities of the scatterers $\left\{
Y^\pm_{k}\right\} _{k\in S}$ are given by functions $g^\pm \left( w\right) $%
. Then arguing as in Subsection \ref{ssec:Jfinrange} we can approximate the
initial data for (\ref{S8E7}) by means of Gaussian stochastic processes with
average $g^+\left( w\right) $ and $g^-\left( w\right) $ respectively which
satisfy:%
\begin{align}
\mathbb{E}\left[ \left( f^\pm_{\eps}\left( y,w,0\right) -g^\pm\left(
w\right) \right) \left( f^\pm_{\eps}\left( y^{\prime },w^{\prime
},0\right) -g^\pm\left( w^{\prime }\right) \right) \right] &=\frac{%
g^\pm\left( w\right) }{\left( L_{\eps}\right) ^{3}}\delta\left(
y-y^{\prime }\right) \delta\left( w-w^{\prime }\right)  \label{S8E9} \\
\mathbb{E}\left[ \left( f^+_{\eps}\left( y,w,0\right) -g^{+}\left(
w\right) \right) \left( f^-_{\eps}\left( y^{\prime },w^{\prime
},0\right) -g^{-}\left( w^{\prime }\right) \right) \right] &=0 .
\label{S8E9b}
\end{align}
The last equation ensures that the distributions of scatterers are mutually
independent.

It is then natural, arguing as in Subsection \ref{ssec:Jfinrange}, to look
for solutions with the form:%
\begin{equation}
f^\pm_{\eps }\left( y,w,\tilde{t}\right) =g^\pm\left( w\right) +\frac{%
1}{\left( L_{\eps }\right) ^{\frac{3}{2}}}\zeta^\pm _{\eps
}\left( y,w,\tilde{t}\right).  \label{S8E9c}
\end{equation}

Then, neglecting higher order terms, we obtain the the particle fluctuations
satisfy approximately the following problems:%
\begin{align}
\partial _{\tilde{t}}\zeta^\pm _{\eps }\left( y,w,\tilde{t}\right)
+w\cdot \nabla _{y}\zeta^\pm _{\eps }\left( y,w,\tilde{t}\right)
\mp\left( L_{\eps }\right) ^{\frac{3}{2}}\nabla _{y}\tilde{\phi}%
_{\eps }\left( y-\xi \right) \cdot \nabla _{w}g^\pm \left( w\right) &
=0 \label{S9E1a}
\end{align}%
with random initial data $\zeta^\pm _{\eps }(y,w,0)=N^\pm(y,w)$
satisfying:%
\begin{align*}
\mathbb{E}\left[ N^\pm\left( y,w\right) \right] = 0, \quad ,& \mathbb{E}%
\left[ N^+ _{\eps }\left( y_{a},w_{a}\right) N^-\left(
y_{b},w_{b}\right) \right] =0 
\\
\mathbb{E}\left[ N^\pm\left( y_{a},w_{a}\right) N^\pm _{\eps }\left(
y_{b},w_{b}\right) \right] & =g^\pm\left( w_{a}\right) \delta \left(
y_{a}-y_{b}\right) \delta \left( w_{a}-w_{b}\right) . 
\end{align*}

The electroneutrality condition requires that:%
\begin{equation*}
\int_{\R^3}g^+\left( w\right) dw=\int_{\R^3} g^- \left(
w\right) dw.
\end{equation*}

Then (\ref{S8E7a}) can be rewritten as:%
\begin{equation}
\frac{d\xi }{d\tilde{t}}=V\ \ ,\ \ \frac{dV}{d\tilde{t}}=-\left(
L_{\eps }\right) ^{\frac{3}{2}}\int_{\R^3}\nabla _{\xi }%
\tilde{\phi}_{\eps }\left( \xi -\eta \right) \left[ \zeta^+
_{\eps }\left( \eta ,w,\tilde{t}\right) -\zeta^- _{\eps
}\left( \eta ,w,\tilde{t}\right) \right] d\eta .  \label{T1E3}
\end{equation}
Summarizing, we have reduced the original problem to (\ref{S9E1a})-(\ref%
{T1E3}). It is natural, in order to derive a kinetic limit, to consider the
previous model as $\eps \rightarrow 0.$ Moreover, we will assume that 
$\tilde{t}$ is sufficiently small so that we can assume that $V(t)\approx V_0$, i.e.~it is
constant. This concludes the justification. 
\bigskip

\noindent\textbf{Interacting particle system}\\
\noindent 
We now consider interacting particle systems with the form (\ref{eq:IntNew}) in
which the interaction potentials are as in (\ref{S8E5}) (where $\Phi \left(
s\right) \sim \frac{1}{s}$ as $s\rightarrow \infty $). This question has
been extensively studied in the physical literature due to his relevance in
astrophysics and the in the theory of plasmas (cf. for instance \cite{RMcDJ, Trub}). In this case we cannot assume as in Subsection \ref{ssec:Jfinrange} that the interaction potential has a large, but finite
range. On the contrary, in this case, the range of the potential is
infinity. However, it turns out that, assuming some stability conditions for
the Vlasov medium similar to the ones discussed above, there exists a
characteristic length, namely the so-called Debye length which yields the
effective interaction length between the particles of the system. This
length is characterized by the fact that the potential energy of a particle is  comparable to its kinetic energy. 
A precise definition will be given in (\ref%
{DebDef}).

We now study the evolution of a system of particles with interactions that behave at large distances as Coulombian potentials. We will
assume in most of the following that the potentials are smooth and in
particular, that the deflection experienced by two colliding particles which
interact by means of this potential is small. In the case of point charges,
it is possible to cut the potential as it was made in \cite{NSV} in order to
separate the Boltzmann collisions (due to close binary encounters between
particles) and weak deflections due to the effect of many random collisions.
Given that the Boltzmann collisions take place in a larger time scale than
the small deflections (due to the presence of the Coulombian logarithm) we
will ignore that part of the potential in the following.

Assume that there are two types of scatterers with opposite charges, namely $%
\left\{ Y_{k}\right\} _{k\in S}$ and $\left\{ \tilde{Y}_{k}\right\} _{k\in 
	\tilde{S}}$ having opposite charges. The main difference between the case
considered in Subsection \ref{ssec:Jfinrange} and the case considered here, besides the fact that we need to include in the dynamics at
least two types of particles in order to have electroneutrality, is the fact
that due to the power law behaviour of the potential, there is not any
intrinsic length scale that we can call the range of the potential.

More precisely, we will assume that the interaction potential between the
particles has the form in (\ref{S8E5}). We can assume in particular that $%
\Phi \left( \xi \right) \simeq \frac{1}{\left\vert \xi \right\vert }$ for
large values, introducing a cutoff near the origin in order to avoid large
deflections as indicated above. We rescale the variables as in (\ref{S8E6a}%
), where $L_{\eps }$ will be chosen now as the so-called Debye length
which will be chosen shortly. We will just assume for the moment that $%
L_{\eps }\gg 1,$ i.e.~that it is larger than the average particle
distance. Then, the evolution of the rescaled system of particles (in which
we distinguish a tagged particle) is given by:\ 
\begin{align}
\frac{d\xi }{d\tilde{t}}& =V\ \ ,\ \ \frac{dV}{d\tilde{t}}=-\sum_{j\in
	S}\nabla _{\xi }\tilde{\phi}_{\eps }\left( \xi -y_{j}^{+}\right)
+\sum_{j\in \tilde{S}}\nabla _{\xi }\tilde{\phi}_{\eps }\left( \xi
-y_{j}^{-}\right) ,\quad \frac{dy_{k}^{\pm }}{d\tilde{t}}=W_{k}^{\pm }
\label{T5E7} \\
\frac{dW_{k}^{\pm }}{d\tilde{t}}& =\mp \nabla _{y}\tilde{\phi}_{\eps
}\left( y_{k}^{\pm }-\xi \right) -\sum_{j\in S,\ j\neq k}\left( \nabla _{y}%
\tilde{\phi}_{\eps }\left( y_{k}^{\pm }-y_{j}^{\pm }\right) -\nabla
_{y}\tilde{\phi}_{\eps }\left( y_{k}^{\pm }-y_{j}^{\mp }\right)
\right) \ \ ,\ \ k\in S.  \notag
\end{align}

Notice that we assume that all the particles have the same mass. It would be
possible to consider more general situations.

We approximate the system of equations (\ref{T5E7}) by a Vlasov equation. To this end we introduce particle densities as it was made in (\ref{S8E7a}). Then,
arguing as above we obtain the following evolution
equation for the tagged particle 
\begin{equation}
\frac{d\xi }{d\tilde{t}}=V\ \ ,\ \ \frac{dV}{d\tilde{t}}=-\left(
L_{\eps }\right) ^{3}\int_{\R^3}\nabla _{\xi }\tilde{\phi}%
_{\eps }\left( \xi -\eta \right) \left[ \rho _{\eps
}^{+}\left( \eta ,\tilde{t}\right) -\rho _{\eps }^{-}\left( \eta ,%
\tilde{t}\right) \right] d\eta ,  \label{A2E7}
\end{equation}%
and for the scatterers distributions:%
\begin{align} \label{eq:f1}
&\partial _{\tilde{t}}f_{\eps }^{\pm
}\left( y,w,\tilde{t}\right) +w\cdot \nabla _{y}f_{\eps }^{\pm
}\left( y,w,\tilde{t}\right) \nonumber\\ &-\left[ \left( L_{\eps }\right)
^{3}\int_{\R^3}\nabla _{y}\tilde{\phi}_{\eps }\left( y-\eta
\right) (\rho _{\eps }^{\pm }-\rho _{\eps }^{\mp })\left( \eta
,\tilde{t}\right) d\eta \pm \nabla _{y}\tilde{\phi}_{\eps }\left(
y-\xi \right) \right] \cdot \nabla _{w}f_{\eps }\left(
y,w,\tilde{t}\right) =0 
\end{align}
where $\rho^\pm _{\eps } = \rho[f^\pm_{\eps}]$.

We can now define the Debye screening length $L_{\eps }.$ This length
will be chosen in order to make the terms associated to particle transport
(i.e. $w\cdot \nabla _{y}f_{\eps },\ w\cdot \nabla _{y}\tilde{f}%
_{\eps }$) and the terms describing the self-interactions between
scatterers (i.e. $\left[ \nabla _{y}\tilde{\phi}_{\eps }\ast \rho
_{\eps }\right] \cdot \nabla _{w}f_{\eps }$ and similar ones)
comparable. To find the correct scale, we use that $\tilde{\phi}%
_{\eps }\left( \xi \right) =\Phi _{\eps }\left( L_{\eps
}\xi \right) \sim \frac{\eps }{L_{\eps }}\frac{1}{\left\vert
	\xi \right\vert }$ if $\left\vert \xi \right\vert $ is of order one we
obtain the following approximation: 
\begin{equation}  \label{eq:sec5approx}
\frac{d\xi }{d\tilde{t}}=V\ \ ,\ \ \frac{dV}{d\tilde{t}}=-\eps \left(
L_{\eps }\right) ^{2}\int_{\R^3}\nabla _{\xi }\left( \frac{1%
}{\left\vert \xi -\eta \right\vert }\right) \left[ \rho^+_{\eps
}\left( \eta ,\tilde{t}\right) -\rho^- _{\eps }\left( \eta ,\tilde{t}%
\right) \right] d\eta \ ,
\end{equation}

\begin{eqnarray*}
	&&\partial _{\tilde{t}}f_{\eps }^{\pm }\left( y,w,\tilde{t}\right)
	+w\cdot \nabla _{y}f_{\eps }^{\pm }\left( y,w,\tilde{t}\right) \\
	&&-\left[ \eps \left( L_{\eps }\right) ^{2}\int_{\mathbb{R}%
		^{3}}\nabla _{y}\left( \frac{1}{\left\vert y-\eta \right\vert }\right) (\rho
	_{\eps }^{\pm }-\rho _{\eps }^{\mp })\left( \eta ,\tilde{t}%
	\right) d\eta \pm \frac{\eps }{L_{\eps }}\nabla _{y}\left( 
	\frac{1}{\left\vert y-\xi(\tilde{t}) \right\vert }\right) \right] \cdot
	\nabla _{w}f_{\eps }^{\pm }\left( y,w,\tilde{t}\right) =0.
\end{eqnarray*}

Then, choosing 
\begin{equation}
\eps \left( L_{\eps }\right) ^{2}=1  \label{DebDef}
\end{equation}
the self consistent force term is of the same order of magnitude as the
convective term.

With this choice of scaling, the system \eqref{A2E7}- \eqref{eq:f1} reads: 
\begin{equation}
\begin{aligned} 
&\frac{d\xi }{d\tilde{t}}=V\ \ ,\ \
\frac{dV}{d\tilde{t}}=-\left( L_{\eps }\right) ^{2}\int_{\R^3}\nabla _{X}\Phi \left( L_{\eps }\left\vert \xi -\eta \right\vert \right) \left[ \rho^+_{\eps }\left( \eta ,\tilde{t}\right) -\rho^- _{\eps }\left( \eta ,\tilde{t}\right) \right] d\eta \\ & 0=\partial _{\tilde{t}}f_{\eps }^{\pm }\left( y,w,\tilde{t}\right) +w\cdot \nabla _{y}f_{\eps }^{\pm }\left( y,w,\tilde{t}\right) \\ -&\left( L_{\eps }\right) ^{2}\left[ \int_{\R^3}\nabla _{Y}\Phi \left( L_{\eps }\left\vert y-\eta \right\vert \right) (\rho _{\eps }^{\pm }-\rho _{\eps }^{\mp })\left( \eta ,\tilde{t}\right) d\eta \pm \frac{\eps }{L_{\eps }}\nabla _{Y}\Phi \left( L_{\eps }\left\vert y-\xi \right\vert \right) \right] \cdot \nabla _{w}f_{\eps }^{\pm }\left( y,w,\tilde{t}\right) .
\end{aligned}
\label{eq:system1}
\end{equation}%
where we use the fact that $\Phi =\Phi \left( X\right) $ and we denote as $%
\nabla _{X}\Phi $ the usual gradient with respect to the variable $X.$

The system \eqref{eq:system1} must be solved with the initial conditions (%
\ref{S8E9})-(\ref{S8E9b}), where the average $\int_{\R^3}\left[
f^+_{\eps }\left( y,w,0\right) -f^-_{\eps }\left( y,w,0\right) %
\right] dw$ vanishes. This justifies the system  \eqref{eq:appCoul}-\eqref{T1E3bis}, \eqref{A1E5} as claimed. 

\color{black}

\bigskip

\subsection{The case of grazing collisions}

  The goal of this section is to justify the system of equations \eqref{T4E7}-\eqref{T4E7bis}, \eqref{T4E4} starting from the Newton equations for the Rayleigh Gas (cf.~\eqref{eq:RayNew}) in the grazing collisions regime. We will assume that there is only one type of
scatterers. More precisely  
\begin{align} 
\frac{dX}{d\tau }& =V\ \ , \ \ 
\frac{dV}{d\tau }=-\sum_{j\in S}\nabla_{X}\Phi _{\eps } \left( X-Y_{j} \right) \label{T3E8} \\
\frac{dY_{k}}{d\tau }& =W_{k}\ \ ,\ \ \frac{dW_{k}}{d\tau }=-\nabla _{Y}\phi
_{\eps }\left( Y_{k}-X\right) \ \ ,\ \ k\in S 
\end{align}
with the interaction potential  $\Phi _{\eps }\left( X\right) =\eps \Phi \left( \frac{%
\left\vert X\right\vert }{\ell _{\eps }}\right) $ . 
Here $\Phi =\Phi \left( s\right) $ is a smooth function which decreases
sufficiently fast as $s\rightarrow \infty ,$ together with 
the asymptotic formulas for its derivatives.  We can assume for
instance that $\Phi \left( s\right) \sim \frac{1}{s^{\alpha }}$ as $%
s\rightarrow \infty $ with $\alpha >1,$ or that $\Phi \left( s\right) $
decreases exponentially. We assume that $\ell _{\eps }\lesssim 1.$ In
particular we might assume that $\ell _{\eps }\rightarrow 0.$ We then
argue as in Subsections \ref{ssec:Jfinrange}, \ref{ssec:JCoulomb} and
introduce a new set of variables by means of:%
\begin{equation*}
y_{k}=\frac{Y_{k}}{L_{\eps }}\ \ ,\ \ \xi =\frac{X}{L_{\eps }}%
\ \ ,\ \ \tilde{t}=\frac{\tau }{L_{\eps }} \ .
\end{equation*}
The length $L_{\eps}$ is only an auxiliary length $L_{\eps}
\gg 1$ that enables us to approximate the background distribution by a
Gaussian. Notice that in this case it is not convenient to take $%
L_{\eps}=\ell_{\eps}$ because, since $\ell_{\eps}%
\lesssim1$, we would not obtain an approximation of the distribution of
particles by means of a random field. Consider again the empirical
density defined by \eqref{eq:feps}. We also define a rescaled potential by
means of 
\begin{equation*}
\Phi_{\eps}\left( X\right) =\tilde{\phi}_{\eps}\left( \frac{X}{%
L_{\eps}}\right) =\tilde{\phi}_{\eps}\left( \xi\right) \ .
\end{equation*}
Using the change of variables $\tilde{t}=\frac{\tau }{L_{\eps }},\ \xi =%
\frac{X}{L_{\eps }}$ we can rewrite (\ref{T3E8}) as:%
\begin{equation*}
\frac{d\xi }{d\tilde{t}}=V,\quad \frac{dV}{d\tilde{t}}=-\sum_{j\in S}\nabla
_{\xi }\tilde{\phi}_{\eps }\left( \xi -y_{j}\right) =-\left(
L_{\eps }\right) ^{3}\int_{\R^3}\nabla _{\xi }\tilde{\phi}
_{\eps }\left( \xi -\eta \right) \rho _{\eps }\left( \eta ,
\tilde{t}\right) d\eta,
\end{equation*}
where $\rho _{\eps }\left( y,\tilde{t}\right) =\rho \lbrack
f_{\eps }(\cdot ,\tilde{t})](y)$ as defined in \eqref{eq:Spdensity}.
On the other hand, using the equations for the scatterers in (\ref{T3E8}) we
obtain: 
\begin{equation}
\partial_{\tilde{t}}f_{\eps}\left( y,w,\tilde{t}\right) +w\cdot
\nabla_{y}f_{\eps}\left( y,w,\tilde{t}\right) -\nabla_{y}\tilde{\phi }%
_{\eps}\left( y-\xi\right) \cdot\nabla_{w}f_{\eps}\left( y,w,%
\tilde{t}\right) =0 \ . \label{T4E2}
\end{equation}

Arguing as in the previous Subsections it is
natural to require that the initial value for (\ref{T4E2}) satisfies
\begin{equation}
\mathbb{E}\left[ \left( f_{\eps }\left( y_{a},w_{a},0\right) -g\left(
w_{a}\right) \right) \left( f_{\eps }\left( y_{b},w_{b},0\right)
-g\left( w_{b}\right) \right) \right] =\frac{g\left( w_{a}\right) }{\left(
L_{\eps }\right) ^{3}}\delta \left( y_{a}-y_{b}\right) \delta \left(
w_{a}-w_{b}\right)  \label{T4E3}
\end{equation}%
where $g\left( w\right) $ is the distribution of velocities of the
scatterers $\left\{ Y_{k}\right\} _{k\in S}.$ We then look for a solution (%
\ref{T4E2}), (\ref{T4E3}) with the form:%
\begin{equation*}
f_{\eps }\left( y,w,\tilde{t}\right) =g\left( w\right) +\frac{1}{%
\left( L_{\eps }\right) ^{\frac{3}{2}}}\zeta _{\eps }\left(
y,w,\tilde{t}\right) \ .
\end{equation*}
Then, approximating (\ref{T4E2}) by the leading order terms we obtain \eqref{T4E4}, namely
\begin{equation*}
(\partial_{\tilde{t}}+w\cdot\nabla_y)\zeta_{\eps}\left( y,w,\tilde{t}%
\right) -\left( L_{\eps}\right) ^{\frac{3}{2}}\nabla_{y}\tilde{\phi}%
_{\eps }\left( y-\xi\right) \cdot\nabla_{w}g\left( w\right) =0, \ \  \ 
\zeta_{\eps}\left( y,w,0\right) =N_{\eps}\left( y,w\right)
\end{equation*}
where the (random) initial data is characterized by (cf.~\eqref{T4E6})
\begin{align*}
\mathbb{E}\left[ N_{\eps}\left( y,w\right) \right] &=0  
\nonumber \\
\mathbb{E}\left[ N_{\eps}\left( y_{a},w_{a}\right) N_{\eps
}\left( y_{b},w_{b}\right) \right] &=\frac{g\left( w_{a}\right) }{\left(
L_{\eps}\right) ^{3}}\delta\left( y_{a}-y_{b}\right) \delta\left(
w_{a}-w_{b}\right).  
\end{align*}
On the other hand, the evolution of the tagged particle (cf. (\ref{T3E8}))
can be approximated as: 
\begin{equation*}
\frac{d\xi }{d\tilde{t}}=V_0\ \ ,\ \ \frac{dV}{d\tilde{t}}=-\left(
L_{\eps }\right) ^{3}\int_{\R^3}\int_{\R^3}\nabla
_{\xi }\tilde{\phi}_{\eps }\left( \xi -\eta \right) \zeta
_{\eps }\left( \eta ,w,\tilde{t}\right) dwd\eta,  
\end{equation*}
where we used that $V(t)\approx V_0$ for short macroscopic times. 

We further notice that in the  case of grazing collisions the justification of the approximating dynamics starting from the interacting particle system works analogously.


\section{Proofs} \label{Sec:Pfs}

\subsection{Interactions with finite range much larger than the particle distance\label{RaylCompSuppPot}}
We now prove Theorem \ref{thm:FiniteRange}. We will first prove the result for the Rayleigh Gas systems and later for interacting particle systems. \medskip

\begin{proofof}[Proof of Theorem~\ref{thm:FiniteRange}]
	\textit{Part I: Rayleigh Gas} \medskip

	Due to the linearity of (\ref{S7E1})-(\ref{eq:zetaFiniteRayleigh}) we can write 
	$$\zeta\left( y,w,\tilde{t}\right) = \zeta_{1}\left( y,w,\tilde{t}\right) +\zeta
	_{2}\left( y,w,\tilde{t}\right) ,$$ where:%
	\begin{equation}
	\partial_{\tilde{t}}\zeta_{1}\left( y,w,\tilde{t}\right) +w\cdot\nabla
	_{y}\zeta_{1}\left( y,w,\tilde{t}\right) =0\ \ ,\ \ \zeta_{1}\left(
	y,w,0\right) =N\left( y,w\right)  \label{S7E2}
	\end{equation}%
	\begin{equation}
	\partial_{\tilde{t}}\zeta_{2}\left( y,w,\tilde{t}\right) +w\cdot\nabla
	_{y}\zeta_{2}\left( y,w,\tilde{t}\right)
	-\theta\nabla_{y}\Phi\left( y-\xi\right)
	\cdot\nabla_{w}g\left( w\right) =0\ \ ,\ \ \zeta _{2}\left( y,w,0\right) =0.
	\label{S7E3}
	\end{equation}
	Equation (\ref{S7E2}) can be solved explicitly using characteristics:%
	\begin{equation}
	\zeta_{1}\left( y,w,\tilde{t}\right) =N\left( y-w\tilde{t},w\right).
	\label{S7E4}
	\end{equation}
	
	The stochastic process in (\ref{S7E4}) has non
	trivial correlations in time. On the other hand, it is easily seen using (%
	\ref{S7E1}) that for any $\tilde{t}\in \mathbb{R}$ the stochastic process $%
	\zeta _{1}\left( y,w,\tilde{t}\right) $ is the same in law as $N\left(
	y,w\right) .$
	
	Thanks to \eqref{eq:MotionFiniteCombined} we have that $V_0$
	is a constant and $\xi \left( \tilde{t}\right) =\xi \left( 0\right) +V_0
	\tilde{t}$. Assuming that $\xi \left( 0\right) =0$
	without loss of generality, we can write (\ref{S7E3}) as:%
	\begin{equation*}
	\partial _{\tilde{t}}\zeta _{2}\left( y,w,\tilde{t}\right) +w\cdot \nabla
	_{y}\zeta _{2}\left( y,w,\tilde{t}\right) -\theta\nabla
	_{y}\Phi \left( y-V_0)\tilde{t}\right) \cdot \nabla _{w}g\left( w\right) =0\ \
	,\ \ \zeta _{2}\left( y,w,0\right) =0,
	\end{equation*}
	whose solution is given by: 
	\begin{equation}
	\zeta_{2}\left( y,w,\tilde{t}\right) =\theta\nabla_{w}g\left(
	w\right) \cdot\int_{0}^{\tilde{t}}\nabla_{y}\Phi\left( y-w\left( \tilde {t}%
	-s\right) -V_0s\right) ds .  \label{S7E5}
	\end{equation}
	
	Using the decomposition into $\zeta_{1}, \zeta_2$ (cf.~\eqref{S7E2}-\eqref{S7E3}) 
	we can decompose the force acting on the tagged particle in \eqref{eq:MotionFiniteCombined}
	two pieces. The one associated to $\zeta_{1}$ is a time-dependent random
	force field which is not affected by the tagged particle. The term
	associated to $\zeta_{2}$ yields a deterministic term which depends only on
	the velocity of the tagged particle. Let $\tilde{\rho}_{1,2 }\left( y,\tilde{%
		t}\right) = \rho[\zeta_{1,2}(\cdot,\tilde{t})](y)$, then $\tilde{\rho}%
	_{1}\left( y,\tilde{t}\right) $ is a Gaussian random variable which can be
	characterized by means of the following set of expectations and correlations
	(cf.~Appendix \ref{appA}):%
	\begin{equation*}
	\mathbb{E}\left[ \tilde{\rho}_{1}\left( y,\tilde{t}\right) \right] =0
	\end{equation*}%
	\begin{equation}
	\mathbb{E}\left[ \tilde{\rho}_{1}\left( y_{1},\tilde{t}_{1}\right) \tilde{%
		\rho}_{1}\left( y_{2},\tilde{t}_{2}\right) \right] =\frac{1}{\left( \tilde{t}%
		_{1}-\tilde{t}_{2}\right) ^{3}}g\left( \frac{y_{1}-y_{2}}{\tilde{t}_{1}-%
		\tilde{t}_{2}}\right) \ \ ,\ \ \ \tilde{t}_{1}>\tilde{t}_{2} \ . \label{S7E7}
	\end{equation}
	
	
	We can then rewrite \eqref{eq:MotionFiniteCombined} as 
	\begin{align}
		\theta F_g(\xi, \tilde{t};\omega)&= -\theta \int_{\mathbb{R}%
		^{3}}\nabla _{\xi }\Phi \left( \xi -\eta \right) \tilde{\rho}_{1}\left( \eta
	,\tilde{t}\right) d\eta -\theta \int_{\R^3}\nabla
	_{\xi }\Phi \left( \eta \right) \tilde{\rho}_{2}\left( \eta +V_0\tilde{t},%
	\tilde{t}\right) d\eta \nonumber \\
		&= \theta F_{g}^{(1)}\left( \xi ,\tilde{t}\right) + \theta ^2 H_{g}\left(\tilde{t};V_0\right) . \label{S8E1a}
	\end{align}
	For $\tilde{t} \rightarrow \infty$, the friction term $H_g$ converges to:
	\begin{align}
	 \lim_{\tilde{t}\to \infty}\theta ^2 H_{g}\left(\tilde{t};V_0\right)   &=
	- \theta^2 \Lambda _{g}\left( V_0\right), \quad 
	\text{with}  \notag \\
	\Lambda _{g}\left( V_0\right) & =-\int_{\R^3}d\eta \int_{\mathbb{R}%
		^{3}}dw\nabla _{\eta }\Phi \left( \eta \right) \left[ \nabla _{w}g\left(
	w\right) \cdot \int_{-\infty }^{0}\nabla _{y}\Phi \left( \eta +\left(
	w-V_0\right) s\right) ds\right] \nonumber\\
	&=-\pi \int_{\R^{3}} \int_{\R^{3}} (k\otimes k) |\hat{\Phi}(k)|^2 \delta (k\cdot(v-w))\nabla g(w) \ud{w} \ud{k} .  \label{S8E1}
	\end{align}	
	The random force field $F_{g}^{(1)}\left( \xi \right) $ is a Gaussian field with
	zero average. The correlations between the values of the field at different
	points can be computed using (\ref{S7E7}) and (\ref{S8E1a}). We have: 
	\begin{equation*}
	\mathbb{E}\left[ F_{g}^{(1)}\left( \xi _{1},\tilde{t}_{1}\right) \otimes
	F_{g}^{(1)}\left( \xi _{2},,\tilde{t}_{2}\right) \right] =\int_{\mathbb{R}%
		^{3}}d\eta _{1}\int_{\R^3}d\eta _{2} \frac{\nabla \Phi \left( \eta
		_{1}\right) \otimes \nabla \Phi \left( \eta _{2}\right) }{\left( \tilde{t}%
		_{1}-\tilde{t}_{2}\right) ^{3}}g\left( \frac{\eta _{1}+\xi _{1}-\eta
		_{2}-\xi _{2}}{\tilde{t}_{1}-\tilde{t}_{2}}\right).
	\end{equation*}
	For $\mathcal{F}_g(x,t):=F_g(\xi,\tilde{t})$ this yields for $v=V_0$:
	\begin{align} \label{eq:Rayclaim1}
		 \lim_{\theta\to 0}\mathbb{E}\left[ (\mathcal{F}_{g}\left( v t_1,{t}_{1}\right)+\Lambda_{g}(v)) \otimes
		(\mathcal{F}_{g}\left( v t_2, {t}_{2}\right)+\Lambda_{g}(v)) \right] = 	D_{g}(v) \delta(t_1-t_2), 
		\\ 
		D_{g}\left( v\right) =\int_{0}^{\infty }\frac{ds}{s^{3}}\int_{\mathbb{R}%
		^{3}}d\eta _{1}\int_{\R^3}d\eta _{2}\nabla _{\eta }\Phi \left(
	\eta _{1}\right) \otimes \nabla _{\eta }\Phi \left( \eta _{2}\right) g\left(
	v+\frac{\eta _{1}-\eta _{2}}{s}\right) \nonumber \\
	=\pi \int_{\R^3}\int_{\R^3} dk_1 dk_2 (k \otimes k) |\hat{\Phi}(k)|^2 \delta(k\cdot (v-w)) g(w) \ud{w}, \label{eq:Rayclaim2}
	\end{align} 
	where in the last line we have made use of Plancherel's identity.
	This concludes the proof of Theorem~\ref{thm:FiniteRange} for the Rayleigh gas case.
\bigskip

\noindent \textit{Part II: Interacting particle systems} \medskip

We can then decompose $\zeta$ in (\ref{S7E1})-(\ref{eq:zetaFiniteRayleigh}) as follows: 
\begin{equation}
\zeta \left( y,w,\tilde{t}\right) =\zeta _{1}\left( y,w,\tilde{t}\right)
+\zeta _{2}\left( y,w,\tilde{t}\right) .  \label{S9E7}
\end{equation}%
Here $\zeta _{1}$ solves
\begin{align}
(\partial _{\tilde{t}}+w\cdot \nabla _{y}\zeta _{1})\left( y,w,\tilde{t}%
\right) -\sigma \nabla _{w}g\left( w\right) \cdot (\nabla _{y}\Phi \ast 
\tilde{\rho}_{1})\left( \eta ,\tilde{t}\right) d\eta & =0  \label{T1E5} \\
\zeta _{1}\left( y,w,0\right) & =N\left( y,w\right)  \label{T1E6}
\end{align}%
with $\tilde{\rho}_{1}\left( \eta ,\tilde{t}\right) =\rho \lbrack \zeta
_{1}(\cdot ,\tilde{t})](y)$. On the other hand, $\zeta _{2}$
solves: 
\begin{align}
(\partial _{\tilde{t}}+w\cdot \nabla _{y})\zeta _{2}\left( y,w,\tilde{t}%
\right) -\sigma \nabla _{w}g\left( w\right) \cdot (\nabla _{y}\Phi \ast 
\tilde{\rho}_{2})(y,\tilde{t})& =\theta \nabla _{y}\Phi \left( y-V_0\tilde{t}\right) \cdot \nabla
_{w}g\left( w\right)  \label{T1E8} \\
\zeta _{2}\left( y,w,0\right) & =0  \label{T1E9}
\end{align}%
with $\tilde{\rho}_{2}\left( \eta ,\tilde{t}\right) =\rho \lbrack \zeta
_{2}(\cdot ,\tilde{t})](y)$.

The set of equations (\ref{S9E7})-(\ref{T1E9}) yields the fluctuations of
the scatterer density in the phase space. The contribution $\zeta _{1}$
contains the ``noisy" part of the fluctuations. The contribution $\zeta
_{2} $ yields the perturbation to the scatterers density induced by the
presence of the distinguished particle $\left( X,V\right)$. It is crucial
to notice that if $\sigma $ is of order one, the resulting densities $\zeta
_{1},$ $\zeta _{2}$ would be different from those obtained for Rayleigh
gases (cf.~(\ref{S7E4}), (\ref{S7E5})). The terms proportional to $\sigma $
give the contribution due to the interactions of the scatterers with
themselves. The problems (\ref{T1E5})-(\ref{T1E6}) and (\ref{T1E8})-(\ref%
{T1E9}) are linear and can be solved using Fourier and Laplace transforms
(cf.~\cite{La, LL2}).
In order to solve the problems (\ref{T1E5})-(\ref{T1E6}) and (\ref{T1E8})-(%
\ref{T1E9}) we define a fundamental solution associated to the operator on
the left hand side of (\ref{T1E5}), (\ref{T1E8}). We define the function $%
G_{\sigma }\left( y,w,w_{0},\tilde{t}\right) $ as the solution of the
problem:%
\begin{align}
(\partial _{\tilde{t}}+w\cdot \nabla _{y}G_{\sigma })\left( y,w,w_{0},\tilde{%
	t}\right) -\sigma \nabla _{w}g\left( w\right) \cdot \int_{\mathbb{R}%
	^{3}}\nabla _{y}\Phi \left( y-\eta \right) \Xi _{\sigma }\left( \eta ,w_{0},%
\tilde{t}\right) d\eta & =0  \label{T2E2} \\
G_{\sigma }\left( y,w,w_{0},0\right)  =\delta \left( y\right) &\delta \left(
w-w_{0}\right),  \label{T2E4}
\end{align}%
where $\Xi _{\sigma }\left( y,w_{0},\tilde{t}\right) =\rho \lbrack G_{\sigma
}\left( \cdot ,w_{0},\tilde{t}\right) ](y)$. 
We can then write $\zeta_{1},\ \zeta_{2}$ in terms of $G_{\sigma}$ as:%
\begin{align}
\zeta_{1}\left( y,w,\tilde{t}\right) &=\int_{\R^3}d\eta \int_{%
	\R^3}dw_{0}G_{\sigma}\left( y-\eta,w,w_{0},\tilde{t}\right)
N\left( \eta,w_{0}\right)  \label{T2E5} \\
\zeta_{2}\left( y,w,\tilde{t}\right) &=\theta
	\int_{0}^{\tilde{t}}\int_{\R^3}\int_{\mathbb{R}%
	^{3}}G_{\sigma}\left( y-\eta,w,w_{0},\tilde{t}-s\right) \left[
\nabla_{y}\Phi\left( \eta-V_0 s\right) \cdot\nabla_{w}g\left( w_{0}\right) %
\right] d\eta dw_{0}ds \ . \label{T2E6}
\end{align}
Using now (\ref{S9E7}), (\ref{T2E5}), (\ref{T2E6}) we can rewrite $F_{g}$ in \eqref{eq:MotionFiniteCombined} as 
\begin{equation}
\theta {F}_{g}\left( \xi ,\tilde{t}\right) = \theta
{F}^{(1)}_{g}\left( \xi ,\tilde{t}%
\right) +\theta^{2} H_{g}\left( 
\tilde{t};V_0\right)  \label{eq:med1}
\end{equation}%
where:%
\begin{align*}
{F}^{(1)}_{g}\left( \xi ,\tilde{t}\right) & =\int_{(\R^3)^3}\nabla
_{y}\Phi \left( y-\xi \right) \Xi _{\sigma }\left( y-\eta ,w_{0},\tilde{t}\right) N\left( \eta
,w_{0}\right) \ud{y} \ud{w_0} \ud{\eta} \\
H_{g}\left( \tilde{t};V_0\right) & =\int_{0}^{\tilde{t}}ds\int_{(\R^3)^3}\nabla \Phi
\left( y-V_0\tilde{t}\right)  \Xi _{\sigma }\left( y-\eta ,w_{0},\tilde{t}-s\right) \left[
\nabla\Phi \left( \eta -V_0s\right) \cdot \nabla _{w}g\left(
w_{0}\right) \right] d\eta dw_{0}dy,
\end{align*}
which describe the friction and the random force field  terms. 

We are interested in the dynamics given by \eqref{eq:MotionFiniteCombined} with $F_g$ as in \eqref{eq:med1}, in times
of the order of the mean free time between collisions. In those times we
have $\tilde{t}\gg 1.$ Therefore it is natural to study the asymptotic
behaviour of ${F}^{(1)}_{g}\left( \xi ,\tilde{t}\right) $ and the friction
coefficient $H_{g}\left( \tilde{t};V_0\right) $ as $\tilde{t}\rightarrow
\infty .$ We define $ \lim_{\tilde{t}\to \infty} H_{g}\left(\tilde{t};V_0\right)   =	- \Lambda _{g}\left( V_0\right)$. Then 
\begin{equation}
-\Lambda_g\left(V_0\right) =\int_{(\R^3)^3}\nabla\Phi \left( y\right) \int_{0}^{\infty }ds \Xi_{\sigma }\left( y-\eta
,w_{0},s\right) \left[ \nabla\Phi \left( \eta +V_0s\right) \cdot
\nabla _{w}g\left( w_{0}\right) \right] d\eta dw_{0}dy.  \label{T2E8}
\end{equation}
Notice that since $G_{\sigma }$  decays sufficiently fast in $\tilde{t}$, the friction term is well defined by means of (\ref{T2E8}%
). 
%
%

We compute now the statistical properties of the noise term. Notice that $F^{(1)}_{g,\infty}\left( \xi ,\tilde{t}\right) $ is a
Gaussian noise with zero average. Therefore, in order to characterize $%
F^{(1)}_{g,\infty}\left( \xi ,\tilde{t}\right) $ we just need to compute the covariance
function which is given by:%
\begin{align*}
\mathbb{E}& \left[ F^{(1)}_{g,\infty}\left( \xi _{1},\tilde{t}_{1}\right) F^{(1)}_{g,\infty}\left( \xi
_{2},\tilde{t}_{2}\right) \right] =\lim_{T\rightarrow \infty }\int_{\mathbb{R%
	}^{3}}\int_{\R^3}\int_{\R^3}\int_{\R^3}\nabla
_{y}\Phi \left( y_{1}\right) \nabla _{y}\Phi \left( y_{1}\right)
dy_{1}dy_{2}dw_{1}dw_{2}\cdot \\
& \cdot \int_{\R^3}d\eta _{1}\int_{\R^3}dw_{0,1}\int_{%
	\R^3}d\eta _{2}\int_{\R^3}dw_{0,2}G_{\sigma }\left(
y_{1}-\eta _{1},w_{1},w_{0,1},\tilde{t}_{1}+T\right) G_{\sigma }\left(
y_{2}-\eta _{2},w_{2},w_{0,2},\tilde{t}_{2}+T\right) \cdot \\
& \cdot \mathbb{E}\left[ N\left( \eta _{1}+\xi _{1},w_{0,1}\right) N\left(
\eta _{2}+\xi _{2},w_{0,2}\right) \right] .
\end{align*}%
Using (\ref{S7E1}) and evaluating the integrals using the dirac masses, we
obtain%
\begin{equation}
\begin{aligned} &\mathbb{E}\left[ F^{(1)}_{g,\infty}\left( \xi _{1},\tilde{t}_{1}\right)
\otimes F^{(1)}_{g,\infty}\left( \xi _{2},\tilde{t}_{2}\right) \right]
=\lim_{T\rightarrow \infty }\int_{(\R^3)^2}\nabla _{y}\Phi \left(
y_{1}\right) \otimes \nabla _{y}\Phi \left( y_{2}\right) dy_{1}dy_{2}\cdot
\\ &\cdot \int_{(\mathbb{R}^3)^2}d\eta_1 dw_{0,1}\Xi_\sigma \left(
y_{1}-\eta _{1},w_0,\tilde{t}_{1}+T\right) \Xi_\sigma \left( y_{2}-\eta
_{1}+\xi _{2}-\xi _{1},w_0,\tilde{t} _{2}+T\right) g\left( w_0\right) ,
\end{aligned}  \label{eq:FFCorr}
\end{equation}%
where $\Xi _{\sigma}$ is the velocity marginal of $G_{\sigma }$ defined by 
\begin{equation}\label{def:Xisigma}
\Xi _{\sigma }(y,w_{0},t)=\int_{{\mathbb{R}}^{3}}G_{\sigma }(y,w,w_{0},t)\;%
\mathrm{d}{w}.
\end{equation}

%


The simplest way of solving (\ref{T2E2})-(\ref{T2E4}) is applying Fourier in
the $x$ variable and Laplace in the time variable $\tilde{t},$ as it was
made in \cite{La, LL2}. We define the following Fourier-Laplace transform:%
\begin{equation*}
\tilde{F}\left( k,v,z\right) =\frac{1}{\left( 2\pi \right) ^{\frac{3}{2}}}%
\int_{0}^{\infty }dz\int_{\R^3}dxe^{-ix\cdot k}e^{-z\tilde{t}%
}F\left( x,v,\tau \right)\  .
\end{equation*}

As it will be proved in Lemma \ref{Lemma:FouLap} below, the fundamental solution $G_\sigma$ and its marginal $\Xi _{\sigma }(y,w_{0},t)$ can be explicitly represented as
\begin{align}
	\tilde{G}_{\sigma }\left( k,w,w_{0},z\right) &=\frac{\sigma \left( \nabla
		_{w}g\left( w\right) \cdot ik\right) \hat{\Phi}\left( k\right) }{\left( z+i
		k\cdot w \right) \left( z+i k\cdot w_{0} \right) \Delta _{\sigma }\left(
		k,z\right) } +\frac{\delta \left( w-w_{0}\right) }{\left( 2\pi \right) ^{%
			\frac{3}{2}}\left( z+i\left( k\cdot w_{0}\right) \right) } \label{eq:GSigma}\\
		\tilde{\Xi}_{\sigma}\left( k,w_{0},z\right) &=\frac{1}{\left( 2\pi\right) ^{%
				\frac{3}{2}}\left( z+i\left( k\cdot w_{0}\right) \right) \Delta_{\sigma
			}\left( k,z\right) }  \label{eq:Xirepresentation}\\
		\Delta_{\sigma}\left( k,z\right) &=1- \sigma(2\pi)^\frac32 \hat{\Phi}\left(
		k\right) \int_{\R^3}\frac{\left( ik\cdot\nabla_{w}g\left( w\right)
			\right) }{z+i\left( k\cdot w\right) }dw .  \label{T3E5}
	\end{align}

In order to avoid the exponential growth of the disturbances we need the
following stability condition:%
\begin{equation}
\Delta _{\sigma }\left( k,z\right) \neq 0\text{ for }\func{Re}\left(
z\right) \geq 0\text{ and any }k\in \R^3.  \label{PenrStab}
\end{equation}

We remark that in case of a smooth decaying potential $\phi $, this
criterion can be satisfied choosing for example $0<\sigma \ll 1$ small.

Notice that in the formulas of $\tilde{F}_{g}\left( \xi ,\tilde{t}\right) $
and $H_{g}\left( \tilde{t};V\right) $ we only make use of $\Xi_{\sigma}$.

Inverting the Fourier-Laplace transform in \eqref{eq:Xirepresentation}, we can rewrite $\Xi _{\sigma}$ as  
\begin{equation}  \label{eq:Gsigma}
\Xi_{\sigma}\left( x,w_{0},\tilde{t}\right) =\frac{1}{2\pi i}\int_{\mathbb{R}%
	^{3}}\frac{e^{ik\cdot x}dk}{\left( 2\pi \right) ^{3}}\int_{\gamma }\frac{e^{z%
		\tilde{t}}dz}{\left( z+i\left( k\cdot w_{0}\right) \right) }\frac{1}{\Delta
	_{\sigma }\left( k,z\right) }.
\end{equation}

We can consider possible analyticity properties of the function $\Delta
_{\sigma }\left( k,z\right) $ \eqref{T3E5} in the variable $z,$ in order to
obtain a possible decay of the integral above in time.

We rewrite the function in terms of the Radon transform $H(s,\theta)$ of $g$, which for $\theta\in S^2$ and $s\in {\mathbb{R}}^3$ is defined as 
\begin{equation*}
H\left( s;\theta \right) =\int_{\left\{ \theta \cdot w=s\right\} }g\left(
w\right) dS\left( w\right).
\end{equation*}
With this definition, the integral appearing in \eqref{T3E5} can be
rewritten as 
\begin{equation*}
\int_{\R^3}\frac{\left( i\theta \cdot \nabla _{w}g\left( w\right)
	\right) }{\zeta +i\left( \theta \cdot w\right) }dw=i\int_{-\infty }^{\infty }%
\frac{ds}{\zeta +is}\frac{\partial H\left( s;\theta \right) }{\partial s}.
\end{equation*}

We will assume that $H\left( s;\theta \right) $ is analytic in a strip $%
\left\{ \left\vert \func{Im}\left( s\right) \right\vert <\delta _{0}\right\} 
$ for all the values of $\theta \in S^{2}.$ Then $\Delta _{\sigma }\left(
k,z\right) $ would be analytic in a strip $\left\{ \left\vert \func{Im}%
\left( \zeta \right) \right\vert <\delta _{0}\right\} .$ Initially we assume
that $\func{Re}\left( \zeta \right) >0.$ We can move the line of integration
from $\mathbb{R}$ to $\mathbb{R}-\frac{\delta _{0}}{2}i.$ It then follows
that 
\begin{equation}  \label{Psidef}
\Psi \left( \zeta ;\theta \right) := \int_{\R^3}\frac{\left( {i}\theta
	\cdot \nabla _{w}g\left( w\right) \right) }{\zeta +i\left( \theta \cdot
	w\right) }dw={i}\int_{\mathbb{R}-\frac{\delta _{0}}{2}i}\frac{ds}{\zeta +is}%
\frac{\partial H\left( s;\theta \right) }{\partial s}
\end{equation}
is analytic for $\func{Re}\left( \zeta \right) >-\frac{\delta _{0}}{2}.$

\bigskip

Now suppose that for (a possibly smaller value) $\delta _{0}>0$ we can
further ensure that $\Delta _{\sigma }\left( k,z\right) $ does not vanish
for $\func{Re}\left( \zeta \right) >-\frac{\delta _{0}}{2}.$ Then the
function 
\begin{equation*}
\Delta _{\sigma }\left( k,z\right) =1-\sigma {(2\pi)^{\frac{3}{2}}}\hat{\Phi}\left( k\right) \Psi
\left( \frac{z}{\left\vert k\right\vert };\frac{k}{\left\vert k\right\vert }%
\right)
\end{equation*} 
is analytic in $z$ in a region of the form $\left\{ \func{Re}\left( z\right)
\geq -\delta _{1}\left\vert k\right\vert \right\} .$

With this we return to the friction term $\Lambda_g\left(V_0\right)$ given by \eqref{T2E8}.
Inserting the Fourier-Laplace representation of $\Xi_\sigma$ we obtain 
\begin{equation*}
-\Lambda_g\left(V_0\right)  =\int_{(\mathbb{R}^3)^4}\int_{0}^{\infty}ds%
\left[ \frac{1}{2\pi i}\frac{e^{ik\cdot \left( y-\eta \right) }}{\left( 2\pi
	\right) ^{3}}\int_{\gamma }dz \frac{e^{zs} \nabla _{y}\Phi \left( y\right) %
	\left[ \nabla _{\eta }\Phi \left( \eta +V_0s\right) \cdot \nabla _{w}g\left(
	w_{0}\right) \right]}{\left( z+i\left( k\cdot w_{0}\right) \right) \Delta
	_{\sigma }\left( k,z\right) }\right] d\eta dw_{0}dy \ .
\end{equation*}
The integrals in $y,\eta$ are explicit and yield by definition of $\Psi$ \eqref{Psidef}:
\begin{eqnarray*}
	-\Lambda_g\left(V_0\right)  = \int_{0}^{\infty}e^{ik\cdot V_0s}ds\int_{%
		\R^3}k\left\vert \hat{\Phi}\left(
	k\right) \right\vert ^{2}dk\frac{1}{2\pi i}\int_{\gamma }\frac{e^{zs}dz}{%
		\Delta _{\sigma }\left( k,z\right) }\left[ \Psi \left( \frac{z}{\left\vert
		k\right\vert };\theta \right) \right] \ .
\end{eqnarray*}

The analyticity properties of $\Delta _{\sigma }\left( k,z\right) $ and $%
\Psi \left( \frac{z}{\left\vert k\right\vert };\theta \right) $ allow to
deform $\gamma $ to a contour contained in $\func{Re}\left( z\right) <0.$ We
then obtain exponential decay of the integral as $s\rightarrow \infty .$
This gives 
\begin{equation}
-\Lambda_g\left(V_0\right) =-\int_{\R^3}%
(k \otimes k) \left\vert \hat{\Phi}\left( k\right)
\right\vert ^{2} \frac{ \delta(k(V_0-w)) }{|\Delta_\sigma(k,-ikV_0)|^2}\nabla g(w) \ud{w}\ud{k} \label{FricBLFastDec}.
\end{equation}
This is the friction coefficient associated to Balescu-Lenard. 

We now use the Fourier-Laplace representation of $\Xi _{\sigma }$ (cf.~\eqref{eq:Gsigma}) to compute the time correlation of the forces \eqref{eq:FFCorr}. Specifically:
\begin{align*}
\mathbb{E}\left[ F_{g}\left( \xi _{1},\tilde{t}_{1}\right)\otimes F_{g}\left( \xi
_{2},\tilde{t}_{2}\right) \right] =& \lim_{T\rightarrow \infty }\int_{%
	\R^3}\int_{\R^3}\nabla _{y}\Phi \left( y_{1}\right)
\nabla _{y}\Phi \left( y_{2}\right) dy_{1}dy_{2}\int_{\R^3}d\eta
_{1}\int_{\R^3}g\left( w_{0}\right) dw_{0}\cdot \\
\cdot & \frac{1}{2\pi i}\int_{\R^3}\frac{e^{ik_{1}\cdot \left(
		y_{1}-\eta _{1}\right) }dk_{1}}{\left( 2\pi \right) ^{3}}\int_{\gamma }\frac{%
	e^{z_{1}\left( \tilde{t}_{1}+T\right) }dz_{1}}{\left( z_{1}+i\left(
	k_{1}\cdot w_{0}\right) \right) }\frac{1}{\Delta _{\sigma }\left(
	k_{1},z_{1}\right) }\cdot \\
& \cdot \frac{1}{2\pi i}\int_{\R^3}\frac{e^{ik_{2}\cdot \left(
		y_{2}-\eta _{1}+\xi _{2}-\xi _{1}\right) }dk_{2}}{\left( 2\pi \right) ^{3}}%
\int_{\gamma }\frac{e^{z_{2}\left( \tilde{t}_{2}+T\right) }dz_{2}}{\left(
	z_{2}+i\left( k_{2}\cdot w_{0}\right) \right) }\frac{1}{\Delta _{\sigma
	}\left( k_{2},z_{2}\right) }.
\end{align*}
In Fourier variables we obtain
\begin{align}
\mathbb{E}& \left[ F_{g}\left( \xi _{1},\tilde{t}_{1}\right) \otimes
F_{g}\left( \xi _{2},\tilde{t}_{2}\right) \right] =\lim_{T\rightarrow \infty
}\int_{\R^3}g(w_{0})dw_{0}\int_{\R^3}k\otimes
k\left\vert \hat{\Phi}\left( k\right) \right\vert ^{2}e^{-ik\cdot \left( \xi
	_{2}-\xi _{1}\right) }dk\cdot \nonumber \\
& \frac{1}{2\pi i}\int_{\gamma }\frac{e^{z_{1}\left( \tilde{t}_{1}+T\right)
	}dz_{1}}{\left( z_{1}+i\left( k\cdot w_{0}\right) \right) }\frac{1}{\Delta
	_{\sigma }\left( k,z_{1}\right) }\frac{1}{2\pi i}\int_{\gamma }\frac{%
	e^{z_{2}\left( \tilde{t}_{2}+T\right) }dz_{2}}{\left( z_{2}-i\left( k\cdot
	w_{0}\right) \right) }\frac{1}{\Delta _{\sigma }\left( -k,z_{2}\right) } \ .
\end{align}

We will assume that $\Delta _{\sigma }\left( k,z_{1}\right) $ is analytic in
the regions indicated above. We can then compute the integrals along the
circuits $\gamma $ using residues. The contribution to the integral in the
region $\func{Re}\left( z_{j}\right) <0$ converges exponentially to zero as $%
T\rightarrow \infty .$ The exponent depends on $\left\vert k\right\vert $
but due to the integrability in $k$ we can take the limit of those terms.
Therefore, we are left only with the contributions due to residues at $%
z_{1}=-i\left( k\cdot w_{0}\right) ,\ z_{2}=i\left( k\cdot w_{0}\right) .$
We then obtain, using the identity $\Delta _{\sigma }(-k,z^{\ast })=(\Delta _{\sigma
}(k,z))^{\ast }$: 
\begin{equation}
\mathbb{E}\left[ F_{g}\left( \xi _{1},\tilde{t}_{1}\right) \otimes
F_{g}\left( \xi _{2},\tilde{t}_{2}\right) \right] =\int_{\mathbb{R}%
	^{3}}g\left( w_{0}\right) dw_{0}\int_{\R^3}\left[ k\otimes k\right]
\left\vert \hat{\Phi}\left( k\right) \right\vert ^{2}\frac{e^{ik\cdot \left[
		w_{0}\left( \tilde{t}_{2}-\tilde{t}_{1}\right) -\left( \xi _{2}-\xi
		_{1}\right) \right] }}{\left\vert \Delta _{\sigma }\left( k,-i\left( k\cdot
	w_{0}\right) \right) \right\vert ^{2}}dk.  \label{NoiseBLFastDec}
\end{equation}

We can now argue as in the Rayleigh gas case (cf.~\eqref{eq:Rayclaim1}-\eqref{eq:Rayclaim2}) and hence \eqref{eqthm:nonlin1} follows. This concludes the proof of Theorem~\ref{thm:FiniteRange} for the interacting particle case.
\end{proofof}

We now state the Lemma which provides the representation formulas \eqref{eq:GSigma}, \eqref{eq:Xirepresentation}, \eqref{T3E5} we used in the above proof. 

\begin{lemma}\label{Lemma:FouLap}
	The fundamental solution $G_\sigma$ and its marginal $\Xi _{\sigma }(y,w_{0},t)$ can be explicitly represented in Fourier-Laplace variables as
	\begin{align}
	\tilde{G}_{\sigma }\left( k,w,w_{0},z\right) &=\frac{\sigma \left( \nabla
		_{w}g\left( w\right) \cdot ik\right) \hat{\Phi}\left( k\right) }{\left( z+i
		k\cdot w \right) \left( z+i k\cdot w_{0} \right) \Delta _{\sigma }\left(
		k,z\right) } +\frac{\delta \left( w-w_{0}\right) }{\left( 2\pi \right) ^{%
			\frac{3}{2}}\left( z+i\left( k\cdot w_{0}\right) \right) } \label{eq:GSigma_1}\\
		\tilde{\Xi}_{\sigma}\left( k,w_{0},z\right) &=\frac{1}{\left( 2\pi\right) ^{%
				\frac{3}{2}}\left( z+i\left( k\cdot w_{0}\right) \right) \Delta_{\sigma
			}\left( k,z\right) }  \label{eq:Xirepresentation_1}\\
		\Delta_{\sigma}\left( k,z\right) &=1- \sigma(2\pi)^\frac32 \hat{\Phi}\left(
		k\right) \int_{\R^3}\frac{\left( ik\cdot\nabla_{w}g\left( w\right)
			\right) }{z+i\left( k\cdot w\right) }dw .  \label{T3E5_1}
	\end{align}
\end{lemma}
\begin{proof}
Taking the transform of (\ref{T2E2})-(\ref{T2E4}) we obtain: 
\begin{align}
(z +i\left( k\cdot w\right) ) \tilde{G}_{\sigma }\left( k,w,w_{0},z\right)
-&\sigma(2\pi)^\frac32 \left( \nabla _{w}g\left( w\right) \cdot ik\right) 
\hat{\Phi}\left( k\right) \tilde{\Xi}_{\sigma }\left( k,w_{0},z\right) =%
\frac{\delta \left( w-w_{0}\right) }{\left( 2\pi \right) ^{\frac{3}{2}}} \ .
\label{T3E1}
\end{align}
We can rewrite (\ref{T3E1}) as:
\begin{equation}
\tilde{G}_{\sigma}\left( k,w,w_{0},z\right) =\frac{\sigma (2\pi)^\frac32
	\left( \nabla _{w}g\left( w\right) \cdot ik\right) \hat{\Phi}\left( k\right) 
}{ (z+i\left( k\cdot w\right)) }\tilde{\Xi}_{\sigma}\left( k,w_{0},z\right) +%
\frac{\delta\left( w-w_{0}\right) }{( 2\pi ) ^{\frac{3}{2}} (z+i\left(
	k\cdot w\right)) } \ .  \label{T3E2a}
\end{equation}
Integrating in $w$ we obtain:%
\begin{equation}
\tilde{\Xi}_{\sigma}\left( k,w_{0},z\right) =\left(\sigma (2\pi)^\frac32 
\tilde{\Xi}_{\sigma}\left( k,w_{0},z\right) \hat{\Phi}\left( k\right) \int_{%
	\R^3}\frac{ ik\cdot\nabla_{w}g\left( w\right) }{z+i k\cdot w }dw+%
\frac{1}{(2\pi)^\frac32 (z+i k w_0) }\right).  \label{T3E3}
\end{equation}
Then we can represent $\tilde{\Xi}_{\sigma}$ explicitly by \eqref{eq:Xirepresentation_1}. Using  (\ref{T3E2a}) we also obtain the representation \eqref{eq:GSigma_1}.
\end{proof}

\bigskip

\subsection{Coulombian interactions \label{RaylCoulomb}} 

In this section we prove Theorem \ref{thm:Coulomb}. To avoid repetition of similar arguments we will only prove the Theorem for the interacting particle case which is technically more involved. Concerning the Rayleigh Gas case the proof can be done following similar arguments to the ones used in Subsection~\ref{RaylCompSuppPot}. We just notice that contrary to the case studied in Subsection~\ref{RaylCompSuppPot}, the
	friction term $H_{g}$ and the noise $B_{g}$ do not stabilize on a short
	Bogoliubov time scale, but they exhibit a logarithmic divergence as ${\eps }\rightarrow 0$. This logarithmic
	divergence yields the logarithmic correction of the time scale
	characteristic of Coulombian interactions.\medskip

\begin{proofof}[Proof of Theorem \ref{thm:Coulomb}]

\noindent \textit{Part II: Interacting Particle systems} \smallskip

\noindent We consider \eqref{A1E5} with initial conditions (\ref{T1E1}), (\ref{T1E2}). We can formulate a
problem for the difference $\lambda _{\eps }\left( y,w,\tilde{t}%
\right) =\zeta _{\eps }^{+}\left( y,w,\tilde{t}\right) -\zeta
_{\eps }^{-}\left( y,w,\tilde{t}\right) .$ Due to the
electroneutrality condition we have: 
\begin{equation}
(\tilde{\rho}_{\eps }^{+}-\tilde{\rho}_{\eps }^{-})\left( \eta
,\tilde{t}\right) =\frac{1}{L_\eps^{\frac{3}{2}}}%
\int_{\R^3}\left[ (\zeta _{\eps }^{+}-\zeta _{\eps
}^{-})\left( \eta ,w,\tilde{t}\right) \right] dw=\frac{1}{L_\eps^\frac32 }\int_{\R^3}\lambda
_{\eps }\left( \eta ,v,\tilde{t}\right) dv.  \label{A1E7}
\end{equation}%
Then from (\ref{A1E5}) we obtain
\begin{align*}
& 0=\partial _{\tilde{t}}\lambda_\eps \left( y,w,\tilde{t}\right)
+w\cdot \nabla _{y}\lambda _{\eps }\left( y,w,\tilde{t}\right) - \\
-2& \left[  L_\eps^2 \int_{\R^3}d\eta
\int_{\R^3}\nabla \Phi \left( L_{\eps }\left\vert
y-\eta \right\vert \right) \lambda _{\eps }\left( \eta ,v,\tilde{t}%
\right) dv+\eps  L_\eps ^{\frac52}\nabla\Phi \left( L_{\eps }\left\vert y-\xi \right\vert \right) \right]
\cdot \nabla _{w}g\left( w\right) 
\end{align*}%
where $\xi=V_0\tilde{t}$. From (\ref{T1E1}), (\ref{T1E2}) we get
\begin{eqnarray}
\mathbb{E}\left[ \lambda _{\eps }\left( y,w,0\right) \right] &=&0 
\notag \\
\mathbb{E}\left[ \lambda _{\eps }\left( y_{a},w_{a},0\right) \lambda
_{\eps }\left( y_{b},w_{b},0\right) \right] &=&2g\left( w_{a}\right)
\delta \left( y_{a}-y_{b}\right) \delta \left( w_{a}-w_{b}\right) .
\label{NoiseInData}
\end{eqnarray}%
Using the linearity of the problem we can decompose $\lambda _{\eps }$
as:
\begin{equation}
\lambda _{\eps }=\lambda _{1}+\lambda _{2}  \label{A2E6}
\end{equation}%
where we do not write explicitly the dependence of $\lambda _{1},\lambda
_{2} $ in $\eps $ and $\lambda _{1},\lambda _{2}$ solve the following
problems:%
\begin{equation}
\partial _{\tilde{t}}\lambda _{1}\left( y,w,\tilde{t}\right) +w\cdot \nabla
_{y}\lambda _{1}\left( y,w,\tilde{t}\right) -2 L_\eps^2 \nabla g\left( w\right) \cdot \int_{\R^3}\nabla \Phi
\left( L_{\eps }\left\vert y-\eta \right\vert \right) d\eta \int_{%
	\R^3}\lambda _{1}\left( \eta ,v,\tilde{t}\right) dv=0  , \label{A1E8}
\end{equation}%
\begin{equation}
\mathbb{E}\left[ \lambda _{1}\left( y,w,0\right) \right] =0\ \ ,\ \ \mathbb{E%
}\left[ \lambda _{1}\left( y_{a},w_{a},0\right) \lambda _{1}\left(
y_{b},w_{b},0\right) \right] =2g\left( w_{a}\right) \delta \left(
y_{a}-y_{b}\right) \delta \left( w_{a}-w_{b}\right)  \label{A1E9}
\end{equation}%
and:%
\begin{align}
& \partial _{\tilde{t}}\lambda _{2}\left( y,w,\tilde{t}\right) +w\cdot
\nabla _{y}\lambda _{2}\left( y,w,\tilde{t}\right)  \notag \\
& =2\left[ \left( L_{\eps }\right) ^{2}\int_{\R^3}\nabla\Phi \left( L_{\eps }\left\vert y-\eta \right\vert \right) d\eta
\int_{\R^3}\lambda _{2}\left( \eta ,v,\tilde{t}\right)
dv+\eps \left( L_{\eps }\right) ^{\frac{5}{2}}\nabla \Phi
\left( L_{\eps }\left\vert y-V_0\tilde{t} \right\vert \right) \right] \cdot
\nabla _{w}g\left( w\right)  \label{A2E1}, \\
& \lambda _{2}\left( y,w,0\right) =0  \ .\label{A2E2}
\end{align}%
In order to obtain $\lambda _{1}$ and $\lambda _{2}$ we need to study the
fundamental solution for the linearized Vlasov system given by the solution
of:%
\begin{align}
& \partial _{\tilde{t}}G\left( y,w;w_{0},\tilde{t}\right) +w\cdot \nabla
_{y}G\left( y,w;w_{0},\tilde{t}\right)   \notag \\
& -2\left( L_{\eps }\right) ^{2}\nabla _{w}g\left( w\right) \cdot
\int_{\R^3}\nabla \Phi \left( L_{\eps }\left\vert
y-\eta \right\vert \right) d\eta \int_{\R^3}G\left( \eta ,v;w_{0},%
\tilde{t}\right) dv=0 , \label{A2E3} \\
& G\left( y,w;w_{0},0\right) =\delta \left( y\right) \delta \left(
w-w_{0}\right) \ . \label{A2E4}
\end{align}

We can solve (\ref{A2E3}), (\ref{A2E4}) applying Fourier in the variable $y$
and Laplace in $\tilde{t}.$ If we denote this transform as $\tilde{G}=\tilde{%
	G}\left( k,w;w_{0},z\right) .$ Then: 
\begin{align*}
(z+i\left( k\cdot w\right))\tilde{G}& \left( k,w;w_{0},z\right) -\frac{%
	2\left( 2\pi \right) ^{\frac{3}{2}}}{\left( L_{\eps }\right) ^{2}}%
\hat{\Phi}\left( \frac{k}{L_{\eps }}\right) \left( ik\cdot \nabla
_{w}g\left( w\right) \right) \tilde{\Xi} \left( k;w_{0},z\right) =\frac{%
	\delta \left( w-w_{0}\right) }{\left( 2\pi \right) ^{\frac{3}{2}}} ,
\end{align*}
where $\tilde{\Xi}\left( k;w_{0},z\right) =\rho[\tilde{G}(\cdot;w_{0},z)](k)$. Then we have 
\begin{equation}
\tilde{G}\left( k,w;w_{0},z\right) =\frac{\delta \left( w-w_{0}\right) }{%
	\left( 2\pi \right) ^{\frac{3}{2}}\left( z+i\left( k\cdot w_{0}\right)
	\right) }+\frac{2\left( 2\pi \right) ^{\frac{3}{2}}}{\left( L_{\eps
	}\right) ^{2}}\hat{\Phi}\left( \frac{k}{L_{\eps }}\right) \left(
ik\cdot \nabla _{w}g\left( w\right) \right) \frac{\tilde{\Xi} \left(
	k;w_{0},z\right) }{z+i\left( k\cdot w\right) } \ . \label{eq:GLaplace}
\end{equation}
Therefore, after integrating in $w$ we find an explicit representation of $%
\tilde{\Xi}$:

\begin{equation}
\tilde{\Xi}\left( k;w_{0},z\right) =\frac{1}{\left( 2\pi \right) ^{\frac{3}{2%
	}}\left( z+i\left( k\cdot w_{0}\right) \right) \Delta _{\eps }\left(
	k,z\right) }  \label{eq:Xi}
\end{equation}%
where:%
\begin{equation}
\Delta _{\eps }\left( k,z\right) =1-\frac{2\left( 2\pi \right) ^{%
		\frac{3}{2}}}{\left( L_{\eps }\right) ^{2}}\hat{\Phi}\left( \frac{k}{%
	L_{\eps }}\right) \int_{\R^3}\frac{\left( ik\cdot \nabla
	_{w}g\left( w\right) \right) }{z+i\left( k\cdot w\right) }dw \ .
\label{def:dielectric}
\end{equation}
Inserting \eqref{eq:Xi} back into \eqref{eq:GLaplace} we obtain the
representation: 
\begin{equation*}
\tilde{G}\left( k,w;w_{0},z\right) =\frac{\delta \left( w-w_{0}\right) }{%
	\left( 2\pi \right) ^{\frac{3}{2}}\left( z+i\left( k\cdot w_{0}\right)
	\right) }+\frac{2}{\left( L_{\eps }\right) ^{2}}\frac{\hat{\Phi}%
	\left( \frac{k}{L_{\eps }}\right) \left( ik\cdot \nabla _{w}g\left(
	w\right) \right) }{\left( z+i\left( k\cdot w\right) \right) \left( z+i\left(
	k\cdot w_{0}\right) \right) \Delta _{\eps }\left( k,z\right) }\ .
\end{equation*}

We will assume the usual stability condition of the medium, i.e. that the
function $\Delta _{\eps }\left( k,\cdot \right) $ is analytic in $%
\left\{ \func{Re}\left( z\right) >0\right\} $ for each $k\in \R^3.$

We can now obtain $\lambda _{\eps }$ using (\ref{A2E6}) as well as
the fundamental solution. We then obtain 
\begin{align*}
& \lambda _{\eps }\left( y,w,\tilde{t}\right) =\int_{\mathbb{R}%
	^{3}}d\eta \int_{\R^3}dw_{0}G\left( y-\eta ,w,w_{0},\tilde{t}%
\right) \lambda _{1}\left( \eta ,w_{0},0\right) + \\
& +2\eps \left( L_{\eps }\right) ^{\frac{5}{2}}\int_{0}^{%
	\tilde{t}}ds\int_{\R^3}d\eta \int_{\R^3}dw_{0}G\left(
y-\eta ,w,w_{0},\tilde{t}-s\right) \nabla \Phi \left( L_{\eps
}\left\vert \eta -V_0s\right\vert \right) \cdot \nabla _{w}g\left( w_{0}\right)
\end{align*}
and, using the first equation in \eqref{eq:appCoul},
\begin{equation}
\frac{1}{L_\eps^\frac32 } {F}_{g}\left( \xi ,\tilde{t}\right) =\frac{1}{L_\eps^\frac32 }F_g^{(1)}( \xi ,\tilde{t})+\frac{1}{L_\eps^3}H_g(\tilde{t};V_0). \label{A3E2}
\end{equation}%
Here $F_g^{(1)}$, $H_g$ are given by:
\begin{equation}
\begin{aligned} F_g^{(1)}( \xi ,\tilde{t})
&=-\int_{(\mathbb{R}^3)^3}L_\eps^2\nabla \Phi
\left( L_\eps \left\vert y-\xi \right\vert \right) \Xi\left( y-\eta
,w_{0},\tilde{t}\right) \lambda _{1}\left( \eta ,w_{0},0\right)d\eta dw_{0} dy \\ H_g( \tilde{t};V_0) &=-2
\int_{(\mathbb{R}^3)^3} L_\eps^2 \nabla \Phi
\left( L_{\eps }\left\vert y-V_0\tilde{t}\right\vert \right)
\int_{0}^{\tilde{t}}ds\ \Xi\left( y-\eta ,w_{0},\tilde{t}-s\right) d\eta dw_{0} dy \\ &\qquad\qquad L_\eps^2\nabla \Phi
\left( L_{\eps }\left\vert \eta -V_0s\right\vert \right) \cdot \nabla
_{w}g\left( w_{0}\right) . \end{aligned}  \label{eq:FHrepr}
\end{equation}
Notice that $F_g^{(1)}( \xi ,\tilde{t}) $ is a Gaussian
random force field. Using (\ref{A1E9}) we obtain:%
\begin{equation}
\mathbb{E}\left[ F_g^{(1)}\left( \xi ,\tilde{t}\right) \right] =0 \ .
\label{A2E9}
\end{equation}%
Similarly, the covariance has the form: 
\begin{align}
& \mathbb{E}\left[ F_g^{(1)}\left( \xi _{1},\tilde{t}_{1}\right) \otimes 
F_g^{(1)}\left( \xi _{2},\tilde{t}_{2}\right) \right] =2\int_{\R^3}dy_{1}\int_{\R^3}dy_{2}L_\eps^2 \nabla \Phi \left( L_{\eps }\left\vert y_{1}-\xi
_{1}\right\vert \right) \otimes \left( L_{\eps }\right) ^{2}\nabla\Phi \left( L_{\eps }\left\vert y_{2}-\xi _{2}\right\vert \right)
\cdot  \notag  \label{A3E1} \\
& \cdot \int_{\R^3}g\left( w_{0,1}\right) dw_{0,1}\int_{\mathbb{R}%
	^{3}}d\eta \int_{\R^3}dw_{1}\int_{\R^3}dw_{2}G\left(
y_{1}-\eta ,w_{1},w_{0,1},\tilde{t}_{1}\right) G\left( y_{2}-\eta
,w_{2},w_{0,1},\tilde{t}_{2}\right) .  \notag
\end{align}

We now compute the asymptotics of the friction term $H_{g}\left( 
\tilde{t};V\right) $ given in \eqref{eq:FHrepr} as $\tilde{t}\rightarrow
\infty $, i.e. for times for which the tagged particle moves at distances 
much larger than the Debye screening length. Using \eqref{eq:Xi} we obtain: 
\begin{align*}
	H_{g}\left( \tilde{t};V_0\right) & =-\frac{1}{\left( 2\pi \right)
	^{3}\pi i}\int_{(\R^3)^{3}}dyd\eta dw_{0}\left( L_{\eps
}\right) ^{2}\nabla \Phi \left( L_{\eps }\left( y-V\tilde{t}%
\right) \right) \int_{0}^{\tilde{t}}ds\int_{{\mathbb{R}}^{3}}dk\ e^{ik\cdot
	\left( y-\eta \right) } \\
& \cdot \int_{\gamma }dz\frac{e^{z\left( \tilde{t}-s\right) }}{\left(
	z+i\left( k\cdot w_{0}\right) \right) }\frac{1}{\Delta _{\eps }\left(
	k,z\right) }\left( L_{\eps }\right) ^{2}\nabla \Phi \left(
L_{\eps }\left( \eta -V_0s\right) \right) \cdot \nabla _{w}g\left(
w_{0}\right) \ .
\end{align*}

Then, rewriting the equation above in terms of the Fourier transform of $%
\Phi $, the friction term becomes 
\begin{equation*}
H_g\left( \tilde{t};V_0\right) =\frac{1}{\pi i}\int_{0}^{\tilde{t}%
}ds\int_{\R^3}dk\ k\frac{e^{ik\cdot V_0s}}{\left( L_{\eps
	}\right) ^{4}}\left\vert \hat{\Phi}\left( \frac{k}{L_{\eps }}\right)
\right\vert ^{2}\int_{\gamma }dz\frac{e^{zs}}{\Delta _{\eps }\left(
	k,z\right) }\int_{\R^3}dw_{0}\frac{\left[ k\cdot \nabla
	_{w}g\left( w_{0}\right) \right] }{\left( z+i\left( k\cdot w_{0}\right)
	\right) } \ .
\end{equation*}
We define, analogously as in the case considered in the previous
Subsections, 
\begin{equation}
\Psi \left( z,k\right) =\Psi \left( \frac{z}{\left\vert k\right\vert },\frac{%
	k}{\left\vert k\right\vert }\right) =\int_{\R^3}dw_{0}\frac{\left[
	k\cdot \nabla _{w}g\left( w_{0}\right) \right] }{\left( z+i\left( k\cdot
	w_{0}\right) \right) }  \label{T8E6}
\end{equation}%
whence, using \eqref{def:dielectric}, 
\begin{equation}
\Delta _{\eps }\left( k,z\right) =1-\frac{2\left( 2\pi \right) ^{%
		\frac{3}{2}}}{\left( L_{\eps }\right) ^{2}}\hat{\Phi}\left( \frac{k}{%
	L_{\eps }}\right) \Psi \left( \frac{z}{\left\vert k\right\vert },%
\frac{k}{\left\vert k\right\vert }\right) .  \label{T8E6a}
\end{equation}
Thus 
\begin{equation}
\tilde{H}_{g}\left( \tilde{t};V_0\right) =\frac{1}{\pi i} \int_{\R^3}%
\frac{dk \, k}{\left( L_{\eps }\right) ^{4}}\left\vert \hat{\Phi}\left( 
\frac{k}{L_{\eps }}\right) \right\vert ^{2}\int_{0}^{\tilde{t}}ds%
\left[ \int_{\gamma }\frac{e^{\left( ik\cdot V_0+z\right) s}dz}{\Delta
	_{\eps }\left( k,z\right) }\Psi \left( \frac{z}{\left\vert
	k\right\vert },\frac{k}{\left\vert k\right\vert }\right) \right].
\label{A3E8}
\end{equation}

We will assume that the function $g$ has analyticity properties analogous to
the ones assumed in Subsection~\ref{RaylCompSuppPot}. Moreover, we will
also assume that the Penrose stability condition holds, i.e. $\Delta
_{\eps }\left( k,z\right) \neq 0$ for $\func{Re}\left( z\right) \geq
0 $ and $k\in \R^3\diagdown \left\{ 0\right\} .$

Using the changes of variables $k=L_{\eps }p,\ z=L_{\eps
}\zeta ,\ s=\frac{\tau }{L_{\eps }}$, which allow to return to the
microscopic variables in Fourier, we can rewrite (\ref{A3E8}) as: 
\begin{equation}
\tilde{H}_{g}\left( \tilde{t};V_0\right) =\frac{1}{\pi i}\int_{\mathbb{R}%
	^{3}}dp\, p\left\vert \hat{\Phi}\left( p\right) \right\vert
^{2}\int_{0}^{L_{\eps }\tilde{t}}d\tau \left[ \int_{\gamma }\frac{%
	e^{\left( ip\cdot V_0+\zeta \right) \tau }d\zeta }{\Delta _{\eps
	}\left( L_{\eps }p,L_{\eps }\zeta \right) }\Psi \left( \frac{%
	\zeta }{\left\vert p\right\vert },\frac{p}{\left\vert p\right\vert }\right) %
\right] .  \label{T8E7}
\end{equation}

Using also (\ref{T8E6a}) we obtain: 
\begin{equation*}
\tilde{H}_{g}\left( \tilde{t};V_0\right) =\frac{1}{\pi i}\int_{0}^{L_{%
		\eps }\tilde{t}}d\tau \int_{\R^3} dp\, p e^{ip\cdot V_0\tau
}\left\vert \hat{\Phi}\left( p\right) \right\vert ^{2}\int_{\gamma }\frac{%
	\Psi \left( \frac{\zeta }{\left\vert p\right\vert },\frac{p}{\left\vert
		p\right\vert }\right) e^{\zeta \tau }d\zeta }{\left( 1-\frac{2\left( 2\pi
		\right) ^{\frac{3}{2}}}{\left( L_{\eps }\right) ^{2}}\hat{\Phi}\left(
	p\right) \Psi \left( \frac{\zeta }{\left\vert p\right\vert },\frac{p}{%
		\left\vert p\right\vert }\right) \right) }\ .
\end{equation*}
Then, due to the Penrose stability condition the function $\Psi$ is analytic and we can perform a contour
deformation for the integration in $\zeta $ to bring the contour to the
region $\left\{ \func{Re}\left( \zeta \right) <0\right\} $.  More precisely, if $|p|\geq L_\eps^{-1}$, we can move the contour to $\func{Re}(\zeta)<0$, $|\func{Re}(\zeta)|\sim 1$. Physically this is related to the fact that the Landau damping of disturbances with a wavelength smaller or equal than the Debye length scale $L_\eps$, takes place in microscopic scales of order $\tau \sim 1$. 

Hence, in the region $\{|p|\geq L_\eps^{-1}\}$ we can replace the integral $\int_0^{L_\eps \tilde{t}} $ by $\int_0^\infty$ with a negligible error. On the other hand, estimating the contribution due to the region $|p|\leq L_\eps^{-1}$ is more involved. The reason for this is the fact that the function $1-\frac{2\left( 2\pi \right) ^{\frac{3}{2}}}{\left( L_{\eps
	}\right) ^{2}}\hat{\Phi}\left( p\right) \Psi \left( \frac{\zeta }{\left\vert
	p\right\vert },\frac{p}{\left\vert p\right\vert }\right)$ has two roots with $\func{Re}(\zeta)<0$ but $|\func{Re}(\zeta)|$ exponentially small as $\eps\rightarrow 0$. 
Physically, these roots are related to the so-called Langmuir waves that are oscillatory solutions of the Vlasov-Poisson equation with wavelength much larger than the Debye length and very slow Landau damping (for a discussion see \cite{NVW,VW2}). The contribution to the integrals $\int_0^{L_\eps\tilde{t}} \int_\gamma d\zeta \int_{\R^3} dp [\ldots ]$ can be estimated computing first the integral $\int d\zeta$ using residuals. It is then possible to prove that this contribution is bounded as $\eps\rightarrow 0$, and we obtain
\begin{align*}
\tilde{H}_{g}\left( \tilde{t};V_0\right) &= H_g\left( \infty ;V_0\right)+ O(1), \quad \text{where } \\
H_g\left( \infty ;V_0\right) &=\frac{1}{\pi i}\int_{\R^3} dp\,  p \left\vert \hat{\Phi}\left( p\right) \right\vert ^{2}\cf_{\{|p|\geq L_\eps^{-1}\}}\int_{\gamma }%
\frac{1}{\left( \zeta +ip\cdot V_0\right) }\frac{\Psi \left( \frac{\zeta }{%
		\left\vert p\right\vert },\frac{p}{\left\vert p\right\vert }\right) d\zeta }{%
	\left( 1-\frac{2\left( 2\pi \right) ^{\frac{3}{2}}}{\left( L_{\eps
		}\right) ^{2}}\hat{\Phi}\left( p\right) \Psi \left( \frac{\zeta }{\left\vert
		p\right\vert },\frac{p}{\left\vert p\right\vert }\right) \right) }\ .
\end{align*}
We will now prove that $H_g\left( \infty ;V_0\right)\sim |\log \eps|$, so it is the dominant term.  
Using now residues to compute the integral along the contour $\gamma $ we
arrive at 
\begin{equation}
\tilde{H}_{g}\left( \infty ;V_0\right) =2\int_{\R^3} dp\, p \cf_{\{|p|\geq L_\eps^{-1}\}}\left\vert 
\hat{\Phi}\left( p\right) \right\vert ^{2}\frac{\Psi \left( -\frac{ip\cdot V_0}{\left\vert p\right\vert },\frac{p}{\left\vert p\right\vert }\right) }{%
	\left( 1-\frac{2\left( 2\pi \right) ^{\frac{3}{2}}}{\left( L_{\eps
		}\right) ^{2}}\hat{\Phi}\left( p\right) \Psi \left( -\frac{ip\cdot V_0}{%
		\left\vert p\right\vert },\frac{p}{\left\vert p\right\vert }\right) \right) } \ 
\label{T8E8}
\end{equation}
and hence, using \eqref{T8E6}, 
\begin{align}
	 \tilde{H}_{g}\left( \infty ;V\right) = -\int_{\R^3}\int_{\R^3}
	 (p \otimes p)  |\hat{\Phi}( p)|^2\frac{ \cf_{\{|p|\geq L_\eps^{-1}\}}\delta(p(v-w)) }{|1-\frac{2\left( 2\pi \right) ^{\frac{3}{2}}}{\left( L_{\eps
	 		}\right) ^{2}}\hat{\Phi}\left( p\right) \Psi \left( -\frac{ip\cdot V}{%
	 		\left\vert p\right\vert },\frac{p}{\left\vert p\right\vert }\right)|^2}\nabla g(w) \ud{w}\ud{p} .
\end{align}
This formula yields the asymptotic friction coefficient acting on a particle
which moves at speed $V_0$. We first notice that the integral in the right
hand side of \eqref{T8E8} is convergent. Indeed, if $|p|\rightarrow \infty $
we can assume that $|\hat{\Phi}\left( p\right) |$ decays sufficiently fast,
say exponentially, due to the cutoff for small distances we made for the
potential. On the other hand, since the potential $\Phi $ is behaving at
large distances as Coulombian potential $\hat{\Phi}\left( p\right) \sim 
\frac{c}{\left\vert p\right\vert ^{2}}$, $c>0$ as $\left\vert p\right\vert
\rightarrow 0$ then at a first glance the terms $p\left\vert \hat{\Phi}%
\left( p\right) ^{2}\right\vert $ would yield a logarithm divergence as $%
|p|\rightarrow 0$. Nevertheless, this divergence does not take place due to
the presence of the term $\frac{1}{L_{\eps }^{2}}\hat{\Phi}\left(
p\right) $ in the denominator of \eqref{T8E8}. This provides a suitable
cutoff of the singularities for $\left\vert p\right\vert $ of order $\frac{1%
}{L_{\eps }}.$

We remark that some care is needed in the deformation of the contour $\gamma 
$ appearing in the integral in (\ref{T8E7}). This is due to the fact that
the region of analyticity of the function $\left( \Delta _{\eps
}\left( k,z\right) \right) ^{-1}$ in the $z$ variable, becomes very small as 
$\left\vert k\right\vert \rightarrow 0.$ This is very closely related to the
so-called Langmuir waves, which have been discussed in \cite{NVW} and therefore, we refer there for more details about this issue.

It is relevant to remark that the relevant contributions in the integral (%
\ref{T8E8}) are those with $\left\vert p\right\vert \approx 1.$ Using
Plancherel's formula it follows that these contributions are those between
the region where we cut the potential (i.e. $\left\vert X\right\vert $ of
order one), until regions of the order of the Debye length (i.e. $\left\vert
X\right\vert \approx L_{\eps }=\frac{1}{\sqrt{\eps }}$).
Moreover, all the dyadic regions within this range of values yield
contributions of a similar size as it might be expected for Coulombian
potentials (cf. Subsection \ref{ThresholdCoulomb}). The asymptotic behaviour
of (\ref{T8E8}) as $\eps \rightarrow 0$ contains then, as it might be
expected, the Coulombian logarithm:%
\begin{equation*}
\tilde{H}_{g}\left( \infty ;V_0\right) \sim 2c^{2}\log \left( L_{\eps
}\right) \int_{\mathbb{R}%
^{3}}dw_{0} \frac{\nabla g(w_{0})}{|v-w_{0}|}\left( I-\frac{(v-w_{0})\otimes
(v-w_{0})}{|v-w_{0}|^{2}}\right) ,
\end{equation*}
where $c=\frac{A}{\sqrt{2\pi}}$, and $A$ is given by \eqref{phiDec_sOne}. 
Using \eqref{TCoulRay} we can compute the friction force on the macroscopic unit of time which reads 
\begin{equation}
H_{g}\left( \infty ;V_0\right) \sim \frac{\pi}{2} c^2   \Lambda_{g}\left( V_0\right)\ ,
\label{FrictAsy}
\end{equation}
as claimed.

In order to obtain the diffusion coefficient we first define the random variable:
\begin{equation*}
d_{g}\left( t\right) =\int_{0}^{t}F_{g}\left( V_0 s,s\right) ds\ .
\end{equation*}
which measures the particle
deflections for small $t.$ We have: 
\begin{equation}
\mathbb{E}\left[ d_{g}\left( t\right) \right] =0  \label{T8E4a}
\end{equation}%
and, denoting for simplicity $\tilde{T}_{\eps }=\frac{1}{\theta_{\eps }^{2}}$,
\begin{align*}
& \mathbb{E}\left[ d_{g}\left( t_{1}\right) \otimes d_{g}\left( t_{2}\right) %
\right] =2L_{{\eps }}\tilde{T}_{\eps }^2\int_{0}^{t_{1}}%
\int_{0}^{t_{2}}ds_{1}ds_{2}\int_{\R^3}d\eta \int_{\mathbb{R}%
	^{3}}dw_{0}g\left( w_{0}\right) \Xi \left( y_{1}-\eta ,w_{0},\tilde{T}%
_{\eps }s_{1}\right) \cdot \\
& \cdot \Xi \left( y_{2}-\eta ,w_{0},\tilde{T}_{\eps }s_{2}\right)
\int_{(\R^3)^{2}}dy_{1}dy_{2}\nabla \Phi \left(
L_{\eps }\left\vert y_{1}-\tilde{T}_{\eps }V_0s_{1}\right\vert
\right) \otimes \nabla \Phi \left( L_{\eps }\left\vert y_{2}-%
\tilde{T}_{\eps }V_0 s_{2}\right\vert \right) .
\end{align*}

We insert the Fourier-Laplace representation of $\Xi $ (cf.~\eqref{eq:Xi}) and 
obtain: 
\begin{align*}
& \mathbb{E}\left[ d_{g}\left( t_{1}\right) \otimes d_{g}\left( t_{2}\right) %
\right] =\frac{2L_{{\eps }}\tilde{T}_{{\eps }}^{2}}{(2\pi )^{3}%
}\int_{0}^{t_{1}}\int_{0}^{t_{2}}ds_{1} ds_{2}\int_{(\mathbb{R}%
	^{3})^{3}}dw_{0}dk \int_{\gamma }\;\mathrm{d}{z_{1}}\int_{\gamma }\;\mathrm{%
	d}{z_{2}}\frac{e^{z_{1}\tilde{T}_{{\eps }}s_{1}}}{z_{1}+ikw_{0}}\cdot
\\
& \cdot g(w_{0})\int_{(\R^3)^{2}}dy_{1} dy_{2}\frac{\nabla\Phi \left( L_{\eps }\left\vert y_{1}-\tilde{T}_{\eps
	}V_0 s_{1}\right\vert \right) \otimes \nabla \Phi \left( L_{\eps
	}\left\vert y_{2}-\tilde{T}_{\eps }V_0 s_{2}\right\vert \right) }{\Delta
	_{{\eps }}(k,z_{1})\Delta _{{\eps }}(-k,z_{2})}\frac{e^{z_{2}%
		\tilde{T}_{{\eps }}s_{2}}}{z_{2}-ikw_{0}}.
\end{align*}
We perform the contour integrals in $\gamma $ and we get 
\begin{align*}
& \mathbb{E}\left[ d_{g}\left( t_{1}\right) \otimes d_{g}\left( t_{2}\right) %
\right] =\frac{2L_{{\eps }}\tilde{T}_{{\eps }}^{2}}{(2\pi )^{3}%
}\int_{0}^{t_{1}}\int_{0}^{t_{2}}ds_{1} ds_{2} \int_{(\mathbb{R}%
	^{3})^{2}}dw_{0}dk \mathrm{d}{z_{2}}e^{-ik\cdot w_{0}\tilde{T}_{\eps
	}(s_{1}-s_{2})}  \cf_{\{|k|\geq L_\eps^{-1}\}}\cdot \\
& \cdot e^{ikL_{\eps }k(y_{1}-y_{2})}g(w_{0})\int_{(\mathbb{R}%
	^{3})^{2}}dy_{1} dy_{2}\frac{\nabla \Phi \left( L_{\eps
	}\left\vert y_{1}-\tilde{T}_{\eps }V_0 s_{1}\right\vert \right) \otimes
	\nabla \Phi \left( L_{\eps }\left\vert y_{2}-\tilde{T}%
	_{\eps }V_0 s_{2}\right\vert \right) }{\Delta _{\eps
	}(k,ikw_{0})\Delta _{\eps }(-k,-ikw_{0})}+o(1).
\end{align*}
Using the identity $\Delta _{\eps }(a,ib)=\Delta _{\eps
}^{\ast }(-a,-ib)$ for $b\in \mathbb{R}$ and changing to microscopic variables $x=L_\ep y$, $\tau= T_\ep t$ we obtain 
\begin{align*}
& \mathbb{E}\left[ d_{g}\left( t_{1}\right) \otimes d_{g}\left( t_{2}\right) %
\right] =\frac{2L_{{\eps }}\tilde{T}_{{\eps }}^{2}}{(2\pi)^{3}L_{\eps }^{6}T_{\eps }^{2}}\int_{0}^{T_{\eps
	}t_{1}}\int_{0}^{T_{\eps }t_{2}}d\tau _{1}d\tau _{2} \int_{(\mathbb{R%
	}^{3})^{2}}dw_{0}dk \mathrm{d}{z_{2}}e^{i/L_{\eps }k\cdot
	(V_0-w_{0})(\tau _{1}-\tau _{2})}\cdot \\
& \cdot e^{ik/L_{\eps }k(x_{1}-x_{2})}g(w_{0})\int_{(\mathbb{R}%
	^{3})^{2}}dx_{1}dx_{2} \frac{\nabla \Phi (x_{1})\otimes \nabla \Phi (x_{2})%
}{|\Delta _{\eps }(k,ikw_{0})|^{2}}  \cf_{\{|k|\geq L_\eps^{-1}\}}+o(1).
\end{align*}%
Performing the integral in $x_{1},x_{2}$ leads to the Fourier
representation: 
\begin{equation}
\begin{aligned}
&\mathbb{E}\left[ d_{g}\left( t_{1}\right) \otimes d_{g}\left( t_{2}\right) %
\right] \\
=&\frac{2}{L_{\eps }^{4}}\int_{0}^{T_{\eps
	}t_{1}}\int_{0}^{T_{\eps }t_{2}}d\tau _{1}d\tau _{2}\int_{(\mathbb{R%
	}^{3})^{2}}dw_{0} dk \frac{g(w_{0})(k\otimes k)|\hat{\phi}(k)|^{2}e^{ik\cdot
		(V_0-w_{0})(\tau _{1}-\tau _{2})}}{|\Delta _{\eps }(L_{\eps
	}k,iL_{\eps }kw_{0})|^{2}}\cf_{\{|k|\geq L_\eps^{-1}\}}+o(1).  \label{DiffDebCont}
\end{aligned}
\end{equation}
To determine the asymptotics of this expression for ${\eps }%
\rightarrow 0$ we observe that the integral in $k$ behaves like%
\begin{align}
&\int_{\R^3}dk\frac{(k\otimes k)|\hat{\phi}(k)|^{2}e^{ik\cdot
		(V_0-w_{0})(\tau _{1}-\tau _{2})}}{|\Delta _{\eps }(L_{\eps
	}k,iL_{\eps }kw_{0})|^{2}} \cf_{\{|k|\geq L_\eps^{-1}\}} \notag \\
&  \quad = \frac{c^{2}\frac{\pi^2}{2}}{|V_0-w_{0}||\tau _{1}-\tau _{2}|}%
\left( I-\frac{(V_0-w_{0})\otimes (V_0-w_{0})}{|V_0-w_{0}|^{2}} \right) \left( {\mathbbm1}%
_{|V_0-w_{0}||\tau _{1}-\tau _{2}|\leq L_{\eps }|\log^\frac12  (\eps)|}+o(1)\right),  \label{eq:kAsym}
\end{align}
Here the factor $|\log^\frac12 (\eps)|$ could be replaced by any slowly diverging factor.
Here we use that the dielectric function \eqref{def:dielectric} behaves as: 
\begin{equation*}
|\Delta _{\eps }(L_{\eps }k,iL_{\eps
}kw_{0})|^{-1}\approx 
\begin{cases}
1,\quad & \text{for $k\gtrsim 1/L_{\eps }$}\ , \\ 
\frac{k}{L_{\eps }},\quad & \text{for $k\lesssim 1/L_{\eps }$}\ .%
\end{cases}%
\end{equation*}
%
A simple way to see this is
to use the tensor structure of the integral and to compute this tensor in
the case in which $(V_0-w_0)$ is parallel to the $x_1$ coordinate axis.

Using the asymptotic formula \eqref{eq:kAsym} we finally obtain 
\begin{align*}
&\mathbb{E}\left[ d_{g}\left( t_{1}\right) \otimes d_{g}\left( t_{2}\right) %
\right] \\
& =\frac{ \pi^2 c^{2}}{L_{\eps }^{4}}\int_{0}^{T_{\eps
	}t_{1}}\int_{0}^{T_{\eps }t_{2}}d\tau _{1} d\tau _{2} \int_{\mathbb{R}%
	^{3}}dw_{0}\frac{{\mathbbm1}%
	_{|V_0-w_{0}||\tau _{1}-\tau _{2}|\leq L_{\eps } |\log^\frac12 (\eps )|}g(w_{0})}{|V_0-w_{0}||\tau _{1}-\tau _{2}|}\left( I-\frac{%
	(V_0-w_{0})\otimes (V_0-w_{0})}{|V_0-w_{0}|^{2}}\right)  \\
& = 2 \pi^2 c^{2}\eps ^{2}( \log (L_{\eps })+o(1))T_{\eps
} t_2 \wedge t_1\int_{\R^3} dw_{0} \frac{g(w_{0})}{%
	|V_0-w_{0}|}\left( I-\frac{(V_0-w_{0})\otimes (V_0-w_{0})}{|V_0-w_{0}|^{2}}\right) \\
& = \pi^2 c^{2}\eps ^{2}T_{\eps }(\log(\eps^{-1})+o(1)) t_2 \wedge t_1 \int_{\R^3}dw_{0}%
\frac{g(w_{0})}{|V_0-w_{0}|}\left( I-\frac{(V_0-w_{0})\otimes (V_0-w_{0})}{%
	|V_0-w_{0}|^{2}}\right)
\end{align*}%
where we use $L_\eps=\eps^{-\frac12 }$ and $t_1 \wedge t_2 = \min\{t_1,t_2\}$. Notice that $%
\int_{\R^3}dw_{0}\frac{g(w_{0})}{|V_0-w_{0}|}<\infty .$We then arrive at 
\begin{align}
\mathbb{E}\left[ d_{g}\left( t_{1}\right) \otimes d_{g}\left( t_{2}\right) %
\right] &\sim\pi^2 c^{2}t_2 \wedge t_1\int_{\mathbb{R}%
	^{3}}dw_{0} \frac{g(w_{0})}{|V_0-w_{0}|}\left( I-\frac{(V_0-w_{0})\otimes
	(V_0-w_{0})}{|V_0-w_{0}|^{2}}\right)  \notag \\
& = t_2 \wedge t_1 D_{g}\left( V_0\right).  \label{DiifBL}
\end{align}
This concludes the proof of Theorem \ref{thm:Coulomb}.
\end{proofof}
\medskip

We collect here some remarks that help to understand some steps of the above proof  or provide some physical intuition.
\begin{remark} 
It is possible to interpret the asymptotics of the inverse Fourier
transforms in \eqref{eq:kAsym} in terms of the correlations of a random
force field evaluated at two different points $V_0 \tau _{1}$ and $V_0\tau _{2}.$
The correlations between two points $x_{1},\ x_{2}$ decrease then as $\frac{1%
}{\left\vert x_{1}-x_{2}\right\vert }$ for distances smaller than the Debye
length. For distances larger than the Debye length $L_{\eps }$ the
decay of the correlations is much faster and therefore, the Coulombian
logarithm is due only to the range of distances between the particle size $%
\ep$ and the Debye length $L_{\ep }.$ This is different in
the case of Rayleigh gases with Coulombian interactions, where the range of
distances contributing to the kinetic regime goes from the particle size $%
\eps$ to the mean free path.
\end{remark}

\begin{remark}
	Notice that, differently from the case of Rayleigh gases with Coulombian
	interactions, after rescaling out the logarithmic term both the friction
	coefficient and the diffusion $D_{g}\left( V_0\right) $ yield a well defined
	limit as $\tilde{t}\rightarrow \infty ,$ where $\tilde{t}$ is the mesoscopic
	limit. In the case of Rayleigh gases both quantities diverge logarithmically
	as $\tilde{t}\rightarrow \infty .$
\end{remark}

\begin{remark}
	Notice that a heuristic explanation of (\ref{DiifBL}) is that the Vlasov
	evolution of the white noise which describes the fluctuations of the
	particle density yields a ``coloured noise" which decorrelates on distances
	of the order of the Debye length.
\end{remark}

\subsection{Grazing collisions\label{GrazingRayl} }

\begin{proofof}[Proof of Theorem~\ref{thm:Grazing}] 
The strategy to prove this result has many analogies with the arguments used in the proof of Theorem~\ref{thm:FiniteRange}. Due to this reason, for simplicity, we will provide the main steps of the proof for Rayleigh Gas systems. 

	We can solve the equations \eqref{T4E4}-\eqref{T4E6} as in Subsection \ref{RaylCompSuppPot}. We then
	obtain:
	\begin{equation*}
	\zeta _{\varepsilon }\left( y,w,\tilde{t}\right) =N_{\varepsilon }\left( y-w%
	\tilde{t},w\right) +\left( L_{\varepsilon }\right) ^{\frac{3}{2}}\nabla
	_{w}g\left( w\right) \cdot \int_{0}^{\tilde{t}}\nabla _{y}\tilde{\phi}%
	_{\varepsilon }\left( y-w\left( \tilde{t}-s\right) -V_0s\right) ds .
	\end{equation*}%
	 Then, 
	\begin{align*}
	\frac{dV}{d\tilde{t}}& =\left( L_{\varepsilon }\right) ^{3}\int_{\mathbb{R}%
		^{3}}\int_{\mathbb{R}^{3}}\nabla _{\eta }\tilde{\phi}_{\varepsilon }\left(
	V_0 \tilde{t}-\eta \right) N_{\varepsilon }\left( \eta -w\tilde{t},w\right) dwd\eta + \\
	& +\left( L_{\varepsilon }\right) ^{3}\int_{\mathbb{R}^{3}}\int_{\mathbb{R}%
		^{3}}\nabla _{\eta }\tilde{\phi}_{\varepsilon }\left( \eta \right) \nabla
	_{w}g\left( w\right) \cdot \int_{-\infty }^{0}\nabla _{\eta }\tilde{\phi}%
	_{\varepsilon }\left( \eta +\left( V_0-w\right) s\right) dsdwd\eta  \ .
	\end{align*}%

	We can now rewrite the equation above using a macroscopic time and spacial scale, keeping the velocity unscaled $v=V_0$. More precisely, we set 
	\begin{equation}
	t=\big(\ep\ell_{\ep}\big)^2 L_{\ep}\tilde{t} 
	, \qquad x=\big(\ep\ell_{\ep}\big)^2 L_{\ep}\xi .
	\end{equation} 
		Then, we obtain the following system of equations
	\begin{equation}
	\frac{dx}{dt}=v\ \ ,\ \ \frac{dv}{dt}=\frac{1}{\theta _{\varepsilon }\sqrt{%
			\ell _{\varepsilon }}}B_{g}\left( \frac{t}{\left( \theta _{\varepsilon
		}\right) ^{2}\ell _{\varepsilon }}\right) -\Lambda _{g}\left( v\right)  \label{A1E4}
	\end{equation}%
	where $\theta _{\varepsilon }=\varepsilon \ell _{\varepsilon }$, the friction term $\Lambda _{g}$ is given by 
	\begin{equation}
	\Lambda _{g}\left( V_0\right) =-\int_{\mathbb{R}^{3}}dY\int_{\mathbb{R}%
		^{3}}dw\nabla _{Y}\Phi \left( \left\vert Y\right\vert \right) \nabla
	_{w}g\left( w\right) \cdot \int_{-\infty }^{0}\nabla _{Y}\Phi \left(
	\left\vert Y+\left( V_0-w\right) s\right\vert \right) ds  \label{T4E8a}
	\end{equation}%
	and $B_{g}\left( \cdot \right)$ is a Gaussian stationary stochastic process
	defined in the time variable such that%
	\begin{align}
	& \mathbb{E}\left[ B_{g}\left( w,s\right) \right] =0\ \ \ ,\nonumber \\& \mathbb{E}\left[
	B_{g}\left(w, 0\right) B_{g}\left( w,s\right) \right] =\int_{\mathbb{R}%
		^{3}}dY\int_{\mathbb{R}^{3}}dw\nabla _{Y}\Phi \left( \left\vert
	Y+(w-V)s\right\vert \right) \otimes \nabla _{Y}\Phi \left( \left\vert
	Y\right\vert \right) g\left( w\right) \ .  \label{T4E8c}
	\end{align}
	We notice that the noise $B_g$ decorrelates in times of order $\ell
	_{\varepsilon }$ as it might be expected.

	Moreover, we set $\tilde{B}^{\ep}_{g}\left(v, t\right)= \frac{1}{\theta _{\varepsilon }\sqrt{%
			\ell _{\varepsilon }}}B_{g}\left( \frac{t}{\left( \theta _{\varepsilon
		}\right) ^{2}\ell _{\varepsilon }}\right) $. Hence 
	$\tilde{B}^{\ep}_{g}$ converges to $\tilde{B}^{\ep}_{g}$ in the limit ${\ep}\to 0$, which satisfies
	\begin{align}
	\mathbb{E}\left[ \tilde{B}_{g}\left( s\right) \right] =& 0\ \ \ ,\ \ \ 
	\mathbb{E}\left[ \tilde{B}_{g}\left( s_{1}\right) \tilde{B}_{g}\left(
	s_{2}\right) \right] =D_g(V_0)\delta \left( s_{1}-s_{2}\right)   \notag \\
	D_g(V_0)=& \int_{0}^{\infty }ds\int_{\mathbb{R}^{3}}dY\int_{\mathbb{R}^{3}}dw\nabla
	_{Y}\Phi \left( \left\vert Y+(w-v)s\right\vert \right) \otimes \nabla _{Y}\Phi
	\left( \left\vert Y\right\vert \right) g\left( w\right) .  \label{A1E4a}
	\end{align}%
	This concludes the proof of Theorem \ref{thm:Grazing}. 
\end{proofof}

\newpage

\appendix

\section{Approximation of particle distributions by means of gaussian
densities.}

\label{appA}

The approximation of Poisson point processes by means of Gaussian fields in
distances much larger than the average particle distance is a well known
consequence of the Central Limit Theorem. We summarize here the main
properties that we will use of the resulting Gaussian fields for particle
distributions distributed in the phase space $\left( x,v\right).$

For each particle configuration $\left\{ \left( x_{k},v_{k}\right) \right\}
_{k\in S}$ in $\R^3\times \R^3$ chosen according to the
rate $dxg\left( dv\right) $ we define $y_{k}=\frac{x_{k}}{L_{\eps }}%
,\ W_{k}=v_{k},\ y=\frac{x}{L_{\eps }},\ w=v,$ with $L_{\eps
}\gg 1$ and the empirical distribution
\begin{equation*}
f_{\eps }\left( y,w\right) =\frac{1}{\left( L_{\eps }\right)
^{3}}\sum_{k}\delta \left( y-y_{k}\right) \delta \left( w-W_{k}\right) \ \
,\ \ f_{\eps }\in \mathcal{M}_{+}\left( \R^3\times \mathbb{R%
}^{3}\right) \ .
\end{equation*} 
Then, for any compactly supported test function $\varphi \in C_{0}\left( 
\R^3\times \R^3\right) $ we have:%
\begin{equation*}
\int_{\R^3\times \R^3}f_{\eps }\left( y,w\right)
\varphi \left( y,w\right) dydw=\frac{1}{\left( L_{\eps }\right) ^{3}}%
\sum_{k}\varphi \left( y_{k},W_{k}\right)\ .
\end{equation*}
Therefore, taking the expectation with respect to the Poisson measure we obtain:%
\begin{equation*}
\mathbb{E}\left[ \int_{\R^3\times \R^3}f_{\eps
}\left( y,w\right) \varphi \left( y,w\right) dydw\right] =\frac{1}{\left(
L_{\eps }\right) ^{3}}\mathbb{E}\left[ \sum_{k}\varphi \left(
y_{k},W_{k}\right) \right] \ .
\end{equation*}
The right-hand side of this identity can be computed approximating $\varphi $
by piecewise constant functions and using the properties of the Poisson
distribution. Then:%
\begin{equation*}
\mathbb{E}\left[ \int_{\R^3\times \R^3}f_{\eps
}\left( y,w\right) \varphi \left( y,w\right) dydw\right] =\int_{\mathbb{R}%
^{3}\times \R^3}\varphi \left( y,w\right) g\left( w\right) dydw\ .
\end{equation*}

We estimate now the variance of $\int_{\R^3\times \mathbb{R}%
^{3}}f_{\eps }\varphi dydw.$ We have:%
\begin{align*}
& \int_{\R^3\times \R^3}\int_{\R^3\times 
\R^3}f_{\eps }\left( y_{a},w_{a},0\right) f_{\eps
}\left( y_{b},w_{b},0\right) \varphi \left( y_{a},w_{a}\right) \varphi
\left( y_{b},w_{b}\right) dy_{a}dw_{a}dy_{b}dw_{b} \\
& =\frac{1}{\left( L_{\eps }\right) ^{6}}\sum_{k}\sum_{\ell }\varphi
\left( y_{k},W_{k}\right) \varphi \left( y_{\ell },W_{\ell }\right)\ .
\end{align*}

In order to compute this integral we approximate $\varphi $ by means of
piecewise functions. Therefore we need to compute the following expectations
(where we denote as $\chi _{A}$ the characteristic function of the set $%
A\subset \R^3\times \R^3$):%
\begin{equation*}
\mathbb{E}\left[ \frac{1}{\left( L_{\eps }\right) ^{6}}%
\sum_{k}\sum_{\ell }\chi _{A}\left( y_{k},W_{k}\right) \chi _{A}\left(
y_{\ell },W_{\ell }\right) \right] \ \ ,\ \ \mathbb{E}\left[ \frac{1}{\left(
L_{\eps }\right) ^{6}}\sum_{k}\sum_{\ell }\chi _{A}\left(
y_{k},W_{k}\right) \chi _{B}\left( y_{\ell },W_{\ell }\right) \right]
\end{equation*}%
where $A$ and $B$ are disjoint sets. We have:%
\begin{align*}
& \mathbb{E}\left[ \sum_{k,\ell
}\chi _{A}\left( y_{k},W_{k}\right) \chi _{A}\left( y_{\ell },W_{\ell
}\right) \right]  =\mathbb{E}\left[\sum_{
\,k\neq \ell}\chi _{A}\left( y_{k},W_{k}\right) \chi _{A}\left( y_{\ell
},W_{\ell }\right) \right] +\mathbb{E}\left[ \sum_{k}\chi _{A}\left( y_{k},W_{k}\right) \chi _{A}\left(
y_{k},W_{k}\right) \right] \\
& =\left( \mathbb{E}\left[
\left( n\left( A\right) \right) ^{2}\right] -\mathbb{E}\left[ n\left(
A\right) \right] \right) +\mathbb{E}\left[ n\left( A\right)\right] =\mathbb{E}\left[ \left( n\left( A\right) \right) ^{2}\right]\ .
\end{align*}
On the other hand we have:
\begin{equation*}
\mathbb{E}\left[ \frac{1}{\left( L_{\eps }\right) ^{6}}%
\sum_{k}\sum_{\ell }\chi _{A}\left( y_{k},W_{k}\right) \chi _{B}\left(
y_{\ell },W_{\ell }\right) \right] =\mathbb{E}\left[ \frac{n\left( A\right)
n\left( B\right) }{\left( L_{\eps }\right) ^{6}}\right] =\frac{1}{%
\left( L_{\eps }\right) ^{6}}\mathbb{E}\left[ n\left( A\right) \right]
\mathbb{E}\left[ n\left( B\right) \right]\ .
\end{equation*}
We write  
$c_{A}=\left( L_{\eps}\right) ^{3}\int_{A}g\left( v\right) dydv=\left(
L_{\eps}\right) ^{3}J_{A}$
whence 
\begin{equation*}
\mathbb{E}\left[ \frac{1}{\left( L_{\eps}\right) ^{6}}%
\sum_{k}\sum_{\ell}\chi_{A}\left( y_{k},W_{k}\right) \chi_{B}\left( y_{\ell
},W_{\ell}\right) \right] =\frac{1}{\left( L_{\eps}\right) ^{6}}%
\mathbb{E}\left[ n\left( A\right) \right] \mathbb{E}\left[ n\left( B\right) %
\right] =J_{A}J_{B}\ .
\end{equation*}
We observe that 
\begin{equation*}
\mathbb{E}\left[ \left( n\left( A\right) \right) ^{2}\right] =\sum
_{N=0}^{\infty}N^{2}\frac{\left( c_{A}\right) ^{N}}{N!}e^{-c_{A}}=\left(
c_{A}\right) ^{2}+c_{A}=\left( L_{\eps}\right) ^{6}\left(
J_{A}\right) ^{2}+\left( L_{\eps}\right) ^{3}J_{A},
\end{equation*}
where we used that  $\sum_{N=0}^{\infty}N^{2}\frac{\left( x\right) ^{N}}{N!}=e^{x}\left(
x^{2}+x\right)$. Then: 
\begin{equation*}
\mathbb{E}\left[ \frac{1}{\left( L_{\eps }\right) ^{6}}%
\sum_{k}\sum_{\ell }\chi _{A}\left( y_{k},W_{k}\right) \chi _{A}\left(
y_{\ell },W_{\ell }\right) \right] =\frac{1}{\left( L_{\eps }\right)
^{6}}\mathbb{E}\left[ \left( n\left( A\right) \right) ^{2}\right] =\left(
J_{A}\right) ^{2}+\frac{1}{\left( L_{\eps }\right) ^{3}}J_{A}\ .
\end{equation*}

By decomposing $\varphi $ in linear combinations of characteristic functions
of disjoint sets we get: 
\begin{align*}
\frac{1}{\left( L_{\eps }\right) ^{6}}\mathbb{E}\left[
\sum_{k,l}\varphi \left( y_{k},W_{k}\right) \varphi \left( y_{\ell },W_{\ell
}\right) \right] & =\int_{(\mathbb{R}^3)^2}\varphi \left( y,w\right) g\left(
v\right) dydv+\frac{1}{\left( L_{\eps }\right) ^{3}}\int_{(\mathbb{R}%
^3)^2}\left( \varphi \left( y,w\right) \right) ^{2}g\left( v\right) dydv
\end{align*}%
whence:%
\begin{align*}
& \mathbb{E}\left[ \int_{\R^3\times \R^3}\int_{\mathbb{R}%
^{3}\times \R^3}f_{\eps }\left( y_{a},w_{a},0\right)
f_{\eps }\left( y_{b},w_{b},0\right) \varphi \left(
y_{a},w_{a}\right) \varphi \left( y_{b},w_{b}\right) dy_{a}dw_{a}dy_{b}dw_{b}%
\right] \\
& =\frac{1}{\left( L_{\eps }\right) ^{3}}\int_{\R^3\times 
\R^3}\left( \varphi \left( y,w\right) \right) ^{2}g\left( v\right)
dydv
\end{align*}%
or in a formal manner:%
\begin{equation*}
\mathbb{E}\left[ f_{\eps }\left( y_{a},w_{a},0\right) f_{\eps
}\left( y_{b},w_{b},0\right) \right] =\frac{g\left( w_{a}\right) }{\left(
L_{\eps }\right) ^{3}}\delta \left( y_{a}-y_{b}\right) \delta \left(
w_{a}-w_{b}\right)\ .
\end{equation*}
We now compute the correlations associated to the density $\tilde{\rho}_{1},$ namely
\begin{align*}
\tilde{\rho}_{1}\left( y,\tilde{t}\right) & =\int_{\R^3}\zeta
_{1}\left( y,w,\tilde{t}\right) dw 
\end{align*}
where $\zeta_{1}\left( y,w,\tilde{t}\right)  =N\left( y-w\tilde{t},w\right)$ with 
$\mathbb{E}\left[ \left( N\left( y_{a},w_{a}\right) \right) N\left(
y_{b},w_{b}\right) \right] =g\left( w_{a}\right) \delta\left(
y_{a}-y_{b}\right) \delta\left( w_{a}-w_{b}\right)\ .$
We have 
\begin{align*}
\mathbb{E}\left[ \tilde{\rho}_{1}\left( y_{1},\tilde{t}_{1}\right) \tilde{%
\rho}_{1}\left( y_{2},\tilde{t}_{2}\right) \right] & =\int_{\mathbb{R}%
^{3}}dw_{1}\int_{\R^3}dw_{2}\mathbb{E}\left[ \zeta_{1}\left(
y_{1},w_{1},\tilde{t}_{1}\right) \zeta_{1}\left( y_{2},w_{2},\tilde{t}%
_{2}\right) \right] \\
& =\int_{\R^3}dw_{1}\int_{\R^3}dw_{2}\mathbb{E}\left[
N\left( y_{1}-w_{1}\tilde{t}_{1},w_{1}\right) N\left( y_{2}-w_{2}\tilde {t}%
_{2},w_{2}\right) \right] \\
& =\int_{\R^3}dw_{1}g\left( w_{1}\right) \delta\left(
y_{1}-y_{2}-w_{1}\left( \tilde{t}_{1}-\tilde{t}_{2}\right) \right)\ .
\end{align*}
Suppose now that $\tilde{t}_{1}\geq \tilde{t}_{2}.$ Then:%
\begin{align*}
& \mathbb{E}\left[ \tilde{\rho}_{1}\left( y_{1},\tilde{t}_{1}\right) \tilde{%
\rho}_{1}\left( y_{2},\tilde{t}_{2}\right) \right] =\int_{\mathbb{R}%
^{3}}dw_{1}g\left( w_{1}\right) \delta \left( y_{1}-y_{2}-w_{1}\left( \tilde{%
t}_{1}-\tilde{t}_{2}\right) \right) \\
& =\frac{1}{\left( \tilde{t}_{1}-\tilde{t}_{2}\right) ^{3}}\int_{\mathbb{R}%
^{3}}dw_{1}g\left( w_{1}\right) \delta \left( \frac{y_{1}-y_{2}}{\tilde{t}%
_{1}-\tilde{t}_{2}}-w_{1}\right) =\frac{1}{\left( \tilde{t}_{1}-\tilde{t}%
_{2}\right) ^{3}}g\left( \frac{y_{1}-y_{2}}{\tilde{t}_{1}-\tilde{t}_{2}}%
\right)\ .
\end{align*}
Henceforth:%
\begin{equation*}
\mathbb{E}\left[ \tilde{\rho}_{1}\left( y_{1},\tilde{t}_{1}\right) \tilde{%
\rho}_{1}\left( y_{2},\tilde{t}_{2}\right) \right] =\frac{1}{\left( \tilde{t}%
_{1}-\tilde{t}_{2}\right) ^{3}}g\left( \frac{y_{1}-y_{2}}{\tilde{t}_{1}-%
\tilde{t}_{2}}\right) \ \ ,\ \ \ \tilde{t}_{1}>\tilde{t}_{2}.
\end{equation*}

\bigskip

\section{Stability properties of a Vlasov medium and the Dielectric function}\label{app:B} 
	
	In this appendix we provide a physical interpretation to both terms $H_{g}\left( 
		\tilde{t};V_0\right) $ and $F_{g}\left( \xi ,\tilde{t}\right) $ as follows.
		The term $H_{g}\left( \tilde{t};V_0\right) $ is due to the fact that the
		presence of the tagged particle induces a force in the surrounding
		distribution of scatterers. These scatterers rearrange their positions as a
		consequence of their mutual interactions and the forces induced by the
		tagged particle and this results in a reaction force acting on the tagged
		particle. On the other hand the fluctuations of the particle density yield a
		random force field. These fluctuations of the particle density are
		rearranged due to the effect of the mutual interactions between the
		particles. The resulting force field after these fluctuations reach a steady
		state $F_{g}\left( \xi ,\tilde{t}\right).$ Notice that:%
		\begin{equation}
		F_{g}\left( \xi ,\tilde{t}\right) =\lim_{T\rightarrow \infty }\left[ \int_{%
			\R^3}\int_{\R^3}\nabla _{y}\Phi \left( y\right) \int_{%
			\R^3}d\eta \int_{\R^3}dw_{0}G_{\sigma }\left( y-\eta
		,w,w_{0},\tilde{t}+T\right) N\left( \eta +\xi ,w_{0}\right) dydw\right].
		\label{T2E9}
		\end{equation}
		To obtain a well defined quantity by means of (\ref{T2E8}) and a
		noise with integrable time correlations by means of (\ref{T2E9}) we need to
		make some assumptions on $g\left( w\right) $ and the interaction potential $%
		\Phi $ in order to have suitable decay properties for $G_{\sigma }\left(
		y,w,w_{0},\tilde{t}\right) $ as $\tilde{t}\rightarrow \infty .$ The
		conditions that must be imposed on $g$ in order to obtain such a decay will
		be discussed below. We just remark here that this decay
		condition for $G_{\sigma }$ means, from the physical point of view, that the
		``medium" formed by the set of particles is stable under its own
		interactions. We will just assume in the following that the decay of $%
		G_{\sigma }$ in $\tilde{t}$ is sufficiently fast to ensure the convergence
		of all the integrals appearing in what follows converge.
	
	As indicated above, in order to obtain a well defined noise term $%
	F_{g}\left( \xi ,\tilde{t}\right) $ and friction coefficient $H_{g}\left( 
	\tilde{t};V_0\right) $ the function $G_{\sigma }$ that solves (\ref{T2E2})-(%
	\ref{T2E4}) has to decay sufficiently fast as $\tilde{t}\rightarrow \infty .$
	This decay is closely related to the stability properties of the system of
	particles described by equations with the form (\ref{eq:IntNew}) with interaction
	potentials as in (\ref{S4E6}), (\ref{S4E7}). For these potentials makes
	sense to approximate (\ref{eq:IntNew}) by means of the Vlasov equation. The
	stability of homogeneous distributions of particles with a distribution of
	velocities given by $g\left( w\right) $ which can be approximated using (\ref%
	{S9E5}) was first considered in \cite{La, LL2} where the linearized Vlasov
	equation around a homogeneous distribution of particles was considered. It
	was found in that paper the possibility of damping of perturbations in a
	homogeneous equation in spite of the fact that the Vlasov equation does not
	include the effect of time irreversible effects like particle collisions.
	
	A general condition on the distribution of velocities $g\left( w\right) $
	yielding stability of a homogeneous medium was given in \cite{Pe}. A
	rigorous proof of stability of the homogeneous state for a large class of
	velocity distributions under the nonlinear Vlasov equation with Coulombian
	potentials in the torus has been obtained in \cite{MV}. In this paper we
	restrict ourselves to the analysis of the linearized problem (\ref{T2E2})-(%
	\ref{T2E4}). We discuss conditions on the potentials $\Phi $ and the
	velocity distributions $g$ yielding a behaviour on $G_{\sigma }$ for large $%
	\tilde{t}$ which allows to define the friction coefficient $\Lambda_{g}$ and the random force field $F_{g}$ by means of (\ref{T2E8}) and (\ref%
	{eq:FFCorr}) respectively.
	
	From the physical point of view the decay of $G_{\sigma }$ for long times
	means that the homogeneous distribution of scatterers is stable under the
	combined effect of their dispersion of velocities and their long range
	mutual interactions. We notice in particular that the interaction between
	the scatterers does not yield an exponential growth of the density
	inhomogeneities in the phase space.

\paragraph{Dielectric function. \label{DielFunc}}
It is customary in the physical literature to assume that the set of
particles described by means of Vlasov equations, can be interpreted as an
effective medium in which the particle density rearranges due to the action
of an external field. The effective properties of the medium are usually
described by means of the so-called dielectric function.

In order to define the dielectric function we consider the following
generalization of the Vlasov equation (cf.~\eqref{S9E5}) in which we replace
the force due to a tagged particle in the Vlasov medium by an arbitrary
force term $F_{ext}\left( y,w\right) :$%
\begin{equation*}
\partial _{\tilde{t}}f\left( y,w,\tilde{t}\right) +w\cdot \nabla _{y}f\left(
y,w,\tilde{t}\right) +\left[ F_{ext}\left( y,w\right) -\sigma \int_{\mathbb{R%
	}^{3}}\nabla _{y}\Phi \left( y-\eta \right) \rho \left( \eta ,\tilde{t}%
\right) d\eta \right] \cdot \nabla _{w}f\left( y,w,\tilde{t}\right) =0
\end{equation*}%
where $\sigma =\eps \left( L_{\eps }\right) ^{3}.$ We consider
solutions of this equation close to the spatially homogeneous solution $%
g\left( w\right) .$ We define $h\left( y,w,\tilde{t}\right) =f\left( y,w,%
\tilde{t}\right) -g\left( w\right) .$ Then, assuming that $F_{ext}$ and $h$
are small and linearizing, we obtain:%
\begin{equation}
\partial _{\tilde{t}}h\left( y,w,\tilde{t}\right) +w\cdot \nabla _{y}h\left(
y,w,\tilde{t}\right) +\left[ F_{ext}\left( y,\tilde{t}\right) -\sigma \int_{%
	\R^3}\nabla _{y}\Phi \left( y-\eta \right) \tilde{\rho}\left( \eta
,\tilde{t}\right) d\eta \right] \cdot \nabla _{w}g\left( w\right) =0
\label{A3E3}
\end{equation}%
with:%
\begin{equation}
\tilde{\rho}\left( \eta ,\tilde{t}\right) =\int_{\R^3}h\left( \eta
,w,\tilde{t}\right) dw \ . \label{A3E4}
\end{equation}

Notice that the total force exerted by the combination of external forces
and  forces due to the particles of the Vlasov medium is:%
\begin{equation}
F\left( y,\tilde{t}\right) =F_{ext}\left( y,\tilde{t}\right) -\sigma \int_{%
	\R^3}\nabla _{y}\Phi \left( y-\eta \right) \tilde{\rho}\left( \eta
,\tilde{t}\right) d\eta \ .  \label{A3E5}
\end{equation}

In order to define the dielectric function we consider external forces with
the form $F_{ext}\left( y,\tilde{t}\right) =e^{i\left( \omega \tilde{t}%
	+k\cdot y\right) }F_{0}$ with $\omega \in \mathbb{R}$, $k\in \mathbb{R}%
^{3},\ F_{0}\in \R^3.$ In that case the force $F\left( y,\tilde{t}%
\right) $ is proportional to $e^{i\left( \omega \tilde{t}+k\cdot y\right) }.$
We define the dielectric function $\eps =\eps \left( k,\omega
\right) \in C\left( \R^3\times \mathbb{R};M_{3\times 3}\left( 
\mathbb{C}\right) \right),$ where we denote as $M_{3\times 3}\left( \mathbb{%
	C}\right)$ the set of $3\times 3$ matrices with complex coefficients, by
means of
\begin{equation}
F\left( y,\tilde{t}\right) =\left[ \eps \left( k,\omega \right) F_{0}%
\right] e^{i\left( \omega \tilde{t}+k\cdot y\right) } \ . \label{A3E6}
\end{equation}

In order to compute the function $h$ for $F_{ext}\left( y,\tilde{t}\right)
=e^{i\left( \omega \tilde{t}+k\cdot y\right) }F_{0}$ we look for solutions
of (\ref{A3E3}), (\ref{A3E4}) with the form $h\left( x,w,t\right)
=e^{i\left( \omega \tilde{t}+k\cdot y\right) }H\left( w\right) .$ Then $%
H\left( w\right) $ satisfies:%
\begin{equation}
i\left( \omega \tilde{t}+k\cdot w\right) H\left( w\right) -\left( 2\pi
\right) ^{\frac{3}{2}}\sigma i\left[ \nabla _{w}g\left( w\right) \cdot k%
\right] \hat{\Phi}\left( k\right) \tilde{\rho}_{0}=-F_{0}\cdot \nabla
_{w}g\left( w\right)  \label{A3E7}
\end{equation}%
where
\begin{equation*}
\tilde{\rho}_{0}=\int_{\R^3}H\left( w\right) dw, \qquad
\int_{\R^3}\nabla _{y}\Phi \left( y-\eta \right) e^{-ik\cdot
	\left( y-\eta \right) }d\eta =\left( 2\pi \right) ^{\frac{3}{2}}ik\hat{\Phi}%
\left( k\right) \ .
\end{equation*}
Therefore
\begin{equation*}
H\left( w\right) =\frac{\left( 2\pi \right) ^{\frac{3}{2}}\sigma \left[
	\nabla _{w}g\left( w\right) \cdot k\right] \hat{\Phi}\left( k\right) \tilde{%
		\rho}_{0}}{\left( \omega \tilde{t}+k\cdot w+i 0^{+}\right) }-\frac{F_{0}\cdot
	\nabla _{w}g\left( w\right) }{i\left( \omega \tilde{t}+k\cdot w+i0^{+}\right) } \ .
\end{equation*}
We now integrate with respect to the $w$ variable and we obtain
\begin{equation*}
\tilde{\rho}_{0}-\left( 2\pi \right) ^{\frac{3}{2}}\sigma \tilde{\rho}_{0}%
\hat{\Phi}\left( k\right) \int_{\R^3}\frac{\nabla _{w}g\left(
	w\right) \cdot k}{\left( \omega \tilde{t}+k\cdot w+i0^{+}\right) }dw=-\int_{%
	\R^3}\frac{F_{0}\cdot \nabla _{w}g\left( w\right) }{i\left( \omega 
	\tilde{t}+k\cdot w+i0^{+}\right) }dw
\end{equation*}%
whence%
\begin{equation*}
\tilde{\rho}_{0}=-\frac{1}{\Delta _{\sigma }\left( k,i\omega \right) }\int_{%
	\R^3}\frac{F_{0}\cdot \nabla _{w}g\left( w\right) }{i\left( \omega 
	\tilde{t}+k\cdot w+i0^{+}\right) }dw \ .
\end{equation*}
Then, using that $F_{ext}\left( y,\tilde{t}\right) =e^{i\left( \omega \tilde{%
		t}+k\cdot y\right) }F_{0}$ it follows that 
\begin{eqnarray*}
	F\left( y,\tilde{t}\right) &=&F_{ext}\left( y,\tilde{t}\right) -\sigma \int_{%
		\R^3}\nabla _{y}\Phi \left( y-\eta \right) \tilde{\rho}\left( \eta
	,\tilde{t}\right) d\eta  \\
	&=&e^{i\left( \omega \tilde{t}+k\cdot y\right) }\left( F_{0}+\frac{\left(
		2\pi \right) ^{\frac{3}{2}}\sigma }{\Delta _{\sigma }\left( k,i\omega
		\right) }\hat{\Phi}\left( k\right) k\int_{\R^3}\frac{\nabla
		_{w}g\left( w\right) \cdot F_{0}}{\left( \omega \tilde{t}+k\cdot
		w+i0^{+}\right) }dw\right)
\end{eqnarray*}%
whence
\begin{equation*}
\eps \left( k,\omega \right) =I+\frac{\left( 2\pi \right) ^{\frac{3}{2}%
	}\sigma }{\Delta _{\sigma }\left( k,i\omega \right) }\hat{\Phi}\left(
k\right) \int_{\R^3}\frac{k\otimes \nabla _{w}g\left( w\right) }{%
	\left( \omega \tilde{t}+k\cdot w+i0^{+}\right) }dw \ .
\end{equation*}

\bigskip

\bigskip 

\noindent\textbf{Acknowledgment. }The authors acknowledge support through the CRC
1060 \textit{The mathematics of emergent effects }at the University of Bonn
that is funded through the German Science Foundation (DFG), as well as the
support of the Hausdorff Research Institute for Mathematics (Bonn), through the Junior
Trimester Program on Kinetic Theory. R.W. acknowledges support  of Universit\'e de Lyon through the IDEXLYON Scientific Breakthrough Project ``Particles drifting and propelling in turbulent flows", 
and the hospitality of the UMPA ENS Lyon.

\bigskip


\bigskip

\end{document}